\documentclass[reqno,11pt]{amsart}
\usepackage[margin=1.2in]{geometry}
\usepackage{amssymb,amsfonts,amsmath,amsthm}
\usepackage{amsaddr}
\usepackage{siunitx}
\usepackage{graphicx,grffile}
\usepackage[longnamesfirst]{natbib} 
\usepackage{hyperref}
\usepackage{hypernat}
\usepackage{graphicx,grffile,setspace}
\usepackage{epstopdf}
\usepackage{endnotes}
\DeclareGraphicsExtensions{.pdf,.png} 

\let\footnote=\endnote

\newtheorem{theorem}{Theorem}

\newtheorem{corollary}[theorem]{Corollary}

\newtheorem{lemma}[theorem]{Lemma}

\newtheorem{proposition}[theorem]{Proposition}

\theoremstyle{definition}

\newtheorem{assumption}{Assumption}
\theoremstyle{remark}

\def \independenT#1#2{\mathrel{\rlap{$#1#2$}\mkern2mu{#1#2}}}
\newcommand{\D}[2]{\frac{\partial #1}{\partial #2}}

\def \E{\mathbb E}
\let \originalleft \left
\let \originalright \right
\renewcommand{\left}{\mathopen{}\mathclose \bgroup \originalleft}
\renewcommand{\right}{\aftergroup \egroup \originalright}

\begin{document}
\title[Smoothed Estimating Equations for {IV-QR}]{Smoothed Estimating Equations \\ for Instrumental Variables Quantile Regression}




\thanks{%
Thanks to Victor Chernozhukov (co-editor) and an anonymous referee for insightful comments and references, and thanks to Peter C.\ B.\ Phillips (editor) for additional editorial help.  Thanks to Xiaohong Chen, Brendan
Beare, Andres Santos, and active seminar and conference participants for insightful
questions and comments.} 

\maketitle

\begin{center}
{\sc David M.\ Kaplan}

\textit{Department of Economics, University of Missouri}

\textit{118 Professional Bldg, 909 University Ave, Columbia, MO 65211-6040}

E-mail: \texttt{kaplandm@missouri.edu}

~\\

{\sc Yixiao Sun}

\textit{Department of Economics, University of California, San Diego}

E-mail: \texttt{yisun@ucsd.edu}
\end{center}






\paragraph{\sc Abstract}

The moment conditions or estimating equations for instrumental variables
quantile regression involve the discontinuous indicator function. We instead
use smoothed estimating equations (SEE), with bandwidth $h$. We show that
the mean squared error (MSE) of the vector of the SEE is minimized for some $h>0$,
leading 
to smaller asymptotic 
MSE 
of the estimating
equations and associated parameter estimators. The same MSE-optimal $h$ also
minimizes the higher-order type I error of a SEE-based $\chi ^{2}$ test 
and increases 
size-adjusted power in large samples. Computation of the SEE
estimator also becomes simpler and more reliable, especially with (more)
endogenous regressors. Monte Carlo simulations demonstrate all of these
superior properties in finite samples, and we apply our estimator to JTPA data. 
Smoothing the estimating equations is
not just a technical operation for establishing Edgeworth expansions and
bootstrap refinements; it also brings the real benefits of having more
precise estimators and more powerful tests.





\allowdisplaybreaks[4]


\section{Introduction}

Many econometric models are specified by moment conditions or estimating
equations. An advantage of this approach is that the full distribution of
the data does not have to be parameterized. In this paper, we consider
estimating equations that are not smooth in the parameter of interest. We
focus on instrumental variables quantile regression (IV-QR), which
includes the usual quantile regression as a special case. Instead of using
the estimating equations that involve the nonsmooth indicator function, we
propose to smooth the indicator function, leading to our smoothed estimating
equations (SEE) and SEE estimator.

Our SEE estimator has several advantages. First, from a computational point
of view, the SEE estimator can be computed using any standard iterative
algorithm that requires smoothness. This is especially attractive in IV-QR
where simplex methods for the usual QR are not applicable. In fact,
the SEE approach has been used in %
\citet{ChenPouzo2009,ChenPouzo2012} for computing their nonparametric sieve
estimators in the presence of nonsmooth moments or generalized residuals.
However, a rigorous investigation is currently lacking. Our paper can be
regarded as a first step towards justifying the SEE approach in
nonparametric settings. 
Relatedly, \citet[][\S7.1]{FanLiao2014} have employed
the same strategy of smoothing the indicator function to
reduce the computational burden of their focused GMM approach. 
Second, from a technical point of view, smoothing
the estimating equations enables us to establish high-order properties of
the estimator. This motivates \citet{Horowitz1998}, for instance, to examine
a smoothed objective function for median regression, to show high-order
bootstrap refinement. Instead of smoothing the objective function, we show
that there is an advantage to smoothing the estimating equations. This point
has not been recognized and emphasized in the literature. For QR estimation
and inference via empirical likelihood, \citet{Otsu2008} and %
\citet{Whang2006} also examine smoothed estimators. To the best of our
knowledge, no paper has examined smoothing the estimating equations for the
usual QR estimator, let alone IV-QR. Third, from a statistical point of
view, the SEE estimator is a flexible class of estimators that includes the
IV/OLS mean regression estimators and median and quantile regression
estimators as special cases. Depending on the smoothing parameter, the SEE
estimator can have different degrees of robustness in the sense of %
\citet{Huber1964}. By selecting the smoothing parameter appropriately, we
can harness the advantages of both the mean regression estimator and the
median/quantile regression estimator. Fourth and most importantly, from an
econometric point of view, smoothing can reduce the mean squared error (MSE)
of the SEE, which in turn leads to a smaller asymptotic MSE of the parameter
estimator and to more powerful tests. We seem to be the first to establish
these advantages.

In addition to investigating the asymptotic properties of the SEE estimator,
we provide a smoothing parameter choice that minimizes different criteria:
the MSE of the SEE, the type I error of a chi-square test subject to exact
asymptotic size control, and the approximate MSE of the parameter estimator. We show
that the first two criteria produce the same optimal smoothing parameter,
which is also optimal under a variant of the third criterion. With the
data-driven smoothing parameter choice, we show that the statistical and
econometric advantages of the SEE estimator are reflected clearly in our
simulation results.

There is a growing literature on IV-QR.  For a recent review, see \citet{ChernozhukovHansen2013}.  Our paper is built upon \citet{ChernozhukovHansen2005}, which establishes a structural framework for IV-QR and provides primitive identification conditions.  Within this framework, \citet{ChernozhukovHansen2006} and \citet{ChernozhukovEtAl2009} develop estimation and inference procedures under strong identification.  For inference procedures that are robust to weak identification, see \citet{ChernozhukovHansen2008} and \citet{Jun2008}, for example. IV-QR can also reduce bias for dynamic panel fixed effects estimation as in \citet{Galvao2011}.  
None of these papers considers smoothing the IV-QR estimating equations; that idea (along with minimal first-order theory) seemingly first appeared in an unpublished draft by \citet{MaCurdyHong1999}, although the idea of smoothing the indicator function in general appears even earlier, as in \citet{Horowitz1992} for the smoothed maximum score estimator. 
An alternative approach to overcome the computational obstacles in the presence of a nonsmooth objective function is to explore the asymptotic equivalence of the Bayesian and classical methods for regular models and use the MCMC approach to obtain the classical extremum estimator; see \citet{ChernozhukovHong2003}, whose Example 3 is IV-QR.  As a complement, our approach deals with the computation problem in the classical framework directly. 

The rest of the paper is organized as follows. Section \ref{sec:see}
describes our setup and discusses some illuminating connections with other
estimators. Sections \ref{sec:mse}, \ref{sec:eI}, and \ref{sec:est}
calculate the MSE of the SEE, the type I and type II errors of a chi-square
test, and the approximate MSE of the parameter estimator, respectively.
Section \ref{sec:emp} applies our estimator to JTPA data, and Section \ref{sec:sim} presents simulation results before we conclude. Longer
proofs and calculations are gathered in the appendix.

\section{Smoothed Estimating Equations\label{sec:see}}

\subsection{Setup}

We are interested in estimating the instrumental variables quantile
regression (IV-QR) model 
\begin{equation*}
Y_{j}=X_{j}^{\prime }\beta _{0}+U_{j}
\end{equation*}%
where $\E\left[Z_{j}\left( 1\{U_{j}<0\}-q\right)\right] =0$ for instrument vector $%
Z_{j}\in \mathbb{R}^{d}$ and $1\{ \cdot \}$ is the indicator function.
Instruments are taken as given; this does not preclude first determining the
efficient set of instruments as in \citet{Newey2004} or %
\citet{NeweyPowell1990}, for example. 
We restrict attention to the ``just identified'' case $X_{j}\in \mathbb{R}%
^{d}$ and iid data for simpler exposition; for the overidentified case, see %
\eqref{eqn:EE-overID} below.

A special case of this model is exogenous QR with $Z_{j}=X_{j}$, which is
typically estimated by minimizing a criterion function: 
\begin{equation*}
\hat{\beta}_{Q}\equiv \mathop{\rm arg\,min}_{\beta }\frac{1}{n}%
\sum_{j=1}^{n}\rho _{q}(Y_{j}-X_{j}^{\prime }\beta ),
\end{equation*}%
where $\rho _{q}(u)\equiv \left( q-1\{u<0\}\right) u$ is the check function.
Since the objective function is not smooth, it is not easy to obtain a
high-order approximation to the sampling distribution of $\hat{\beta}_{Q}$.
To avoid this technical difficulty, \citet{Horowitz1998} proposes to smooth
the objective function to obtain 
\begin{equation*}
\hat{\beta}_{H}=\mathop{\rm arg\,min}_{\beta }\frac{1}{n}\sum_{j=1}^{n}\rho
_{q}^{H}(Y_{j}-X_{j}^{\prime }\beta ),\quad \rho _{q}^{H}(u)\equiv \left[
q-G\left( -u/h\right) \right] u,
\end{equation*}%
where $G(\cdot )$ is a smooth function and $h$ is the smoothing parameter or
bandwidth. Instead of smoothing the objective function, we smooth the
underlying moment condition and define $\hat{\beta}$ to be the solution of
the vector of smoothed estimating equations (SEE) $m_{n}(\hat{\beta})=0$,
where\footnote{%
It suffices to have $m_{n}(\hat{\beta})=o_{p}(1)$, which allows for a small
error when $\hat{\beta}$ is not the exact solution to $m_{n}(\hat{\beta})=0$.%
} 
\begin{equation*}
m_{n}(\beta )\equiv \frac{1}{\sqrt{n}}\sum_{j=1}^{n}W_{j}(\beta )\text{ and }%
W_{j}(\beta )\equiv Z_{j}\left[ G\left( \frac{X_{j}^{\prime }\beta -Y_{j}}{h}%
\right) -q\right] .
\end{equation*}

Our approach is related to kernel-based nonparametric conditional quantile
estimators. The moment condition there is $\E\left[ 1\{X=x\} \left(
1\{Y<\beta \}-q\right) \right] =0$. Usually the $1\{X=x\}$ indicator
function is ``smoothed'' with a kernel, while the latter term is not. This
yields the nonparametric conditional quantile estimator $\hat{\beta}_{q}(x)=%
\mathop{\rm arg\,min}_{b}\sum_{i=1}^{n}\rho _{q}(Y_{i}-b)K[(x-X_{i})/h]$ for
the conditional $q$-quantile at $X=x$, estimated with kernel $K(\cdot)$ and
bandwidth $h$. Our approach is different in that we smooth the indicator $%
1\{Y<\beta \}$ rather than $1\{X=x\}$. Smoothing both terms may help but is
beyond the scope of this paper.

Estimating $\hat{\beta}$ from the SEE is computationally easy: $d$ equations
for $d$ parameters, and a known, analytic Jacobian. Computationally, solving
our problem is faster and more reliable than the IV-QR method in %
\citet{ChernozhukovHansen2006}, which requires specification of a grid of
endogenous coefficient values to search over, computing a conventional QR
estimator for each grid point. This advantage is important particularly when
there are more endogenous variables.

If the model is overidentified with $\dim (Z_{j})>\dim (X_{j})$, we can use
a $\dim (X_{j})\times \dim (Z_{j})$ matrix $\mathbb{W}$ to transform the
original moment conditions $\E\left[Z_{j}\left(q-1\left\{Y_{j}<X_{j}^{\prime }\beta \right\}\right)\right]=0$
into 
\begin{equation}  \label{eqn:EE-overID}
\E\left[ \tilde{Z}_{j}\left( q-1\left\{Y_{j}<X_{j}^{\prime }\beta \right\} \right)
\right] =0,\text{ for }\tilde{Z}_{j}=\mathbb{W}Z_{j}\in \mathbb{R}^{\dim
(X_{j})}.
\end{equation}%
Then we have an exactly identified model with transformed instrument vector $%
\tilde{Z}_{j}$, and our asymptotic analysis can be applied to %
\eqref{eqn:EE-overID}.

By the theory of optimal estimating equations or efficient two-step GMM, the
optimal $\mathbb{W}$ takes the following form:%
\begin{align*}
\mathbb{W} &= \left. \frac{\partial }{\partial \beta }\E\left[ Z^{\prime
}\left( q-1\left\{Y<X^{\prime }\beta \right\} \right) \right] \right \vert _{\beta
=\beta _{0}}\mathrm{Var}\left[Z\left(q-1\{Y<X^{\prime }\beta _{0}\} \right) \right]^{-1}
\\
&= \E \left[XZ^{\prime }f_{U|Z,X}(0)\right] \left\{ \E\left[ZZ^{\prime }\sigma
^{2}\left( Z\right) \right]\right\} ^{-1} ,
\end{align*}%
where $f_{U|Z,X}(0)$ is the conditional PDF of $U$ evaluated at $U=0$ given $%
\left( Z,X\right) $ and $\sigma ^{2}\left( Z\right) =\mathrm{Var}%
\left(1\left \{ U<0\right \} \mid Z\right)$. The standard two-step approach
requires an initial estimator of $\beta _{0}$ and nonparametric estimators
of $f_{U|Z,X}(0)$ and $\sigma ^{2}\left( Z\right)$. The underlying
nonparametric estimation error may outweigh the benefit of having an optimal
weighting matrix. This is especially a concern when the dimensions of $X$
and $Z$ are large.  The problem is similar to what \citet{HwangSun2015} consider in a time series GMM framework where the optimal weighting matrix
is estimated using a nonparametric HAC approach.  Under the alternative and
more accurate asymptotics that captures the estimation error of the
weighting matrix, they show that the conventionally optimal two-step
approach does not necessarily outperform a first-step approach that does not
employ a nonparametric weighting matrix estimator. While we expect a similar
qualitative message here, we leave a rigorous analysis to future research. 

In practice, a simple procedure is to ignore $%
f_{U|Z,X}(0) $ and $\sigma ^{2}\left( Z\right) $ (or assume that they are
constants) and employ the following empirical weighting matrix, 
\begin{equation*}
\mathbb{W}_{n} = \left( \frac{1}{n}\sum_{j=1}^{n}X_{j}Z_{j}^{\prime }\right) %
\left( \frac{1}{n}\sum_{j=1}^{n}Z_{j}Z_{j}^{\prime }\right) ^{-1}.
\end{equation*}%
This choice of $\mathbb{W}_{n}$ is in the spirit of the influential work of %
\citet{LiangZeger1986} who advocate the use of a working correlation matrix
in constructing the weighting matrix. Given the above choice of $\mathbb{W}%
_{n}$, $\tilde{Z}_{j}$ is the least squares projection of $X_{j}$ on $Z_{j}$%
. It is easy to show that with some notational changes our asymptotic
results remain valid in this case. 

An example of an overidentified model is the conditional moment
model 
\begin{equation*}
\E\left[ \left( 1\{U_{j}<0\}-q\right) \mid Z_{j}\right] =0.
\end{equation*}%
In this case, any measurable function of $Z_{j}$ can be
used as an instrument. As a result, the model could be overidentified.
According to \citet{Chamberlain1987} and \citet{Newey1990}, the optimal set of instruments in our setting is given by
\begin{equation*}
\left. \left[ \frac{\partial }{\partial \beta }\E\left( 1\{Y_{j}-X_{j}^{%
\prime }\beta <0\}\mid Z_{j}\right) \right] \right \vert _{\beta =\beta _{0}}.
\end{equation*}%
Let $F_{U|Z,X}\left( u\mid z,x\right) $ and $%
f_{U|Z,X}\left( u\mid z,x\right) $ be the conditional
distribution function and density function of $U$ given $\left(
Z,X\right) =(z,x)$.  Then under some regularity conditions,
\begin{align*}
\left. \left[ \frac{\partial }{\partial \beta }\E\left( 1\{Y_{j}-X_{j}^{%
\prime }\beta <0\}\mid Z_{j}\right) \right] \right \vert _{\beta =\beta _{0}}
  &= \left. \left\{ \frac{\partial }{\partial \beta }\E\left[ \E\left(
1\{Y_{j}-X_{j}^{\prime }\beta <0\} \mid Z_{j},X_{j}\right) \mathrel{\big|} Z_{j}\right]
\right\} \right \vert _{\beta =\beta _{0}} \\
  &= \E\left\{ \left. \left[\frac{\partial }{\partial \beta }%
F_{U_{j}|Z_{j},X_{j}}\left( X_{j}\left( \beta -\beta _{0}\right) \mid
Z_{j},X_{j}\right)\right]\right\vert _{\beta =\beta_{0}} \mathrel{\bigg|} Z_{j}\right\}  \\
  &= \E\left[  f_{U_{j}|Z_{j},X_{j}}(0\mid Z_{j},X_{j})X_{j} \mathrel{\big|} Z_{j}\right] .
\end{align*}

The optimal instruments involve the conditional density 
$f_{U|Z,X}\left( u\mid z,x\right) $ and a conditional expectation. In
principle, these objects can be estimated nonparametrically. However, the
nonparametric estimation uncertainty can be very high,
adversely affecting the reliability of inference. A simple and practical
strategy\footnote{We are not alone in recommending this simple strategy for empirical work. 
\citet{ChernozhukovHansen2006} make the same recommendation in their Remark 5 and use this strategy in
their empirical application. See also \citet{Kwak2010}.} 
is to construct the optimal instruments as the OLS projection of each $%
X_{j}$ onto some sieve basis functions $\Phi ^{K}\left(
Z_{j}\right) \equiv \left[ \Phi
_{1}(Z_{j}),\ldots,\Phi _{K}(Z_{j})\right] ^{\prime }$, leading to
\begin{equation*}
\tilde{Z}_{j}=\left[ \frac{1}{n}\sum_{j=1}^{n}X_{j}\Phi ^{K}\left(
Z_{j}\right) ^{\prime }\right] \left[ \frac{1}{n}\sum_{j=1}^{n}\Phi
^{K}\left( Z_{j}\right) \Phi ^{K}\left( Z_{j}\right) ^{\prime }\right]
^{-1}\Phi ^{K}\left( Z_{j}\right) \in \mathbb{R}^{\dim (X_{j})}
\end{equation*}%
as the instruments. Here $\left \{ \Phi _{i}\left( \cdot \right)
\right \} $ are the basis functions such as power functions. Since
the dimension of $\tilde{Z}_{j}$ is the same as the dimension of $%
X_{j}$, our asymptotic analysis can be applied for any fixed value
of $K$.\footnote{A theoretically efficient estimator can be obtained using the sieve minimum
distance approach. It entails first estimating the conditional expectation $\E%
\left[ \left( 1\{Y_{j}<X_{j}\beta \}-q\right) \mid Z_{j}\right] $ using $\Phi
^{K}\left( Z_{j}\right) $ as the basis functions and then choosing $\beta $
to minimize a weighted sum of squared conditional expectations. See, for
example, \citet{ChenPouzo2009,ChenPouzo2012}. To achieve the semiparametric
efficiency bound, $K$ has to grow with the sample size at an appropriate
rate. In work in progress, we consider nonparametric quantile regression
with endogeneity and allow $K$ to diverge, which is necessary for both
identification and efficiency. Here we are content with a fixed $K$ for
empirical convenience at the cost of possible efficiency loss.}

\subsection{Comparison with other estimators\label{sec:SEE-comp}}


\subsubsection*{Smoothed criterion function}

For the special case $Z_{j}=X_{j}$, we compare the SEE with the estimating equations derived
from smoothing the criterion function as in \citet{Horowitz1998}. The first
order condition of the smoothed criterion function, evaluated at the true $%
\beta _{0}$, is 
\begin{align}
\notag
0& =\left. \frac{\partial }{\partial \beta }\right \vert _{\beta =\beta
_{0}}n^{-1}\sum_{i=1}^{n}\left[ q-G\left( \frac{X_{i}^{\prime }\beta -Y_{i}}{%
h}\right) \right] (Y_{i}-X_{i}^{\prime }\beta ) \\
\notag
& =n^{-1}\sum_{i=1}^{n}\Big[-qX_{i}-G^{\prime
}(-U_{i}/h)(X_{i}/h)Y_{i}+G^{\prime }(-U_{i}/h)(X_{i}/h)X_{i}^{\prime }\beta
_{0}+G(-U_{i}/h)X_{i}\Big] \\
\notag
& =n^{-1}\sum_{i=1}^{n}X_{i}\left[ G(-U_{i}/h)-q\right] +n^{-1}%
\sum_{i=1}^{n}G^{\prime }(-U_{i}/h)\left[ (X_{i}/h)X_{i}^{\prime }\beta
_{0}-(X_{i}/h)Y_{i}\right] \\
\label{eqn:SCF-EE}
& =n^{-1}\sum_{i=1}^{n}X_{i}\left[ G(-U_{i}/h)-q\right] +n^{-1}%
\sum_{i=1}^{n}(1/h)G^{\prime }(-U_{i}/h)[-X_{i}U_{i}]. 
\end{align}%
The first term agrees with our proposed SEE. Technically, it should be
easier to establish high-order results for our SEE estimator since it has
one fewer term. Later we show that the absolute bias of our SEE estimator is
smaller, too. Another subtle point is that our SEE requires only the
estimating equation $\E\left[X_{j}\left( 1\{U_{j}<0\}-q\right)\right] =0$, whereas %
\citet{Horowitz1998} has to impose an additional condition to ensure that
the second term in the FOC is approximately mean zero.


\subsubsection*{IV mean regression\label{sec:see-comp-IV}}

When $h\rightarrow \infty $, $G(\cdot )$ only takes arguments near zero and
thus can be approximated well linearly. For example, with the $G(\cdot )$
from \citet{Whang2006} and \citet{Horowitz1998}, $G(v)=0.5+(105/64)v+O(v^{3})
$ as $v\rightarrow 0$. Ignoring the $O(v^{3})$, the corresponding estimator $%
\hat{\beta}_{\infty }$ is defined by 
\begin{align*}
0& =\sum_{i=1}^{n}Z_{i}\left[ G\left( \frac{X_{i}^{\prime }\hat{\beta}%
_{\infty }-Y_{i}}{h}\right) -q\right]  \\
& \doteq \sum_{i=1}^{n}Z_{i}\left[ \left( 0.5+(105/64)\frac{X_{i}^{\prime }%
\hat{\beta}_{\infty }-Y_{i}}{h}\right) -q\right]  \\
& =(105/64h)Z^{\prime }X\hat{\beta}_{\infty }-(105/64h)Z^{\prime
}Y+(0.5-q)Z^{\prime }\mathbf{1}_{n,1} \\
& =(105/64h)Z^{\prime }X\hat{\beta}_{\infty }-(105/64h)Z^{\prime
}Y+(0.5-q)Z^{\prime }(Xe_{1}) ,
\end{align*}%
where $e_{1}=(1,0,\ldots ,0)^{\prime }$ is $d\times 1$, $\mathbf{1}%
_{n,1}=(1,1,\ldots ,1)^{\prime }$ is $n\times 1$, $X$ and $Z$ are $n\times d$
with respective rows $X_{i}^{\prime }$ and $Z_{i}^{\prime }$, and using the
fact that the first column of $X$ is $\mathbf{1}_{n,1}$ so that $Xe_{1}=%
\mathbf{1}_{n,1}$. It then follows that 
\begin{equation*}
\hat{\beta}_{\infty }=\hat{\beta}_{IV}+\left( (64h/105)(q-0.5),0,\ldots
,0\right) ^{\prime }.
\end{equation*}%
As $h$ grows large, the smoothed QR estimator approaches the IV estimator
plus an adjustment to the intercept term that depends on $q$, the bandwidth,
and the slope of $G(\cdot )$ at zero. In the special case $Z_{j}=X_{j}$, the
IV estimator is the OLS estimator.\footnote{%
This is different from \citet{ZhouEtAl2011}, who add the $d$ OLS moment
conditions to the $d$ median regression moment conditions before estimation;
our connection to IV/OLS emerges naturally from smoothing the (IV)QR
estimating equations.}

The intercept is often not of interest, and when $q=0.5$, the adjustment is
zero anyway. The class of SEE estimators is a continuum (indexed by $h$)
with two well-known special cases at the extremes: unsmoothed IV-QR and mean
IV. For $q=0.5$ and $Z_j=X_j$, this is median regression and mean regression
(OLS). Well known are the relative efficiency advantages of 
the median and the mean for different error distributions. Our estimator
with a data-driven bandwidth can harness the advantages of both, without
requiring the practitioner to make guesses about the unknown error
distribution.


\subsubsection*{Robust estimation}

With $Z_j=X_j$, the result that our SEE can yield OLS when $h\to \infty$ or
median regression when $h=0$ calls to mind robust estimators like the
trimmed or Winsorized mean (and corresponding regression estimators).
Setting the trimming/Winsorization parameter to zero generates the mean
while the other extreme generates the median. However, our SEE mechanism is
different and more general/flexible; trimming/Winsorization is not directly
applicable to $q\ne0.5$; our method to select the smoothing parameter is
novel; and the motivations for QR extend beyond (though include) robustness.

With $X_{i}=1$ and $q=0.5$ (population median estimation), our SEE becomes 
\begin{equation*}
0=n^{-1}\sum_{i=1}^{n}\left[ 2G\left(\frac{\beta -Y_{i}}{h}\right)-1\right] .
\end{equation*}%
If $G'(u)=1\{-1\le u\le 1\}/2$ (the uniform kernel), then $H(u)\equiv 2G(u)-1=u$ for $u\in[-1,1]$, $H(u)=1$ for $u>1$, and $H(u)=-1$ for $u<-1$.  The SEE is then $0=\sum_{i=1}^{n}\psi\left(Y_i;\beta\right)$ with $\psi\left(Y_i;\beta\right)=H\left(\left(\beta-Y_i\right)/h\right)$.  This produces the Winsorized mean estimator of the type in \citet[example (iii), p.\ 79]{Huber1964}.\footnote{%
For a strict mapping, multiply by $h$ to get $\psi (Y_{i};\beta )=hH[(\beta
-Y_{i})/h]$. The solution is equivalent since $\sum h\psi (Y_{i};\beta )=0$
is the same as $\sum \psi (Y_{i};\beta )=0$ for any nonzero constant $h$.}



Further theoretical comparison of our SEE-QR with trimmed/Winsorized mean
regression (and the IV versions) would be interesting but is beyond the
scope of this paper. For more on robust location and regression estimators,
see for example \citet{Huber1964}, \citet{KoenkerBassett1978}, and %
\citet{RuppertCarroll1980}.

\section{MSE of the SEE\label{sec:mse}}

Since statistical inference can be made based on the estimating equations
(EEs), we examine the mean squared error (MSE) of the SEE. 
An advantage of using EEs directly is that inference can be made robust to the
strength of identification. Our focus on the EEs is also in the same
spirit of the large literature on optimal estimating equations. For the
historical developments of EEs and their applications in econometrics, see %
\citet{BeraEtAl2006}. The MSE of the SEE is also related to the estimator MSE
and inference properties both intuitively and (as we will show) theoretically. Such results may provide
helpful guidance in contexts where the SEE MSE is easier to compute than the
estimator MSE, and it provides insight into how smoothing works in the QR
model as well as results that will be used in subsequent sections.

We maintain different subsets of the following assumptions for different
results. 
We write $f_{U|Z}(\cdot \mid z)$ and $F_{U|Z}(\cdot \mid z)$ as the
conditional PDF and CDF of $U$ given $Z=z$. We define $f_{U|Z,X}(\cdot \mid
z,x)$ and $F_{U|Z,X}(\cdot \mid z,x) $ similarly.

\begin{assumption}
\label{a:sampling} $(X_{j}^{\prime },Z_{j}^{\prime },Y_{j})$ is iid across $%
j=1,2,\ldots ,n$, where $Y_{j}=X_{j}^{\prime }\beta _{0}+U_{j}$, $X_{j}$ is
an observed $d\times 1$ vector of stochastic regressors that can include a
constant, $\beta _{0}$ is an unknown $d\times 1$ constant vector, $U_{j}$ is
an unobserved random scalar, and $Z_{j}$ is an observed $d\times 1$ vector
of instruments such that $\E\left[Z_{j}\left( 1\{U_{j}<0\}-q\right)\right] =0$.
\end{assumption}

\begin{assumption}
\label{a:rank} (i) $Z_{j}$ has bounded support. (ii) $\E\left(
Z_{j}Z_{j}^{\prime }\right) $ is nonsingular.
\end{assumption}

\begin{assumption}
\label{a:fUZ}(i) $P(U_{j}<0\mid Z_{j}=z)=q$ for almost all $z\in \mathcal{Z}$%
, the support of $Z$. (ii) For all $u$ in a neighborhood of zero and almost
all $z\in \mathcal{Z}$, $f_{U|Z}(u\mid z)$ exists, is bounded away from
zero, and is $r$ times continuously differentiable with $r\geq 2$. (iii)
There exists a function $C(z)$ such that $\left \vert f_{U|Z}^{(s)}(u\mid
z)\right
\vert \leq C(z)$ for $s=0,2,\ldots ,r$, almost all $z\in \mathcal{Z%
}$ and $u$ in a neighborhood of zero, and $\E\left[ C(Z)\left
\Vert
Z\right
\Vert ^{2}\right] <\infty $.
\end{assumption}

\begin{assumption}
\label{a:G} (i) $G(v)$ is a bounded function satisfying $G(v)=0$ for $%
v\leq-1 $, $G(v)=1$ for $v\geq 1$, and $1-\int_{-1}^{1}G^{2}(u)du>0$. (ii) $%
G^{\prime }(\cdot )$ is a symmetric and bounded $r$th order kernel with $%
r\geq 2$ so that $\int_{-1}^{1}G^{\prime }(v)dv=1$, $\int_{-1}^{1}v^{k}G^{%
\prime }(v)dv=0$ for $k=1,2,\ldots ,r-1$, $\int_{-1}^{1}\left \vert
v^{r}G^{\prime }(v)\right \vert dv<\infty $, and $\int_{-1}^{1}v^{r}G^{%
\prime }(v)dv\neq 0$. (iii) Let $\tilde{G}(u)=\left( G(u),[G(u)]^{2},\ldots
,[G(u)]^{L+1}\right) ^{\prime }$ for some $L\geq 1$. For any $\theta \in 
\mathbb{R}^{L+1}$ satisfying $\left \Vert \theta \right \Vert =1$, there is
a partition of $[-1,1]$ given by $-1=a_{0}<a_{1}<\cdots <a_{\tilde L}=1$ for some finite $\tilde L$ such
that $\theta ^{\prime }\tilde{G}(u)$ is either strictly positive or strictly
negative on the intervals $(a_{i-1},a_{i})$ for $i=1,2,\ldots ,\tilde L$.
\end{assumption}

\begin{assumption}
\label{a:h} $h\propto n^{-\kappa }$ for $1/\left( 2r\right) <\kappa <1$. 
\end{assumption}

\begin{assumption}
\label{a:beta} $\beta=\beta_0$ uniquely solves $\E\left[ Z_{j}\left(
q-1\{Y_{j}<X_{j}^{\prime }\beta \} \right) \right] =0$ over $\beta \in 
\mathcal{B}$.
\end{assumption}

\begin{assumption}
\label{a:power_mse}(i) $f_{U|Z,X}(u\mid z,x)$ is $r$ times continuously
differentiable in $u$ in a neighborhood of zero for almost all $x\in 
\mathcal{X}$ and $z\in \mathcal{Z}$ for $r>2$. (ii) $\Sigma _{ZX}\equiv \E%
\left[ Z_{j}X_{j}^{\prime }f_{U|Z,X}(0\mid Z_{j},X_{j})\right] $ is
nonsingular.
\end{assumption}

%

Assumption \ref{a:sampling} describes the sampling process. Assumption \ref%
{a:rank} is analogous to Assumption 3 in both \citet{Horowitz1998} and %
\citet{Whang2006}. As discussed in these two papers, the boundedness
assumption for $Z_{j}$, which is a technical condition, is made only for
convenience and can be dropped at the cost of more complicated proofs.

Assumption \ref{a:fUZ}(i) allows us to use the law of iterated expectations
to simplify the asymptotic variance. Our qualitative conclusions do not rely
on this assumption. Assumption \ref{a:fUZ}(ii) is critical. If we are not
willing to make such an assumption, then smoothing will be of no benefit.
Inversely, with some small degree of smoothness of the conditional error
density, smoothing can leverage this into the advantages described here.
Also note that \citet{Horowitz1998} assumes $r\geq 4$, which is sufficient
for the estimator MSE result in Section \ref{sec:est}.

%
Assumptions \ref{a:G}(i--ii) are analogous to the standard high-order kernel
conditions in the kernel smoothing literature. The integral condition in (i)
ensures that smoothing reduces (rather than increases) variance.  Note that
\begin{align*}
1-\int_{-1}^{1}G^{2}(u)du& =2\int_{-1}^{1}uG(u)G^{\prime }(u)du \\
& =2\int_{0}^{1}uG(u)G^{\prime }(u)du+2\int_{-1}^{0}uG(u)G^{\prime }(u)du \\
& =2\int_{0}^{1}uG(u)G^{\prime }(u)du-2\int_{0}^{1}vG(-v)G^{\prime }(-v)dv \\
& =2\int_{0}^{1}uG^{\prime }(u)\left[ G(u)-G(-u)\right] du,
\end{align*}
using the evenness of $G'(u)$.  When $r=2$, we can use any $G(u)$ such that $G'(u)$ is a symmetric PDF on $[-1,1]$.  
In this case, $1-\int_{-1}^1 G^2(u)du>0$ holds automatically. 
When $r>2$, $G'(u)<0$ for some $u$, and $G(u)$ is not monotonic.  It is not easy to sign $1-\int_{-1}^1 G^2(u)du$ generally, but it is simple to calculate this quantity for any chosen $G(\cdot)$.  For example, consider $r=4$ and the $G(\cdot)$ function in \citet{Horowitz1998} and \citet{Whang2006} shown in Figure \ref{fig:G}:
\begin{equation}
G(u)=\left \{ 
\begin{array}{ll}
0, & u\leq -1 \\ 
0.5+\frac{105}{64}\left( u-\frac{5}{3}u^{3}+\frac{7}{5}u^{5}-\frac{3}{7}%
u^{7}\right) , & u\in \lbrack -1,1] \\ 
1 & u\geq 1%
\end{array}%
\right.   \label{eqn:G_fun}
\end{equation}%
The range of the function is outside $[0,1]$.  Simple calculations show that $1-\int_{-1}^1 G^2(u)du>0$. 

\begin{figure}[htbp]
\begin{center}
\includegraphics[width=0.7\textwidth,clip=true,trim=35 35 20 70]{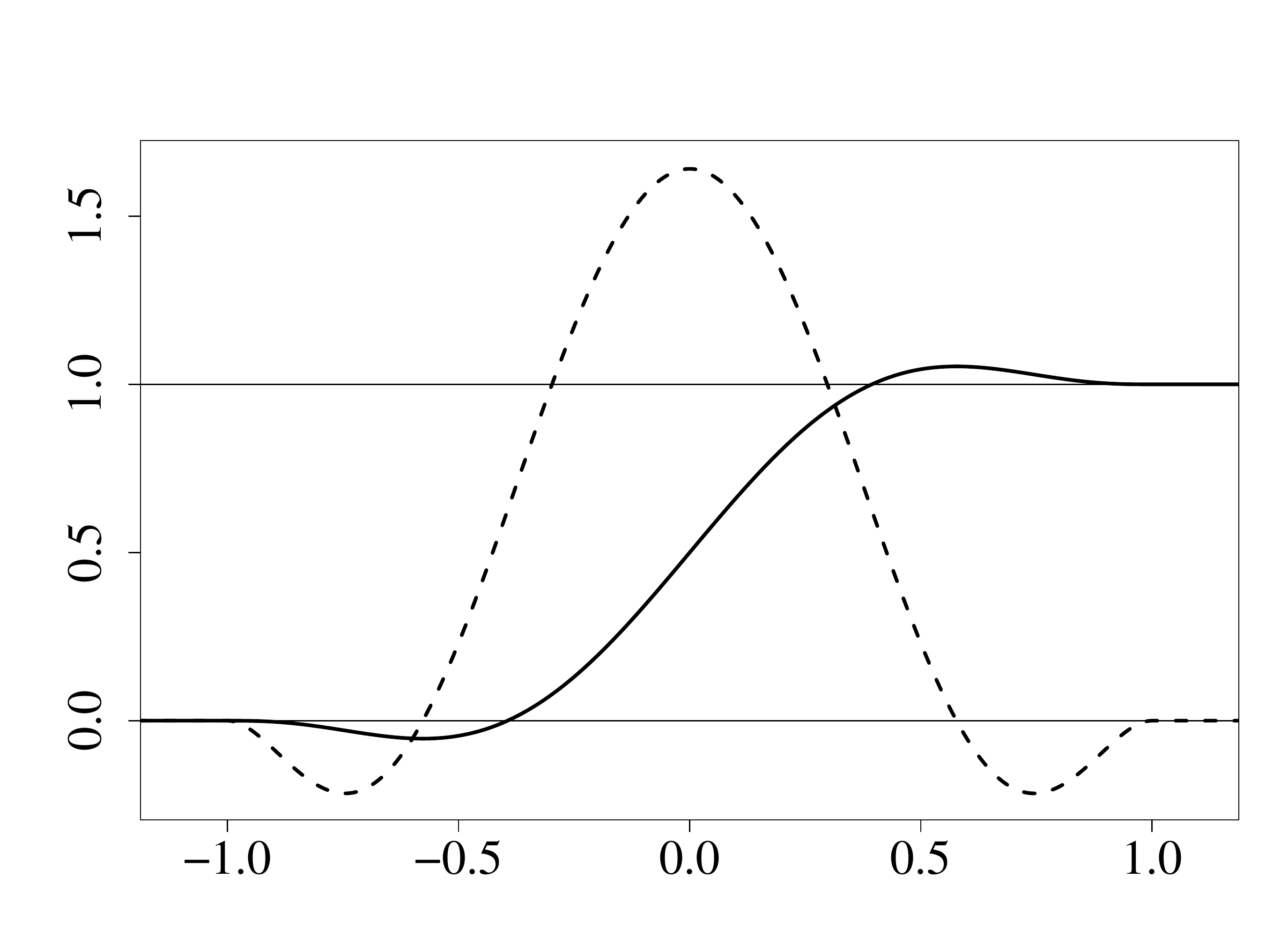}%
\end{center}
\caption{Graph of $G(u)=0.5+\frac{105}{64}\left( u-\frac{5}{3}u^{3}+\frac{7}{%
5}u^{5}-\frac{3}{7}u^{7}\right) $ (solid line) and its derivative (broken).}
\label{fig:G}
\end{figure}

Assumption \ref{a:G}(iii) is needed for the Edgeworth expansion. As \citet{Horowitz1998} and %
\citet{Whang2006} discuss, Assumption \ref{a:G}(iii) is a technical
assumption that (along with Assumption \ref{a:h}) leads to a form of Cram\'{e%
}r's condition, which is needed to justify the Edgeworth expansion used in 
Section \ref{sec:eI}. Any $G(u)$ constructed by integrating polynomial kernels in \citet{Muller1984} satisfies Assumption \ref{a:G}(iii).  In fact, $G(u)$ in \eqref{eqn:G_fun} is obtained by integrating a fourth-order kernel given in Table 1 of \citet{Muller1984}. 
Assumption \ref{a:h} ensures that the bias of the SEE is of
smaller order than its variance. It is needed for the asymptotic normality
of the SEE as well as the Edgeworth expansion.  

Assumption \ref{a:beta} is an identification assumption. See Theorem 2 of %
\citet{ChernozhukovHansen2006} for more primitive conditions. It ensures the
consistency of the SEE estimator. Assumption \ref{a:power_mse} is necessary
for the $\sqrt{n}$-consistency and asymptotic normality of the SEE estimator.

Define%
\begin{equation*}
W_{j}\equiv W_{j}(\beta _{0})=Z_{j}\left[ G(-U_{j}/h)-q\right]
\end{equation*}%
and abbreviate $m_{n}\equiv m_{n}(\beta _{0})=n^{-1/2}\sum_{j=1}^{n}W_{j}$.
The theorem below gives the first two moments of $W_{j}$ and the first-order
asymptotic distribution of $m_{n}$.

\begin{theorem}
\label{thm:Wj} Let Assumptions \ref{a:rank}(i), \ref{a:fUZ}, and \ref{a:G}%
(i--ii) hold. Then 
\begin{align}
\E(W_{j})& =\frac{(-h)^{r}}{r!}\left[ \int_{-1}^{1}G^{\prime
}(v)v^{r}dv\right] \E\left[ f_{U|Z}^{(r-1)}(0\mid Z_{j})Z_{j}\right] +o\left(
h^{r}\right) ,  \label{eqn:see-bias} \\
\E(W_{j}^{\prime }W_{j})
  &= q(1-q)\E\left(Z_{j}^{\prime }Z_{j}\right)
    -h\left[1-\int_{-1}^{1}G^{2}(u)du\right] \E\left[f_{U|Z}(0\mid Z_{j})Z_{j}^{\prime
}Z_{j}\right]+O(h^{2}),  \label{eqn:W-var-trace} \\
\E(W_{j}W_{j}^{\prime })& =q(1-q)\E\left(Z_{j}Z_{j}^{\prime }\right)-h\left[
1-\int_{-1}^{1}G^{2}(u)du\right] \E\left[f_{U|Z}(0\mid Z_{j})Z_{j}Z_{j}^{\prime
}\right]+O(h^{2}).  \notag
\end{align}%
If additionally Assumptions \ref{a:sampling} and \ref{a:h} hold, then 
\begin{equation*}
m_{n}\overset{d}{\rightarrow }N(0,V),\quad V\equiv \lim_{n\rightarrow \infty
}\E\left\{ \left[W_{j}-\E(W_{j})\right]\left[W_{j}-\E(W_{j})\right]^{\prime }\right\}
=q(1-q)\E\left(Z_{j}Z_{j}^{\prime }\right).
\end{equation*}
\end{theorem}

Compared with the EE derived from smoothing the criterion function as in \citet{Horowitz1998}, our SEE has smaller bias and variance, and these differences affect the bias and variance of the parameter estimator. 
The former approach only applies to exogenous QR with $Z_j=X_j$. 
The EE derived from smoothing the criterion function in \eqref{eqn:SCF-EE} for $Z_j=X_j$ can be written
\begin{align}\label{eqn:SCF-EE-Wj}
0 &= n^{-1}\sum_{j=1}^n W_j, \quad
W_j \equiv X_j\left[G(-U_j/h)-q\right] + (1/h)G'(-U_j/h)(-X_jU_j) .
\end{align}
Consequently, as calculated in the appendix,
\begin{align}\label{eqn:SCF-EWj}
\E(W_j) 
  & =(r+1)\frac{(-h)^r}{r!} \left[
\int G^{\prime }(v)v^{r}dv\right] \E\left[ f_{U|Z}^{(r-1)}(0\mid Z_{j})Z_{j}%
\right] +o\left( h^{r}\right) , \\
\E(W_jW_j') \label{eqn:SCF-EWjWj}
  &= q(1-q)\E(X_jX_j')
     +h \int_{-1}^1[G'(v)v]^2dv \, \E\left[f_{U|X}(0\mid X_j)X_jX_j'\right]
     +O(h^2) , \\
\E & \left[\D{}{\beta'} n^{-1/2}m_n(\beta_0)\right] \label{eqn:SCF-EdBmn}
   = \E\left[f_{U|X}(0\mid X_j)X_jX_j'\right] 
    -h \E\left[f_{U|X}'(0\mid X_j)X_jX_j'\right] 
    +O(h^2) .
\end{align}
The dominating term of the bias of our SEE in \eqref{eqn:see-bias} is $r+1$ times smaller in absolute value than that of the EE derived from a smoothed criterion function in \eqref{eqn:SCF-EWj}.  
A larger bias can lead to less accurate confidence regions if the same
variance estimator is used. 
Additionally, the smoothed criterion function analog of $\E(W_jW_j')$ in \eqref{eqn:SCF-EWjWj} has a positive $O(h)$ term instead of the negative $O(h)$ term for SEE.  
The connection between these terms and the estimator's asymptotic mean squared error (AMSE) is shown in Section \ref{sec:est} to rely on the inverse of the matrix in equation \eqref{eqn:SCF-EdBmn}.  Here, though, the sign of the $O(h)$ term is indeterminant since it depends on a PDF derivative.  (A negative $O(h)$ term implies higher AMSE since this matrix is inverted in the AMSE expression, and positive implies lower.)  If $U=0$ is a mode of the conditional (on $X$) distribution, then the $O(h)$ term is zero and the AMSE comparison is driven by $\E(W_j)$ and $\E(W_jW_j')$.  Since SEE yields smaller $\E(W_jW_j')$ and smaller absolute $\E(W_j)$, it will have smaller estimator AMSE in such cases.  Simulation results in Section \ref{sec:sim} add evidence that the SEE estimator usually has smaller MSE in practice. 

The first-order asymptotic variance $V$ is the same as the asymptotic
variance of 
\begin{equation*}
n^{-1/2}\sum_{j=1}^n Z_j \left(1\{U_j<0\}-q\right) ,
\end{equation*}
the scaled EE of the unsmoothed IV-QR. The effect of smoothing to reduce
variance is captured by the term of order $h$, where $1-%
\int_{-1}^{1}G^2(u)du>0$ by Assumption \ref{a:G}(i). This reduction in
variance is not surprising. Replacing the discontinuous indicator function $%
1\{U<0\}$ by a smooth function $G(-U/h)$ pushes the dichotomous values of
zero and one into some values in between, leading to a smaller variance. The
idea is similar to \citeauthor{Breiman1994}'s (\citeyear{Breiman1994})
bagging (bootstrap aggregating), among others. 

Define the MSE of the SEE to be $\E\left(m_{n}^{\prime }V^{-1}m_{n}\right)$. Building
upon \eqref{eqn:see-bias} and \eqref{eqn:W-var-trace}, and using $W_{i}%
\mathpalette{\protect \independenT}{\perp}W_{j}$ for $i\neq j$, we have: 
\begin{align}
& \E\left(m_{n}^{\prime }V^{-1}m_{n}\right)  \notag \\
& =\frac{1}{n}\sum_{j=1}^{n}\E\left(W_{j}^{\prime }V^{-1}W_{j}\right)+\frac{1}{n}%
\sum_{j=1}^{n}\sum_{i\neq j}\E\left( W_{i}^{\prime }V^{-1}W_{j}\right)  
\notag \\
& =\frac{1}{n}\sum_{j=1}^{n}\E\left(W_{j}^{\prime }V^{-1}W_{j}\right)+\frac{1}{n}%
n(n-1)\E(W_{j}^{\prime })V^{-1}\E(W_{j})  \notag \\
& =q(1-q)\E\left(Z_{j}^{\prime }V^{-1}Z_{j}\right)+nh^{2r} \E(B)'\E(B) -h\mathrm{tr}\left[ \E\left( AA^{\prime }\right) \right]
+o\left(h+nh^{2r}\right),  \notag \\
& =d+nh^{2r} \E(B)'\E(B) -h\mathrm{tr}\left[
\E\left( AA^{\prime }\right) \right] +o\left(h+nh^{2r}\right),  \label{eqn:mse}
\end{align}%
where 
\begin{align*}
A& \equiv \left[ 1-\int_{-1}^{1}G^{2}(u)du\right] ^{1/2}\left[ f_{U|Z}(0\mid
Z)\right] ^{1/2}V^{-1/2}Z, \\
B& \equiv \left[ \frac{1}{r!}\int_{-1}^{1}G^{\prime }(v)v^{r}dv\right]
f_{U|Z}^{(r-1)}(0\mid Z)V^{-1/2}Z.
\end{align*}

Ignoring the $o(\cdot )$ term, we obtain the asymptotic MSE of the SEE. We
select the smoothing parameter to minimize the asymptotic MSE: 
\begin{equation}
h_{\text{SEE}}^{\ast }
   \equiv \mathop{\rm arg\,min}_{h}
              nh^{2r}\E(B)'\E(B) 
            - h\mathrm{tr}\left[ \E\left(AA'\right) \right] .  \label{eqn:def_h_SEE}
\end{equation}%
The proposition below gives the optimal smoothing parameter $h_{\text{SEE}%
}^{\ast }$.

\begin{proposition}
\label{prop:hSEE} Let Assumptions \ref{a:sampling}, \ref{a:rank}, \ref{a:fUZ}%
, and \ref{a:G}(i--ii) hold. The bandwidth that minimizes the asymptotic MSE
of the SEE is 
\begin{equation*}
h_{\text{SEE}}^{\ast }
   = \left( \frac{\mathrm{tr}\left[ \E\left(AA'\right) \right] }{\E(B)'\E(B) }\frac{1}{2nr}\right) ^{%
\frac{1}{2r-1}}.
\end{equation*}%
Under the stronger assumption $U\mathpalette{\protect \independenT}{\perp} Z$%
, 
\begin{equation*}
h_{\text{SEE}}^{\ast }=\left( \frac{\left( r!\right) ^{2}\left[
1-\int_{-1}^{1}G^{2}(u)du\right] f_{U}(0)}{2r\left[ \int_{-1}^{1}G^{\prime
}(v)v^{r}dv\right] ^{2}\left[ f_{U}^{\left( r-1\right) }(0)\right] ^{2}}%
\frac{d}{n}\right) ^{\frac{1}{2r-1}}.
\end{equation*}
\end{proposition}

When $r=2$, the MSE-optimal $h_{\text{SEE}}^{\ast }\asymp
n^{-1/(2r-1)}=n^{-1/3}$. This is smaller than $n^{-1/5}$, the rate that
minimizes the MSE of estimated standard errors of the usual regression
quantiles. Since nonparametric estimators of $f_{U}^{(r-1)}(0)$ converge
slowly, we propose a parametric plug-in described in Section \ref{sec:sim}.

We point out in passing that the optimal smoothing parameter $h_{\text{SEE}%
}^{\ast }$ is invariant to rotation and translation of the (non-constant)
regressors. This may not be obvious but can be proved easily.

For the unsmoothed IV-QR, let 
\begin{equation*}
\tilde{m}_{n} = \frac{1}{\sqrt{n}}\sum_{j=1}^{n}Z_{j}\left( 1\left \{
Y_{j}\leq X_{j}'\beta \right \} -q\right) ,
\end{equation*}%
then the MSE of the estimating equations is $\E\left(\tilde{m}_{n}^{\prime }V^{-1}%
\tilde{m}_{n}\right)=d$. Comparing this to the MSE of the SEE given in %
\eqref{eqn:mse}, we find that the SEE has a smaller MSE when $h=h_{\text{SEE}%
}^{\ast }$ because 
\begin{equation*}
n(h_{\text{SEE}}^{\ast })^{2r} \E(B)'\E(B)
-h_{\text{SEE}}^{\ast }\mathrm{tr}\left[ \E\left(AA'\right) \right]
=-h_{\text{SEE}}^{\ast }\left( 1-\frac{1}{2r}\right) \mathrm{tr}%
\left[ \E\left(AA'\right) \right] <0.
\end{equation*}%
In terms of MSE, it is advantageous to smooth the estimating equations. To
the best of our knowledge, this point has never been discussed before in the
literature.

\section{Type {I} and Type {II} Errors of a Chi-square Test\label{sec:eI}}

In this section, we explore the effect of smoothing on a chi-square test.
Other alternatives for inference exist, such as the Bernoulli-based
MCMC-computed method from \citet{ChernozhukovEtAl2009}, empirical likelihood
as in \citet{Whang2006}, and bootstrap as in \citet{Horowitz1998}, where the
latter two also use smoothing. Intuitively, when we minimize the MSE, we may
expect lower type I error: the $\chi ^{2}$ critical value is from the
unsmoothed distribution, and smoothing to minimize MSE makes large values
(that cause the test to reject) less likely. 
The reduced MSE also makes it easier to distinguish the null hypothesis from
some given alternative. This combination leads to improved size-adjusted
power. As seen in our simulations, this is true especially for the IV case.

Using the results in Section \ref{sec:mse} and under Assumption \ref{a:h}, we
have 
\begin{equation*}
m_{n}^{\prime }V^{-1}m_{n}\overset{d}{\rightarrow }\chi _{d}^{2},
\end{equation*}%
where we continue to use the notation $m_{n}\equiv m_{n}(\beta _{0})$. From
this asymptotic result, we can construct a hypothesis test that rejects the
null hypothesis $H_{0}:\beta =\beta _{0}$ when 
\begin{equation*}
S_{n}\equiv m_{n}^{\prime }\hat{V}^{-1}m_{n}>c_{\alpha },
\end{equation*}%
where 
\begin{equation*}
\hat{V}=q(1-q)\frac{1}{n}\sum_{j=1}^{n}Z_{j}Z_{j}^{\prime }
\end{equation*}%
is a consistent estimator of $V$ and $c_{\alpha }\equiv \chi _{d,1-\alpha
}^{2}$ is the $1-\alpha $ quantile of the chi-square distribution 
with $d$ degrees of freedom. As desired, the asymptotic size is 
\begin{equation*}
\lim_{n\rightarrow \infty }P\left( S_{n}>c_{\alpha }\right) =\alpha .
\end{equation*}%
Here $P\equiv P_{\beta _{0}}$ is the probability measure under the true
model parameter $\beta _{0}$. We suppress the subscript $\beta _{0}$ when
there is no confusion.

It is important to point out that the above result does not rely on the
strong identification of $\beta _{0}$. It still holds if $\beta _{0}$ is
weakly identified or even unidentified. This is an advantage of focusing on
the estimating equations instead of the parameter estimator. When a direct
inference method based on the asymptotic normality of $\hat{\beta}$ is used,
we have to impose Assumptions \ref{a:beta} and \ref{a:power_mse}.

\subsection{Type {I} error and the associated optimal bandwidth}

To more precisely measure the type I error $P\left( S_{n}>c_{\alpha }\right) 
$, we first develop a high-order stochastic expansion of $S_{n}$. Let $%
V_{n}\equiv \mathrm{Var}\left( m_{n}\right) $. Following the same
calculation as in \eqref{eqn:mse}, we have 
\begin{align*}
V_{n}& =V-h\left[ 1-\int_{-1}^{1}G^{2}(u)du\right] \E\left[f_{U|Z}(0\mid
Z_{j})Z_{j}Z_{j}^{\prime }\right]+O(h^{2}) \\
& =V^{1/2}\left[ I_{d}-h\E\left( AA'\right) +O\left(h^{2}\right)\right] \left(
V^{1/2}\right) ^{\prime },
\end{align*}%
where $V^{1/2}$ is the matrix square root of $V$ such that $V^{1/2}\left(
V^{1/2}\right) ^{\prime }=V$. We can choose $V^{1/2}$ to be symmetric but do
not have to.

Details of the following are in the appendix; here we outline our strategy
and highlight key results. Letting 
\begin{equation}
\Lambda _{n}=V^{1/2}\left[ I_{d}-h\E\left( AA'\right) +O\left(h^{2}\right)%
\right] ^{1/2}  \label{Lambda_n}
\end{equation}%
such that $\Lambda _{n}\Lambda _{n}^{\prime }=V_{n}$, and defining 
\begin{equation}
\bar{W}_{n}^{\ast }\equiv \frac{1}{n}\sum_{j=1}^{n}W_{j}^{\ast }\text{ and }%
W_{j}^{\ast }=\Lambda _{n}^{-1}Z_{j}\left[ G(-U_{j}/h)-q\right] ,
\label{define_W_star}
\end{equation}%
we can approximate the test statistic as 
$S_{n}=S_{n}^{L}+e_{n}$, 
where 
\begin{equation*}
S_{n}^{L}=\left( \sqrt{n}\bar{W}_{n}^{\ast }\right) ^{\prime }\left( \sqrt{n}%
\bar{W}_{n}^{\ast }\right) -h\left( \sqrt{n}\bar{W}_{n}^{\ast }\right)
^{\prime }\E\left( AA'\right) \left( \sqrt{n}\bar{W}_{n}^{\ast
}\right) 
\end{equation*}%
and $e_{n}$ is the remainder term satisfying $P\left( \left \vert
e_{n}\right \vert >O\left( h^{2}\right) \right) =O\left( h^{2}\right) $.%

The stochastic expansion above allows us to approximate the characteristic
function of $S_{n}$ with that of $S_{n}^{L}$. Taking the Fourier--Stieltjes
inverse of the characteristic function yields an approximation of the
distribution function, from which we can calculate the type I error by
plugging in the critical value $c_{\alpha}$.

\begin{theorem}
\label{thm:inf} Under Assumptions \ref{a:sampling}--\ref{a:h}, we have 
\begin{align*}
P\left(S_{n}^{L}<x\right) 
  &= \mathcal{G}_{d}(x) -\mathcal{G}_{d+2}^{\prime}(x) \left \{
nh^{2r}\E(B)'\E(B) -h\mathrm{tr}\left
[\E\left( AA'\right) \right] \right \} +R_n, \\
P\left( S_{n}>c_{\alpha }\right) &= \alpha +\mathcal{G}_{d+2}^{\prime
}(c_{\alpha }) \left \{ nh^{2r}\E(B)'\E(B) -h\mathrm{tr}%
\left[\E\left( AA'\right) \right] \right \} +R_{n},
\end{align*}%
where $R_{n}=O\left(h^{2}+nh^{2r+1}\right)$ and $\mathcal{G}_{d}(x)$ is the CDF of the $%
\chi _{d}^{2}$ distribution.
\end{theorem}

From Theorem \ref{thm:inf}, an approximate measure of the type I error of
the SEE-based chi-square test is%
\begin{equation*}
\alpha +\mathcal{G}_{d+2}^{\prime }(c_{\alpha })\left\{ nh^{2r}\E(B)'\E(B) -h\mathrm{tr}%
\left[\E\left( AA'\right) \right] \right\} ,
\end{equation*}%
and an approximate measure of the coverage probability error (CPE) is%
\footnote{%
The CPE is defined to be the nominal coverage minus the true coverage
probability, which may be different from the usual definition. Under this
definition, smaller CPE corresponds to higher coverage probability (and
smaller type I error).} 
\begin{equation*}
\mathrm{CPE}=\mathcal{G}_{d+2}^{\prime }(c_{\alpha })\left\{ nh^{2r}\E(B)'\E(B) -h\mathrm{tr}%
\left[\E\left( AA'\right) \right] \right\} ,
\end{equation*}%
which is also the error in rejection probability under the null.

Up to smaller-order terms, the term $nh^{2r}\E(B)'\E(B)$
characterizes the bias effect from smoothing. The bias increases type I
error and reduces coverage probability. The term $h\mathrm{tr}\left[\E\left( AA'\right) \right]$ characterizes the variance
effect from smoothing. The variance reduction decreases type I error and
increases coverage probability. 
The type I error is $\alpha $ up to order $O\left(h+nh^{2r}\right)$. 
There exists some $h>0$ that makes bias and variance effects cancel, leaving
type I error equal to $\alpha$ up to smaller-order terms in $R_{n}$.

Note that $nh^{2r}\E(B)'\E(B) -h\mathrm{tr}\left[\E\left( AA'\right) \right]$ is identical to the
high-order term in the asymptotic MSE of the SEE in \eqref{eqn:mse}. The $h_{%
\text{CPE}}^{\ast }$ that minimizes type I error is the same as $h_{\text{SEE%
}}^{\ast }$.

\begin{proposition}
\label{prop:hCPE} Let Assumptions \ref{a:sampling}--\ref{a:h} hold. The
bandwidth that minimizes the approximate type I error of the chi-square test
based on the test statistic $S_{n}$ is 
\begin{equation*}
h_{\text{CPE}}^{\ast }=h_{\text{SEE}}^{\ast }=\left( \frac{\mathrm{tr}\left[\E\left( AA'\right) \right]}{\E(B)'\E(B)}%
\frac{1}{2nr}\right) ^{\frac{1}{2r-1}}.
\end{equation*}
\end{proposition}

The result that $h_{\text{CPE}}^{\ast }=h_{\text{SEE}}^{\ast }$ is
intuitive. Since $h_{\text{SEE}}^{\ast }$ minimizes $\E\left(m_{n}^{\prime
}V^{-1}m_{n}\right)$, for a test with $c_{\alpha }$ and $\hat{V}$ both invariant
to $h$, the null rejection probability $P\left(m_{n}^{\prime }\hat{V}%
^{-1}m_{n}>c_{\alpha }\right)$ should be smaller when the SEE's MSE is smaller.

When $h=h_{\text{CPE}}^{\ast }$, 
\begin{equation*}
P\left( S_{n}>c_{\alpha }\right) =\alpha -C^{+}\mathcal{G}_{d+2}^{\prime
}(c_{\alpha })h_{\text{CPE}}^{\ast }\left[1+o(1)\right]
\end{equation*}%
where $C^{+}=\left( 1-\frac{1}{2r}\right) \mathrm{tr}\left[\E\left( AA'\right) \right] >0$. If instead we construct the test statistic
based on the unsmoothed estimating equations, $\tilde{S}_{n}=\tilde{m}%
_{n}^{\prime }\hat{V}^{-1}\tilde{m}_{n}$, then it can be shown that 
\begin{equation*}
P\left( \tilde{S}_{n}>c_{\alpha }\right) =\alpha +Cn^{-1/2}\left[1+o(1)\right]
\end{equation*}%
for some constant $C$, which is in general not equal to zero. Given that $%
n^{-1/2}=o(h_{\text{CPE}}^{\ast })$ and $C^{+}>0$, we can expect the
SEE-based chi-square test to have a smaller type I error in large samples.

\subsection{Type II error and local asymptotic power}

To obtain the local asymptotic power of the $S_{n}$ test, we let the true
parameter value be $\beta _{n}=\beta _{0}-\delta /\sqrt{n}$, where $\beta _{0}
$ is the parameter value that satisfies the null hypothesis $H_{0}$. In this
case, 
\begin{equation*}
m_{n}\left( \beta _{0}\right) =\frac{1}{\sqrt{n}}\sum_{j=1}^{n}Z_{j}\left[
G\left( \frac{X_{j}^{\prime }\delta /\sqrt{n}-U_{j}}{h}\right) -q\right] .
\end{equation*}%
In the proof of Theorem \ref{thm:power}, we show that 
\begin{align*}
\E\left[m_{n}\left( \beta _{0}\right)\right] & =\Sigma _{ZX}\delta +\sqrt{n}%
(-h)^{r}V^{1/2}\E(B)+O\left( n^{-1/2}+\sqrt{n}h^{r+1}\right) , \\
V_{n}& =\mathrm{Var}\left[ m_{n}\left( \beta _{0}\right) \right]
=V-hV^{1/2}\left[ \E\left( AA'\right)\right] (V^{1/2})^{\prime }+O\left( n^{-1/2}+h^{2}\right) .
\end{align*}%

\begin{theorem}
\label{thm:power} Let Assumptions \ref{a:sampling}--\ref{a:h} and \ref{a:power_mse}(i) hold. Define $\Delta \equiv \E\left[V_{n}^{-1/2}m_{n}(\beta _{0})\right]$ and $\tilde{\delta}\equiv V^{-1/2}\Sigma_{ZX}\delta $. We have%
\begin{align*}
P_{\beta _{n}}\left( S_{n}<x\right) 
  &= \mathcal{G}_{d}\left( x;\left \Vert\Delta \right \Vert ^{2}\right)
    +\mathcal{G}_{d+2}^{\prime }\left(x;\Vert \Delta\Vert ^{2}\right)
       h\mathrm{tr}\left[\E\left( AA'\right) \right]  \\
& \quad +\mathcal{G}_{d+4}^{\prime }\left(x;\left \Vert \Delta \right \Vert ^{2}\right)h%
\left[ \Delta ^{\prime }\E\left( AA'\right) \Delta \right] +O\left(
h^{2}+n^{-1/2}\right)  \\
  &= \mathcal{G}_{d}\left( x;\Vert \tilde{\delta}\Vert ^{2}\right) 
    -\mathcal{G}_{d+2}^{\prime }\left(x;\Vert \tilde{\delta}\Vert ^{2}\right)
     \left\{ nh^{2r}\E(B)'\E(B)-h\mathrm{tr}\left[\E\left( AA'\right)\right]\right\}\\
& \quad +\left[ \mathcal{G}_{d+4}^{\prime }\left(x;\Vert \tilde{\delta}\Vert
^{2}\right)-\mathcal{G}_{d+2}^{\prime }\left(x;\Vert \tilde{\delta}\Vert ^{2}\right)\right] h%
\left[ \tilde{\delta}^{\prime }\E\left( AA'\right) \tilde{\delta}%
\right]  \\
& \quad -\mathcal{G}_{d+2}^{\prime }\left(x;\Vert \tilde{\delta}\Vert ^{2}\right)2%
\tilde{\delta}^{\prime }\sqrt{n}(-h)^{r}\E(B)+O\left( h^{2}+n^{-1/2}\right) ,
\end{align*}%
where $\mathcal{G}_{d}(x;\lambda )$ is the CDF of the noncentral chi-square
distribution with degrees of freedom $d$ and noncentrality parameter $%
\lambda $. If we further assume that $\tilde{\delta}$ is uniformly
distributed on the sphere $\mathcal{S}_{d}(\tau )=\{ \tilde{\delta}\in 
\mathbb{R}^{d}:\Vert \tilde{\delta}\Vert =\tau \}$, then 
\begin{align*}
\E_{\tilde{\delta}}& \left[P_{\beta _{n}}\left( S_{n}>c_{\alpha }\right) \right] \\
& =1-\mathcal{G}_{d}\left( c_{\alpha };\tau ^{2}\right) +\mathcal{G}%
_{d+2}^{\prime }(c_{\alpha };\tau ^{2})\left\{ nh^{2r}\E(B)'\E(B)
-h\mathrm{tr}\left[\E\left( AA'\right)\right]\right\}  \\
& \quad -\left[ \mathcal{G}_{d+4}^{\prime }(c_{\alpha };\tau ^{2})-\mathcal{G%
}_{d+2}^{\prime }(c_{\alpha };\tau ^{2})\right] \frac{\tau ^{2}}{d}h\mathrm{tr}\left[\E\left( AA'\right)\right]
+O\left(h^{2}+n^{-1/2}\right) 
\end{align*}%
where $\E_{\tilde{\delta}}$ takes the average uniformly over the sphere $%
\mathcal{S}_{d}(\tau )$.
\end{theorem}

When $\delta=0$, which implies $\tau=0$, 
the expansion in Theorem \ref{thm:power} reduces to that in Theorem \ref%
{thm:inf}.

When $h=h_{\text{SEE}}^{\ast }$, it follows from Theorem \ref{thm:inf} that 
\begin{align*}
P_{\beta _{0}}\left( S_{n}>c_{\alpha }\right) &= 1-\mathcal{G}_{d}\left(
c_{\alpha }\right) -C^{+}\mathcal{G}_{d+2}^{\prime }(c_{\alpha })h_{\text{SEE%
}}^{\ast }+o(h_{\text{SEE}}^{\ast }) \\
&= \alpha -C^{+}\mathcal{G}_{d+2}^{\prime }(c_{\alpha })h_{\text{SEE}}^{\ast
}+o(h_{\text{SEE}}^{\ast }).
\end{align*}%
To remove the error in rejection probability of order $h_{\text{SEE}}^{\ast
} $, we make a correction to the critical value $c_{\alpha }$. Let $%
c_{\alpha }^{\ast }$ be a high-order corrected critical value such that $%
P_{\beta_{0}}\left( S_{n}>c_{\alpha }^{\ast }\right) =\alpha +o(h_{\text{SEE}%
}^{\ast})$. Simple calculation shows that 
\begin{equation*}
c_{\alpha }^{\ast }=c_{\alpha }-\frac{\mathcal{G}_{d+2}^{\prime }(c_{\alpha
})}{\mathcal{G}_{d}^{\prime }\left( c_{\alpha }\right) }C^{+}h_{\text{SEE}%
}^{\ast }
\end{equation*}%
meets the requirement.

To approximate the size-adjusted power of the $S_{n}$ test, we use $%
c_{\alpha }^{\ast }$ rather than $c_{\alpha }$ because $c_{\alpha }^{\ast }$
leads to a more accurate test in large samples. Using Theorem \ref{thm:power}%
, we can prove the following corollary.

\begin{corollary}
\label{cor:power} Let the assumptions in Theorem \ref{thm:power} hold. 
Then for $h=h_{\text{SEE}}^{\ast }$, 
\begin{equation}
\begin{split}
\E_{\tilde{\delta}}& \left[P_{\beta _{n}}\left( S_{n}>c_{\alpha }^{\ast }\right)\right] \\
& =1-\mathcal{G}_{d}\left( c_{\alpha };\tau ^{2}\right) +Q_{d}\left(
c_{\alpha },\tau ^{2},r\right) \mathrm{tr}\left[\E\left( AA'\right)\right]
h_{\text{SEE}}^{\ast }+O\left( h_{\text{SEE}}^{\ast 2}+n^{-1/2}\right) ,
\end{split}
\label{eqn:asym-local-power}
\end{equation}%
where 
\begin{align*}
Q_{d}\left( c_{\alpha },\tau ^{2},r\right) & =\left( 1-\frac{1}{2r}\right) %
\left[ \mathcal{G}_{d}^{\prime }\left( c_{\alpha };\tau ^{2}\right) \frac{%
\mathcal{G}_{d+2}^{\prime }(c_{\alpha })}{\mathcal{G}_{d}^{\prime }\left(
c_{\alpha }\right) }-\mathcal{G}_{d+2}^{\prime }(c_{\alpha };\tau ^{2})%
\right] \\
& \quad -\frac{1}{d}\left[ \mathcal{G}_{d+4}^{\prime }(c_{\alpha };\tau
^{2})-\mathcal{G}_{d+2}^{\prime }(c_{\alpha };\tau ^{2})\right] \tau ^{2}.
\end{align*}
\end{corollary}

In the asymptotic expansion of the local power function in %
\eqref{eqn:asym-local-power}, $1-\mathcal{G}_{d}\left( c_{\alpha };\tau
^{2}\right) $ is the usual first-order power of a standard chi-square test.
The next term of order $O(h_{\text{SEE}}^{\ast })$ captures the effect of
smoothing the estimating equations. To sign this effect, we plot the
function $Q_{d}\left( c_{\alpha },\tau ^{2},r\right) $ against $\tau ^{2}$
for $r=2$, $\alpha =10\%$, and different values of $d$ in Figure \ref%
{fig:power}. Figures for other values of $r$ and $\alpha $ are qualitatively
similar. The range of $\tau ^{2}$ considered in Figure \ref{fig:power} is
relevant as the first-order local asymptotic power, i.e.,\ $1-\mathcal{G}%
_{d}\left( c_{\alpha };\tau ^{2}\right) $, increases from $10\%$ to about $%
94\%$, $96\%$, $97\%$, and $99\%$, respectively for $d=1,2,3,4$. It is clear
from this figure that $Q_{d}\left( c_{\alpha },\tau ^{2},r\right) >0$ for
any $\tau ^{2}>0$. This indicates that smoothing leads to a test with
improved power. The power improvement increases with $r$. The smoother the
conditional PDF of $U$ in a neighborhood of the origin is, the larger the
power improvement is.

\begin{figure}[tbp]
\centering
\includegraphics[width=0.65\textwidth,clip=true,trim=40 190 80 200]{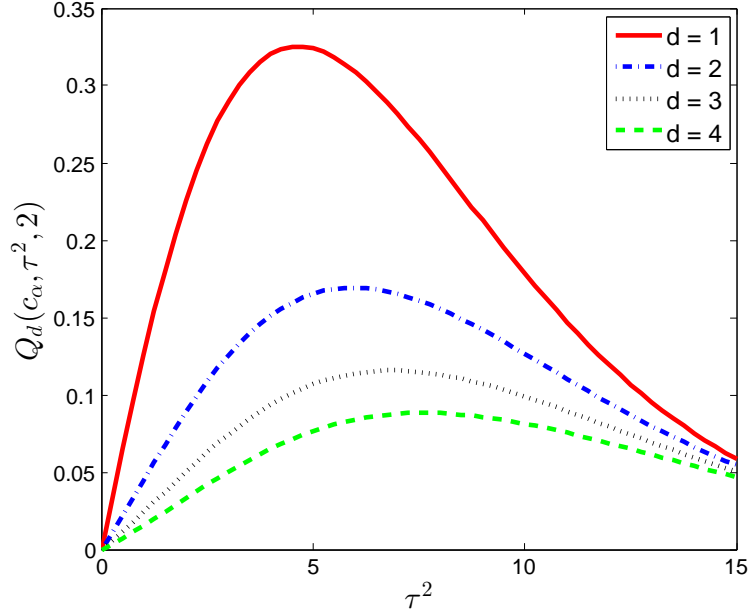}
\caption{Plots of $Q_{d}\left( c_{\protect \alpha },\protect \tau %
^{2},2\right) $ against $\protect \tau ^{2}$ for different values of $d$ with 
$\protect \alpha =10\%$.}
\label{fig:power}
\end{figure}

\section{MSE of the Parameter Estimator\label{sec:est}}

In this section, we examine the approximate MSE of the parameter estimator.
The approximate MSE, being a Nagar-type approximation \citep{Nagar1959}, can
be motivated from the theory of optimal estimating equations, as presented
in \citet{Heyde1997}, for example.

The SEE estimator $\hat{\beta}$ satisfies $m_{n}(\hat{\beta})=0$. In Lemma %
\ref{lem:stochastic_expansion_beta} in the appendix, we show that%
\begin{equation}
\sqrt{n}\left( \hat{\beta}-\beta _{0}\right) 
  = -\left \{ \E\left[\frac{\partial }{%
\partial \beta ^{\prime }}\frac{1}{\sqrt{n}}m_{n}\left( \beta _{0}\right)
\right]\right\} ^{-1}m_{n}+O_{p}\left( \frac{1}{\sqrt{nh}}\right)
\label{eqn:expansion1}
\end{equation}%
and 
\begin{equation}
\E\left[\frac{\partial }{\partial \beta ^{\prime }}\frac{1}{\sqrt{n}}m_{n}\left(
\beta _{0}\right)\right]
  = \E\left[ Z_{j}X_{j}^{\prime }f_{U|Z,X}(0\mid Z_{j},X_{j})%
\right] +O\left(h^{r}\right).  \label{eqn:expansion2}
\end{equation}%
Consequently, the approximate MSE (AMSE) of $\sqrt{n}\left( \hat{\beta}%
-\beta _{0}\right) $ is\footnote{%
Here we follow a common practice in the estimation of nonparametric and
nonlinear models and define the AMSE to be the MSE of $\sqrt{n}\left( \hat{%
\beta}-\beta _{0}\right) $ after dropping some smaller-order terms. So the
asymptotic MSE we define here is a Nagar-type approximate MSE. See %
\citet{Nagar1959}.} 
\begin{align*}
\mathrm{AMSE}_{\beta } 
  &= \left \{ \E\left[\frac{\partial }{\partial \beta ^{\prime
}}\frac{1}{\sqrt{n}}m_{n}\left( \beta _{0}\right) \right] \right\} ^{-1}
\E\left(m_n m_n'\right) \left\{ \E\left[\frac{\partial }{\partial \beta
^{\prime }}\frac{1}{\sqrt{n}}m_{n}\left( \beta _{0}\right) \right]\right\}
^{-1\prime } \\
&= \Sigma _{ZX}^{-1}V\Sigma _{XZ}^{-1}+\Sigma _{ZX}^{-1}V^{1/2}\left[
nh^{2r} \E(B)\E(B') -h\E\left( AA'\right) \right] \left( V^{1/2}\right) ^{\prime }\Sigma _{XZ}^{-1} \\
&\quad+O\left(h^r\right)+o\left(h+nh^{2r}\right),
\end{align*}%
where 
\begin{equation*}
\Sigma _{ZX} = \E\left[ Z_{j}X_{j}^{\prime }f_{U|Z,X}(0\mid Z_{j},X_{j})\right] 
\text{ and }\Sigma _{XZ}=\Sigma _{ZX}^{\prime }.
\end{equation*}

The first term of $\mathrm{AMSE}_{\beta }$ is the asymptotic variance of the
unsmoothed QR estimator. The second term captures the higher-order effect of
smoothing on the AMSE of $\sqrt{n}(\hat{\beta}-\beta _{0})$. When $%
nh^{r}\rightarrow \infty $ and $n^{3}h^{4r+1}\rightarrow \infty$, we have $%
h^{r}=o\left( nh^{2r}\right) $ and $1/\sqrt{nh}=o\left( nh^{2r}\right) $, so
the terms of order $O_{p}(1/\sqrt{nh})$ in \eqref{eqn:expansion1} and of
order $O\left( h^{r}\right) $ in \eqref{eqn:expansion2} are of smaller order
than the $O(nh^{2r})$ and $O(h)$ terms in the AMSE. If $h\asymp
n^{-1/(2r-1)} $ as before, these rate conditions are satisfied when $r>2$.

\begin{theorem}
\label{thm:est-MSE} Let Assumptions \ref{a:sampling}--\ref{a:G}(i--ii), \ref%
{a:beta}, and \ref{a:power_mse} hold. If $nh^{r}\rightarrow \infty $ and $%
n^{3}h^{4r+1}\rightarrow \infty $, then the AMSE of $\sqrt{n}(\hat{\beta}%
-\beta _{0})$ is%
\begin{equation*}
\Sigma _{ZX}^{-1} V^{1/2} 
\left[ I_{d}+nh^{2r}\E(B)\E(B') - h\E\left( AA'\right) \right] 
\left(V^{1/2}\right) ^{\prime }\left( \Sigma _{ZX}^{\prime }\right)^{-1}
+O\left(h^r\right)+o\left(h+nh^{2r}\right).
\end{equation*}
\end{theorem}

The optimal $h^{\ast }$ that minimizes the high-order AMSE satisfies 
\begin{align*}
&\Sigma _{ZX}^{-1} V^{1/2} \left[ n\left( h^{\ast }\right) ^{2r}\E(B)\E(B') - h^{\ast }\E\left( AA'\right) \right]
\left( V^{1/2} \right)^\prime \left( \Sigma _{ZX}^{\prime }\right) ^{-1} \\
&\quad \leq \Sigma _{ZX}^{-1} V^{1/2} \left[ nh^{2r}\E(B)\E(B') - h\E\left( AA'\right) \right]
\left( V^{1/2} \right)^\prime \left( \Sigma _{ZX}^{\prime
}\right) ^{-1}
\end{align*}%
in the sense that the difference between the two sides is nonpositive
definite for all $h$. This is equivalent to 
\begin{equation*}
n\left( h^{\ast }\right) ^{2r}\E(B)\E(B') - h^{\ast }\E\left( AA'\right)
\leq nh^{2r}\E(B)\E(B') - h\E\left( AA'\right) .
\end{equation*}

This choice of $h$ can also be motivated from the theory of optimal
estimating equations. Given the estimating equations $m_{n}=0$, we follow %
\citet{Heyde1997} and define the standardized version of $m_{n}$ by 
\begin{equation*}
m_{n}^{s}(\beta _{0},h)
  = -\E\left\{\frac{\partial }{\partial \beta ^{\prime }}%
m_{n}\left( \beta _{0}\right) \left[ \E(m_{n}m_{n}^{\prime })\right]
^{-1}m_{n} \right\}.
\end{equation*}%
We include $h$ as an argument of $m_{n}^{s}$ to emphasize the dependence of $%
m_{n}^{s}$ on $h$. The standardization can be motivated from the following
considerations. On one hand, the estimating equations need to be close to
zero when evaluated at the true parameter value. Thus we want $%
\E(m_{n}m_{n}^{\prime })$ to be as small as possible. On the other hand, we
want $m_{n}\left( \beta +\delta \beta \right) $ to differ as much as
possible from $m_{n}\left( \beta \right) $ when $\beta $ is the true value.
That is, we want $\E\frac{\partial }{\partial \beta ^{\prime }}m_{n}\left(
\beta _{0}\right) $ to be as large as possible. To meet these requirements,
we choose $h$ to maximize%
\begin{equation*}
\E\left\{m_{n}^{s}(\beta _{0},h)\left[ m_{n}^{s}\left( \beta _{0},h\right) \right]
^{\prime }\right\}
  = \left[ \E\frac{\partial }{\partial \beta ^{\prime }}m_{n}\left(
\beta _{0}\right) \right] \left[ \E(m_{n}m_{n}^{\prime })\right] ^{-1}\left[ \E%
\frac{\partial }{\partial \beta ^{\prime }}m_{n}\left( \beta _{0}\right) %
\right] ^{\prime }.
\end{equation*}%
More specifically, $h^{\ast }$ is optimal if 
\begin{equation*}
\E\left\{m_{n}^{s}(\beta _{0},h^{\ast })\left[ m_{n}^{s}\left( \beta _{0},h^{\ast
}\right) \right] ^{\prime }\right\}
  - \E\left\{m_{n}^{s}(\beta _{0},h)\left[
m_{n}^{s}\left( \beta _{0},h\right) \right] ^{\prime }\right\}
\end{equation*}%
is nonnegative definite for all $h\in \mathbb{R}^{+}$. But $%
\E\left[m_{n}^{s}\left( m_{n}^{s}\right) ^{\prime }\right]=\left( \mathrm{AMSE}_{\beta
}\right) ^{-1}$, so maximizing $\E\left[m_{n}^{s}\left( m_{n}^{s}\right) ^{\prime
}\right]$ is equivalent to minimizing $\mathrm{AMSE}_{\beta }$.

The question is whether such an optimal $h$ exists. If it does, then the
optimal $h^{\ast }$ satisfies 
\begin{equation}  \label{AMSE_obj}
h^{\ast } = \mathop{\rm arg\,min}_{h} u^{\prime }\left[ nh^{2r}\E(B)\E(B') - h\E\left( AA'\right) \right] u
\end{equation}%
for all $u\in \mathbb{R}^{d}$, by the definition of nonpositive definite
plus the fact that the above yields a unique minimizer for any $u$. Using
unit vectors $e_1=(1,0,\ldots,0)$, $e_2=(0,1,0,\ldots,0)$, etc.,\ for $u$,
and noting that $\mathrm{tr}(A)=e_1^{\prime
}Ae_1+\cdots+e_d^{\prime }Ae_d$ for $d\times d$ matrix $A$, this implies
that 
\begin{align*}
h^{\ast } &= \mathop{\rm arg\,min}_{h}\mathrm{tr}\left[ nh^{2r}\E(B)\E(B') - h\E\left( AA'\right) \right] \\
&= \mathop{\rm arg\,min}_{h}\left\{ nh^{2r}\E(B)'\E(B) - h\mathrm{tr}\left[\E\left( AA'\right) \right] \right\} .
\end{align*}%
In view of \eqref{eqn:def_h_SEE}, $h_\text{SEE}^*=h^*$ if $h^*$ exists.
Unfortunately, it is easy to show that no single $h$ can minimize the
objective function in \eqref{AMSE_obj} for all $u\in \mathbb{R}^{d}$. Thus,
we have to redefine the optimality with respect to the direction of $u$. The
direction depends on which linear combination of $\beta $ is the focus of
interest, as $u^{\prime }\left[ nh^{2r}\E(B)\E(B') - h\E\left( AA'\right) \right] u$ is the high-order AMSE of 
$c^{\prime }\sqrt{n}(\hat{\beta}-\beta _{0})$ for $c=\Sigma_{XZ} \left(
V^{-1/2}\right)^{\prime} u$.

Suppose we are interested in only one linear combination. Let $h_{c}^{\ast }$
be the optimal $h$ that minimizes the high-order AMSE of $c^{\prime }\sqrt{n}%
(\hat{\beta}-\beta _{0})$. Then 
\begin{equation*}
h_{c}^{\ast }=\left( \frac{u^{\prime }\E\left( AA^{\prime }\right) u}{%
u^{\prime }\E(B)\E(B')u}\frac{1}{2nr}%
\right) ^{\frac{1}{2r-1}}
\end{equation*}%
for $u=\left( V^{1/2}\right) ^{\prime } \Sigma_{XZ}^{-1} c$. Some algebra
shows that 
\begin{equation*}
h_{c}^{\ast }\geq \left( \frac{1}{\E(B)'\left[
\E\left(AA'\right)\right]^{-1}\E(B)}\frac{1}{2nr}\right) ^{\frac{1}{2r-1}}>0.
\end{equation*}%
So although $h_{c}^{\ast }$ depends on $c$ via $u$, it is nevertheless
greater than zero.

Now suppose without loss of generality we are interested in $d$ directions $%
\left( c_{1},\ldots ,c_{d}\right) $ jointly where $c_{i}\in \mathbb{R}^{d}$.
In this case, it is reasonable to choose $h_{c_{1},\ldots ,c_{d}}^{\ast }$
to minimize the sum of direction-wise AMSEs, i.e.,
\begin{equation*}
h_{c_{1},\ldots ,c_{d}}^{\ast }=\mathop{\rm arg\,min}_{h}%
\sum_{i=1}^{d}u_{i}^{\prime }\left[ nh^{2r}\E(B)\E(B') - h\E\left( AA'\right) \right] u_{i},
\end{equation*}%
where $u_{i}=\left( V^{1/2}\right) ^{\prime }\Sigma _{XZ}^{-1}c_{i}$. It is
easy to show that%
\begin{equation*}
h_{c_{1},\ldots ,c_{d}}^{\ast }
  = \left[ \frac{\sum_{i=1}^{d}u_{i}^{\prime
}\E\left( AA'\right) u_{i}}{\sum_{i=1}^{d}u_{i}^{\prime }\E(B)\E(B')u_{i}}\frac{1}{2nr}\right]^{\frac{1}{%
2r-1}}.
\end{equation*}

As an example, consider $u_{i}=e_{i}=\left( 0,\ldots ,1,\ldots ,0\right) $,
the $i$th unit vector in $\mathbb{R}^{d}$. Correspondingly%
\begin{equation*}
\left( \tilde{c}_{1},\ldots,\tilde{c}_{d}\right) = \Sigma_{XZ} \left(
V^{-1/2}\right)^{\prime } \left( e_{1},\ldots,e_{d}\right) .
\end{equation*}
It is clear that%
\begin{equation*}
h_{\tilde{c}_{1},\ldots ,\tilde{c}_{d}}^{\ast }=h_{\text{SEE}}^{\ast }=h_{%
\text{CPE}}^{\ast } ,
\end{equation*}%
so all three selections coincide with each other. A special case of interest
is when $Z=X$, non-constant regressors are pairwise independent and
normalized to mean zero and variance one, and $U%
\mathpalette{\protect
\independenT}{\perp} X$. Then $u_{i}=c_{i}=e_{i}$ and the $d$ linear
combinations reduce to the individual elements of $\beta $.

The above example illustrates the relationship between $h_{c_{1},\ldots
,c_{d}}^{\ast }$ and $h_{\text{SEE}}^{\ast }$. While $h_{c_{1},\ldots
,c_{d}}^{\ast }$ is tailored toward the flexible linear combinations $%
\left(c_{1},\ldots ,c_{d}\right)$ of the parameter vector, $h_{\text{SEE}%
}^{\ast }$ is tailored toward the fixed $\left( \tilde{c}_{1},\ldots ,\tilde{%
c}_{d}\right) $. While $h_{c_{1},\ldots ,c_{d}}^{\ast }$ and $h_{\text{SEE}%
}^{\ast }$ are of the same order of magnitude, in general there is no
analytic relationship between $h_{c_{1},\ldots ,c_{d}}^{\ast }$ and $h_{%
\text{SEE}}^{\ast }$.

To shed further light on the relationship between $h_{c_{1},\ldots
,c_{d}}^{\ast }$ and $h_{\text{SEE}}^{\ast }$, let $\left \{ \lambda
_{k},k=1,\ldots ,d\right \} $ be the eigenvalues of $nh^{2r}\E(B)\E(B') - h\E\left( AA'\right)$ with the
corresponding orthonormal eigenvectors $\left \{ \ell _{k},k=1,\ldots
,d\right \} $. Then we have $nh^{2r}\E(B)\E(B') - h\E\left( AA'\right) =\sum_{k=1}^{d}\lambda _{k}\ell
_{k}\ell _{k}^{\prime }$ and $u_{i}=\sum_{j=1}^{d}u_{ij}\ell _{j}$ for $%
u_{ij}=u_{i}^{\prime }\ell _{j}$. Using these representations, the objective
function underlying $h_{c_{1},\ldots ,c_{d}}^{\ast }$ becomes 
\begin{align*}
& \sum_{i=1}^{d}u_{i}^{\prime }\left[ nh^{2r}\E(B)\E(B') - h\E\left( AA'\right)\right] u_{i} \\
& =\sum_{i=1}^{d}\left( \sum_{j=1}^{d}u_{ij}\ell _{j}^{\prime }\right)
\left( \sum_{k=1}^{d}\lambda _{k}\ell _{k}\ell _{k}^{\prime }\right) \left(
\sum_{\tilde{j}=1}^{d}u_{i\tilde{j}}\ell _{\tilde{j}}\right) \\
& =\sum_{j=1}^{d}\left( \sum_{i=1}^{d}u_{ij}^{2}\right) \lambda _{j}.
\end{align*}%
That is, $h_{c_{1},\ldots ,c_{d}}^{\ast }$ minimizes a weighted sum of the
eigenvalues of $nh^{2r}\E(B)\E(B') - h\E\left( AA'\right)$ with weights depending on $c_{1},\ldots ,c_{d}$. By
definition, $h_{\text{SEE}}^{\ast }$ minimizes the simple unweighted sum of
the eigenvalues, viz.\ $\sum_{j=1}^{d}\lambda _{j}$. While $h_{\text{SEE}%
}^{\ast }$ may not be ideal if we know the linear combination(s) of
interest, it is a reasonable choice otherwise.

In empirical applications, we can estimate $h_{c_{1},\ldots ,c_{d}}^{\ast }$
using a parametric plug-in approach similar to our plug-in implementation of 
$h_{\text{SEE}}^{\ast }$. If we want to be agnostic about the directional
vectors $c_{1},\ldots ,c_{d}$, we can simply use $h_{\text{SEE}}^{\ast }$.


\section{Empirical example: JTPA\label{sec:emp}} 

We revisit the IV-QR analysis of Job Training Partnership Act (JTPA) data in %
\citet{AbadieEtAl2002}, specifically their Table III.\footnote{%
Their data and Matlab code for replication are helpfully provided online in
the Angrist Data Archive, %
\url{http://economics.mit.edu/faculty/angrist/data1/data/abangim02}.} They
use 30-month earnings as the outcome, randomized offer of JTPA services as
the instrument, and actual enrollment for services as the endogenous
treatment variable. Of those offered services, only around 60 percent
accepted, so self-selection into treatment is likely. Section 4 of %
\citet{AbadieEtAl2002} provides much more background and descriptive
statistics.

We compare estimates from a variety of methods.\footnote{Code and data for replication is available on the first author's website.}
 ``AAI'' is the original paper's
estimator. AAI restricts $X$ to have finite support (see condition (iii) in
their Theorem 3.1), which is why all the regressors in their example are
binary. Our fully automated plug-in estimator is ``SEE ($\hat{h}$).''
``CH'' is \citet{ChernozhukovHansen2006}. Method ``tiny $h$'' uses $h=400$ (compared with our plug-in values on the order of $10\,000$), while ``huge $h$'' uses $h=5\times 10^{6}$. 2SLS is the usual (mean) two-stage least squares estimator, put in the $q=0.5$ column only for convenience of comparison.

\begin{table}[htbp]
\centering
\caption{\label{tab:emp}IV-QR estimates of coefficients for certain regressors as in Table III of \citet{AbadieEtAl2002} for adult men.}
\begin{tabular}[c]{ccS[table-format=9.0]S[table-format=8.0]S[table-format=6.0]S[table-format=7.0]S[table-format=8.0]}
\hline\hline
 &  &  \multicolumn{5}{c}{Quantile} \\
 \cline{3-7}
Regressor    &     Method     &   \multicolumn{1}{c}{0.15}  &   \multicolumn{1}{c}{0.25}  &   \multicolumn{1}{c}{0.50}  &   \multicolumn{1}{c}{0.75}  &   \multicolumn{1}{c}{0.85}  \\
\hline
Training     & AAI            &      121 &      702 &  1544 &  3131 &  3378 \\
Training     & SEE ($\hat h$) &       57 &      381 &  1080 &  2630 &  2744 \\
Training     & CH             &     -125 &      341 &      385 &  2557 &  3137 \\
Training     & tiny $h$       &     -129 &      500 &      381 &  2760 &  3114 \\
Training     & huge $h$       &  1579 &  1584 &  1593 &  1602 &  1607 \\
Training     & 2SLS           &          &          &  1593 &          &          \\
HS or GED    & AAI            &      714 &  1752 &  4024 &  5392 &  5954 \\
HS or GED    & SEE ($\hat h$) &      812 &  1498 &  3598 &  6183 &  6753 \\
HS or GED    & CH             &      482 &  1396 &  3761 &  6127 &  6078 \\
HS or GED    & tiny $h$       &      463 &  1393 &  3767 &  6144 &  6085 \\
HS or GED    & huge $h$       &  4054 &  4062 &  4075 &  4088 &  4096 \\
HS or GED    & 2SLS           &          &          &  4075 &          &          \\
Black        & AAI            &     -171 &     -377 & -2656 & -4182 & -3523 \\
Black        & SEE ($\hat h$) &     -202 &     -546 & -1954 & -3273 & -3653 \\
Black        & CH             &      -38 &     -109 & -2083 & -3233 & -2934 \\
Black        & tiny $h$       &      -18 &     -139 & -2121 & -3337 & -2884 \\
Black        & huge $h$       & -2336 & -2341 & -2349 & -2357 & -2362 \\
Black        & 2SLS           &          &          & -2349 &          &          \\
Married      & AAI            &  1564 &  3190 &  7683 &  9509 & 10185 \\
Married      & SEE ($\hat h$) &  1132 &  2357 &  7163 & 10174 & 10431 \\
Married      & CH             &      504 &  2396 &  7722 & 10463 & 10484 \\
Married      & tiny $h$       &      504 &  2358 &  7696 & 10465 & 10439 \\
Married      & huge $h$       &  6611 &  6624 &  6647 &  6670 &  6683 \\
Married      & 2SLS           &          &          &  6647 &          &          \\
Constant     & AAI            &     -134 &  1049 &  7689 & 14901 & 22412 \\
Constant     & SEE ($\hat h$) &      -88 &  1268 &  7092 & 15480 & 22708 \\
Constant     & CH             &      242 &  1033 &  7516 & 14352 & 22518 \\
Constant     & tiny $h$       &      294 &  1000 &  7493 & 14434 & 22559 \\
Constant     & huge $h$       & -1157554    & -784046 &  10641 & 805329 & 1178836 \\
Constant     & 2SLS           &          &          & 10641 &          &          \\
\hline
\end{tabular}
\end{table}

Table \ref{tab:emp} shows results from the sample of $5102$ adult men, for a
subset of the regressors used in the model. Not shown in the table are
coefficient estimates for dummies for Hispanic, working less than 13 weeks
in the past year, five age groups, originally recommended service strategy,
and whether earnings were from the second follow-up survey. CH is very close
to ``tiny $h$''; that is,
simply using the smallest possible $h$ with SEE provides a good
approximation of the unsmoothed estimator in this case. Demonstrating our theoretical results in Section \ref{sec:see-comp-IV}, ``huge $h$'' is very close to
2SLS for everything except the constant term for $q\neq 0.5$. The IVQR-SEE estimator using
our plug-in bandwidth has some economically significant differences with the
unsmoothed estimator. Focusing on the treatment variable (``Training''), the unsmoothed median effect estimate is below 
$400$ (dollars), whereas SEE$(\hat{h})$ yields $1080$, both of
which are smaller than AAI's $1544$ (AAI is the most positive at all
quantiles). For the $0.15$-quantile effect, the unsmoothed estimates are
actually slightly negative, while SEE$(\hat{h})$ and AAI are
slightly positive. 
For $q=0.85$, though, the SEE$(\hat{h})$ estimate is smaller than
the unsmoothed one, and the two are quite similar for $q=0.25$ and $q=0.75$;
there is no systematic ordering.

Computationally, our code takes only one second total to calculate the plug-in bandwidths and coefficient estimates at all five quantiles.  Using the fixed $h=400$ or $h=5\times10^6$, computation is immediate. 

\begin{table}[htbp]
\centering
\caption{\label{tab:emp2}IV-QR estimates similar to Table \ref{tab:emp}, but replacing age dummies with a quartic polynomial in age and adding baseline measures of weekly hours worked and wage.}
\begin{tabular}[c]{ccS[table-format=4.0]S[table-format=4.0]S[table-format=4.0]S[table-format=4.0]S[table-format=4.0]}
\hline\hline
 &  &  \multicolumn{5}{c}{Quantile} \\
 \cline{3-7}
Regressor    &     Method     &   \multicolumn{1}{c}{0.15}  &   \multicolumn{1}{c}{0.25}  &   \multicolumn{1}{c}{0.50}  &   \multicolumn{1}{c}{0.75}  &   \multicolumn{1}{c}{0.85}  \\
\hline\rule{0pt}{12pt}
Training     & SEE ($\hat h$) &       74 &      398 &  1045 &  2748 &  2974 \\
Training     & CH             &      -20 &      451 &      911 &  2577 &  3415 \\
Training     & tiny $h$       &      -50 &      416 &      721 &  2706 &  3555 \\
Training     & huge $h$       &  1568 &  1573 &  1582 &  1590 &  1595 \\
Training     & 2SLS           &          &          &  1582 &          &          \\
\hline
\end{tabular}
\end{table}

Table \ref{tab:emp2} shows estimates of the endogenous coefficient when various ``continuous'' control variables are added, specifically a quartic polynomial in age (replacing the age range dummies), baseline weekly hours worked, and baseline hourly wage.\footnote{Additional JTPA data downloaded from the W.E.\ Upjohn Institute at \url{http://upjohn.org/services/resources/employment-research-data-center/national-jtpa-study}; variables are named \texttt{age}, \texttt{bfhrswrk}, and \texttt{bfwage} in file \texttt{expbif.dta}.}  The estimates do not change much; the biggest difference is for the unsmoothed estimate at the median.  Our code again computes the plug-in bandwidth and SEE coefficient estimates at all five quantiles in one second.  Using the small $h=400$ bandwidth now takes nine seconds total (more iterations of \texttt{fsolve} are needed); $h=5\times10^6$ still computes almost immediately.


\section{Simulations\label{sec:sim}}

For our simulation study,\footnote{%
Code to replicate our simulations is available on the first author's
website.} we use $G(u)$ given in \eqref{eqn:G_fun} as in %
\citet{Horowitz1998} and \citet{Whang2006}. This satisfies Assumption \ref%
{a:G} with $r=4$. Using (the integral of) an Epanechnikov
kernel with $r=2$ also worked well in the cases we consider here, though
never better than $r=4$. Our error distributions always have at
least four derivatives, so $r=4$ working somewhat better is expected.
Selection of optimal $r$ and $G(\cdot)$, and the quantitative impact
thereof, remain open questions.

We implement a plug-in version ($\hat h$) of the infeasible $h^{\ast }\equiv h_{\text{%
SEE}}^{\ast }$. We make the plug-in assumption $U%
\mathpalette{\protect
\independenT}{\perp} Z$ and parameterize the distribution of $U$. Our
current method, which has proven quite accurate and stable, fits the
residuals from an initial $h=(2nr)^{-1/(2r-1)}$ IV-QR to Gaussian, $t$,
gamma, and generalized extreme value distributions via maximum likelihood.
With the distribution parameter estimates, $f_{U}(0)$ and $f_{U}^{(r-1)}(0)$
can be computed and plugged in to calculate $\hat{h}$.  With larger $n$, a nonparametric kernel estimator may perform better, but nonparametric estimation of $f_U^{(r-1)}(0)$ will have high variance in smaller samples.  
Viewing the unsmoothed estimator as a reference point, 
potential regret (of using $\hat h$ instead of $h=0$) is largest when $\hat h$ is too large, 
so we separately calculate $\hat{h}$
for each of the four distributions and take the smallest. Note that this
particular plug-in approach works well even under heteroskedasticity and/or
misspecification of the error distribution: DGPs 3.1--3.6 in Section \ref{sec:sim-additional} have error
distributions other than these four, and DGPs 1.3, 2.2, 3.3--3.6 are
heteroskedastic, as are the JTPA-based simulations. For the infeasible $h^{\ast }$, if the PDF derivative in
the denominator is zero, it is replaced by $0.01$ to avoid $h^{\ast }=\infty 
$.

For the unsmoothed IV-QR estimator, we use code based on %
\citet{ChernozhukovHansen2006} from the latter author's website. 
We use the option to let their code determine the
grid of possible endogenous coefficient values from the data. 
This code in turn uses the interior point method in \texttt{rq.m} (developed
by Roger Koenker, Daniel Morillo, and Paul Eilers) to solve exogenous QR
linear programs. 

\subsection{JTPA-based simulations\label{sec:sim-JTPA}}

We use two DGPs based on the JTPA data examined in Section \ref%
{sec:emp}. The first DGP corresponds to the variables used in the original
analysis in \citet{AbadieEtAl2002}. For individual $i$, let $Y_{i}$ be the
scalar outcome (30-month earnings), $X_{i}$ be the vector of exogenous
regressors, $D_{i}$ be the scalar endogenous training dummy, $Z_{i}$ be the
scalar instrument of randomized training offer, and $U_{i}\sim \text{Unif}%
(0,1)$ be a scalar unobservable term. We draw $X_{i}$ from the joint
distribution estimated from the JTPA data. We randomize $Z_{i}=1$ with
probability $0.67$ and zero otherwise. If $Z_{i}=0$, then we set the
endogenous training dummy $D_{i}=0$ (ignoring that in reality, a few percent
still got services). If $Z_{i}=1$, we set $D_{i}=1$ with a probability
increasing in $U_{i}$. Specifically, $P(D_{i}=1\mid Z_{i}=1,U_{i}=u)=\min
\{1,u/0.75\}$, which roughly matches the $P(D_{i}=1\mid Z_{i}=1)=0.62$ in
the data. This corresponds to a high degree of self-selection into treatment
(and thus endogeneity). Then, $Y_{i}=X_{i}\beta _{X}+D_{i}\beta
_{D}(U_{i})+G^{-1}(U_{i})$, where $\beta _{X}$ is the IVQR-SEE $\hat{\beta}%
_{X}$ from the JTPA data (rounded to the nearest $500$), the function $\beta
_{D}(U_{i})=2000U_{i}$ matches $\hat{\beta}_{D}(0.5)$ and the increasing
pattern of other $\hat{\beta}_{D}(q)$, and $G^{-1}(\cdot )$ is a recentered
gamma distribution quantile function with parameters estimated to match the
distribution of residuals from the IVQR-SEE estimate with JTPA data.  
In each of $1000$ simulation replications, we generate $n=5102$ iid observations. 

For the second DGP, we add a second endogenous regressor (and instrument)
and four exogenous regressors, all with normal distributions. Including the
intercept and two endogenous regressors, there are $20$ regressors. The
second instrument is $Z_{2i}\overset{iid}{\sim}N(0,1)$, and the second
endogenous regressor is $D_{2i}=0.8Z_{2i}+0.2\Phi^{-1}(U_i)$. The
coefficient on $D_{2i}$ is $1000$ at all quantiles. The new exogenous
regressors are all standard normal and have coefficients of $500$ at all
quantiles. To make the asymptotic bias of 2SLS relatively more important,
the sample size is increased to $n=50\,000$.

\begin{table}[htbp]
\centering
\caption{\label{tab:sim-JTPA1}Simulation results for endogenous coefficient estimators with first JTPA-based DGP.  ``Robust MSE'' is squared median-bias plus the square of the interquartile range divided by $1.349$, $\textrm{Bias}_{\textrm{median}}^{2}+(\textrm{IQR}/1.349)^2$; it is shown in units of $10^5$, so $7.8$ means $7.8\times10^5$, for example.  ``Unsmoothed'' is the estimator from \citet{ChernozhukovHansen2006}.}
\begin{tabular}[c]{cS[table-format=2.1]S[table-format=2.1]S[table-format=2.1]cS[table-format=5.1]S[table-format=4.1]S[table-format=5.1]}
\hline\hline\rule{0pt}{12pt}
 & \multicolumn{3}{c}{Robust MSE / $10^5$} & & \multicolumn{3}{c}{Median Bias} \\ 
  \cline{2-4}\cline{6-8} \rule{0pt}{14pt}
$q$ & \multicolumn{1}{c}{Unsmoothed} & \multicolumn{1}{c}{SEE ($\hat h$)} & \multicolumn{1}{c}{2SLS} & &
      \multicolumn{1}{c}{Unsmoothed} & \multicolumn{1}{c}{SEE ($\hat h$)} & \multicolumn{1}{c}{2SLS} \\
\hline 
$0.15$ &  78.2 &  43.4 &  18.2  &&  -237.6 &   8.7 & 1040.6 \\
$0.25$ &  30.5 &  18.9 &  14.4  &&  -122.2 &  16.3 & 840.6 \\
$0.50$ &   9.7 &   7.8 &   8.5  &&   24.1 &  -8.5 & 340.6 \\
$0.75$ &   7.5 &   7.7 &   7.6  &&   -5.8 & -48.1 & -159.4 \\
$0.85$ &  11.7 &   9.4 &   8.6  &&   49.9 & -17.7 & -359.4 \\
\hline
\end{tabular}
\end{table}

Table \ref{tab:sim-JTPA1} shows results for the first JTPA-based DGP, for three estimators of the endogenous coefficient: \citet{ChernozhukovHansen2006}; SEE with our data-dependent $\hat h$; and 2SLS.  The first and third can be viewed as limits of IVQR-SEE estimators as $h\to0$ and $h\to\infty$, respectively.  We show median bias and ``robust MSE,'' which is squared median bias plus the square of the interquartile range divided by $1.349$, $\textrm{Bias}_{\textrm{median}}^{2}+(\textrm{IQR}/1.349)^2$.  We report these ``robust'' versions of bias and MSE since the (mean) IV estimator does not even possess a first moment in finite samples \citep{Kinal1980}.  We are unaware of an analogous result for IV-QR but remain wary of presenting bias and MSE results for IV-QR, too, especially since the IV estimator is the limit of the SEE IV-QR estimator as $h\to\infty$.  At all quantiles, for all methods, the robust MSE is dominated by the IQR rather than bias.  Consequently, even though the 2SLS median bias is quite large for $q=0.15$, it has less than half the robust MSE of SEE($\hat h$), which in turn has half the robust MSE of the unsmoothed estimator.  With only a couple exceptions, this is the ordering among the three methods' robust MSE at all quantiles.  Although the much larger bias of 2SLS than that of SEE($\hat h$) or the unsmoothed estimator is expected, the smaller median bias of SEE($\hat h$) than that of the unsmoothed estimator is surprising.  However, the differences are not big, and they may be partly due to the much larger variance of the unsmoothed estimator inflating the simulation error in the simulated median bias, especially for $q=0.15$.  The bigger difference is the reduction in variance from smoothing. 

\begin{table}[htbp]
\centering
\caption{\label{tab:sim-JTPA2}Simulation results for endogenous coefficient estimators with second JTPA-based DGP.}
\begin{tabular}[c]{cS[table-format=6.0]S[table-format=6.0]S[table-format=7.0]cS[table-format=4.1]S[table-format=4.1]S[table-format=5.1]}
\hline\hline\rule{0pt}{12pt}
 & \multicolumn{3}{c}{Robust MSE} & & \multicolumn{3}{c}{Median Bias} \\ 
  \cline{2-4}\cline{6-8} \rule{0pt}{14pt}
$q$ & \multicolumn{1}{c}{$h=400$} & \multicolumn{1}{c}{SEE ($\hat h$)} & \multicolumn{1}{c}{2SLS} & &
      \multicolumn{1}{c}{$h=400$} & \multicolumn{1}{c}{SEE ($\hat h$)} & \multicolumn{1}{c}{2SLS} \\
\hline 
\multicolumn{8}{c}{\textit{Estimators of binary endogenous regressor's coefficient}} \\
$0.15$ &   780624 &   539542 &  1071377  &&  -35.7 &  10.6 & 993.6 \\
$0.25$ &   302562 &   227508 &   713952  &&  -18.5 &  17.9 & 793.6 \\
$0.50$ &   101433 &    96350 &   170390  &&  -14.9 & -22.0 & 293.6 \\
$0.75$ &    85845 &    90785 &   126828  &&   -9.8 & -22.5 & -206.4 \\
$0.85$ &   147525 &   119810 &   249404  &&  -15.7 & -17.4 & -406.4 \\
\multicolumn{8}{c}{\textit{Estimators of continuous endogenous regressor's coefficient}} \\
$0.15$ &     9360 &     7593 &    11434  &&   -3.3 &  -5.0 &  -7.0 \\
$0.25$ &    10641 &     9469 &    11434  &&   -3.4 &  -3.5 &  -7.0 \\
$0.50$ &    13991 &    12426 &    11434  &&   -5.7 &  -9.8 &  -7.0 \\
$0.75$ &    28114 &    25489 &    11434  &&  -12.3 & -17.9 &  -7.0 \\
$0.85$ &    43890 &    37507 &    11434  &&  -17.2 & -17.8 &  -7.0 \\
\hline
\end{tabular}
\end{table}

Table \ref{tab:sim-JTPA2} shows results from the second JTPA-based DGP.  The first estimator is now a nearly-unsmoothed SEE estimator instead of the unsmoothed \citet{ChernozhukovHansen2006} estimator.  Although in principle \citet{ChernozhukovHansen2006} can be used with multiple endogenous coefficients, the provided code allows only one, and Tables \ref{tab:emp} and \ref{tab:emp2} show that SEE with $h=400$ produces very similar results in the JTPA data.  For the binary endogenous regressor's coefficient, the 2SLS estimator now has the largest robust MSE since the larger sample size reduces the variance of all three estimators but does not reduce the 2SLS median bias (since it has first-order asymptotic bias).  The plug-in bandwidth yields smaller robust MSE than the nearly-unsmoothed $h=400$ at four of five quantiles.  At the median, for example, compared with $h=400$, $\hat h$ slightly increases the median bias but greatly reduces the dispersion, so the net effect is to reduce robust MSE.  This is consistent with the theoretical results.  For the continuous endogenous regressor's coefficient, the same pattern holds for the nearly-unsmoothed and $\hat h$-smoothed estimators.  Since this coefficient is constant across quantiles, the 2SLS estimator is consistent and very similar to the SEE estimators with $q=0.5$.

\subsection{Comparison of SEE and smoothed criterion function\label{sec:sim-SCF}}

For exogenous QR, smoothing the criterion function (SCF) is a different
approach, as discussed. The following simulations compare the MSE of our SEE
estimator with that of the SCF estimator. All DGPs have $n=50$, $X_{i}%
\overset{iid}{\sim }\text{Unif}(1,5)$, $U_{i}\overset{iid}{\sim }N(0,1)$, $%
X_{i}\mathpalette{\protect \independenT}{\perp}U_{i}$, and $%
Y_{i}=1+X_{i}+\sigma (X_{i})\left( U_{i}-\Phi ^{-1}(q)\right) $. DGP 1 has $%
q=0.5$ and $\sigma (X_{i})=5$. DGP 2 has $q=0.25$ and $\sigma
(X_{i})=1+X_{i} $. DGP 3 has $q=0.75$ and $\sigma (X_{i})=1+X_{i}$. In
addition to using our plug-in $\hat{h}$, we also compute the estimators for
a much smaller bandwidth in each DGP: $h=1$, $h=0.8$, and $h=0.8$,
respectively. Each simulation ran $1000$ replications. We
compare only the slope coefficient estimators.

\begin{table}[htbp]
\center
\caption{\label{tab:sim-SCF}Simulation results comparing SEE and SCF exogenous QR estimators.}
\begin{tabular}[c]{cS[table-format=1.3]S[table-format=1.3]cS[table-format=1.3]S[table-format=1.3]cS[table-format=2.3]S[table-format=2.3]cS[table-format=2.3]S[table-format=2.3]}
\hline\hline
 & \multicolumn{5}{c}{MSE} & & \multicolumn{5}{c}{Bias} \\
  \cline{2-6}\cline{8-12} 
\rule{0pt}{14pt}
 & \multicolumn{2}{c}{Plug-in $\hat h$} && \multicolumn{2}{c}{small $h$} && 
   \multicolumn{2}{c}{Plug-in $\hat h$} && \multicolumn{2}{c}{small $h$} \\
  \cline{2-3}\cline{5-6} \cline{8-9}\cline{11-12}
\rule{0pt}{12pt}
DGP & \multicolumn{1}{c}{SEE} & \multicolumn{1}{c}{SCF} &
    & \multicolumn{1}{c}{SEE} & \multicolumn{1}{c}{SCF} &
    & \multicolumn{1}{c}{SEE} & \multicolumn{1}{c}{SCF} &
    & \multicolumn{1}{c}{SEE} & \multicolumn{1}{c}{SCF} \\
\hline 
1 & 0.423 & 0.533 && 0.554 & 0.560 &&  -0.011 & -0.013 && -0.011 & -0.009 \\
2 & 0.342 & 0.433 && 0.424 & 0.430 &&  0.092 & -0.025 && 0.012 & 0.012 \\
3 & 0.146 & 0.124 && 0.127 & 0.121 &&  -0.292 & -0.245 && -0.250 & -0.232 \\
\hline
\end{tabular}
\end{table}

Table \ref{tab:sim-SCF} shows MSE and bias for the SEE and SCF estimators, for our plug-in $\hat h$ as well as the small, fixed $h$ mentioned above.  The SCF estimator can have slightly lower MSE, as in the third DGP ($q=0.75$ with heteroskedasticity), but the SEE estimator has more substantially lower MSE in more DGPs, including the homoskedastic conditional median DGP.  The differences are quite small with the small $h$, as expected.  Deriving and implementing an MSE-optimal bandwidth for the SCF estimator could shrink the differences, but based on these simulations and the theoretical comparison in Section \ref{sec:see}, such an effort seems unlikely to yield improvement over the SEE estimator.

\subsection{Additional simulations\label{sec:sim-additional}}

We tried additional data generating processes (DGPs).  The first three DGPs are for exogenous QR, taken directly from \citet{Horowitz1998}.  In each case, $q=0.5$, $Y_i=X_i'\beta_0+U_i$, $\beta_0=(1,1)'$, $X_i=(1,x_i)'$ with $x_i\stackrel{iid}{\sim}\textrm{Uniform}(1,5)$, and $n=50$.  In DGP 1.1, the $U_i$ are sampled iid from a $t_3$ distribution scaled to have variance two.  In DGP 1.2, the $U_i$ are iid from a type I extreme value distribution again scaled and centered to have median zero and variance two.  In DGP 1.3, $U_i=(1+x_i)V/4$ where $V_i\stackrel{iid}{\sim}N(0,1)$. 

DGPs 2.1, 2.2, and 3.1--3.6 are shown in the working paper version; they include variants of the \citet{Horowitz1998} DGPs with $q\ne0.5$, different error distributions, and another regressor. 

DGPs 4.1--4.3 have endogeneity.  DGP 4.1 has $q=0.5$, $n=20$, and $\beta _{0}=(0,1)^{\prime }$.  It uses the reduced form
equations in \citet[equation 2]{CattaneoEtAl2012} with $\gamma _{1}=\gamma
_{2}=1$, $x_{i}=1$, $z_{i}\sim N(0,1)$, and $\pi =0.5$. Similar to their
simulations, we set $\rho =0.5$, $(\tilde{v}_{1i},\tilde{v}_{2i})$ iid $%
N(0,1)$, and $(v_{1i},v_{2i})^{\prime }=\left(\tilde{v}_{1i},\sqrt{1-\rho ^{2}}%
\tilde{v}_{2i}+\rho \tilde{v}_{1i}\right)'$.  DGP 4.2 is similar to DGP 4.1 but with $(\tilde v_{1i},\tilde v_{2i})'$ iid Cauchy, $n=250$, and $\beta_0=\left(0,\left[\rho-\sqrt{1-\rho^2}\right]^{-1}\right)'$. 
DGP 4.3 is the same as DGP 4.1 but with $q=0.35$ (and consequent re-centering of the error term) and $n=30$. 

We compare MSE for our SEE estimator using the plug-in $\hat h$ and estimators using different (fixed) values of $h$.  We include $h=0$ by using unsmoothed QR or the method in \citet{ChernozhukovHansen2006} for the endogenous DGPs.  We also include $h=\infty$ (although not in graphs) by using the usual IV estimator.  For the endogenous DGPs, we consider both MSE and the ``robust MSE'' defined in Section \ref{sec:sim-JTPA} as $\textrm{Bias}_{\textrm{median}}^{2}+(\textrm{IQR}/1.349)^2$. 

For ``size-adjusted'' power (SAP) of a test with nominal size $\alpha$, the
critical value is picked as the $(1-\alpha)$-quantile of the simulated test
statistic distribution. This is for demonstration, not practice. The size
adjustment fixes the left endpoint of the size-adjusted power curve to the
null rejection probability $\alpha $. The resulting size-adjusted power
curve is one way to try to visualize a combination of type I and type II
errors, in the absence of an explicit loss function. One shortcoming is that
it does not reflect the variability/uniformity of size and power over the
space of parameter values and DGPs.

Regarding notation in the size-adjusted power figures, the vertical axis in the size-adjusted
power figures shows the simulated rejection probability. The horizontal axis
shows the magnitude of deviation from the null hypothesis, where a
randomized alternative is generated in each simulation iteration as that
magnitude times a random point on the unit sphere in $\mathbb{R}^{d}$, where 
$\beta \in \mathbb{R}^{d}$. As the legend shows, the dashed line corresponds
to the unsmoothed estimator ($h=0$), the dotted line to the infeasible $h_{%
\text{SEE}}^{\ast }$, and the solid line to the plug-in $\hat{h}$.

For the MSE graphs, the flat horizontal solid and dashed lines are the MSE
of the intercept and slope estimators (respectively) using feasible plug-in $%
\hat{h}$ (recomputed each replication). The other solid and dashed lines
(that vary with $h$) are the MSE when using the value of $h$ from the
horizontal axis. The left vertical axis shows the MSE values for the
intercept parameter; the right vertical axis shows the MSE for slope
parameter(s); and the horizontal axis shows a log transformation of the
bandwidth, $\log _{10}(1+h)$.

Our plug-in bandwidth is quite stable.  The range of $\hat{h}$ values over the
simulation replications is usually less than a factor of $10$, and the range
from $0.05$ to $0.95$ empirical quantiles is around a factor of two. This corresponds to
a very small impact on MSE; note the log transformation in the x-axis in the MSE 
graphs.

\begin{figure}[tbph]
\begin{center}
\includegraphics[clip=true,trim=25 175 145 200,width=.49\textwidth]
	{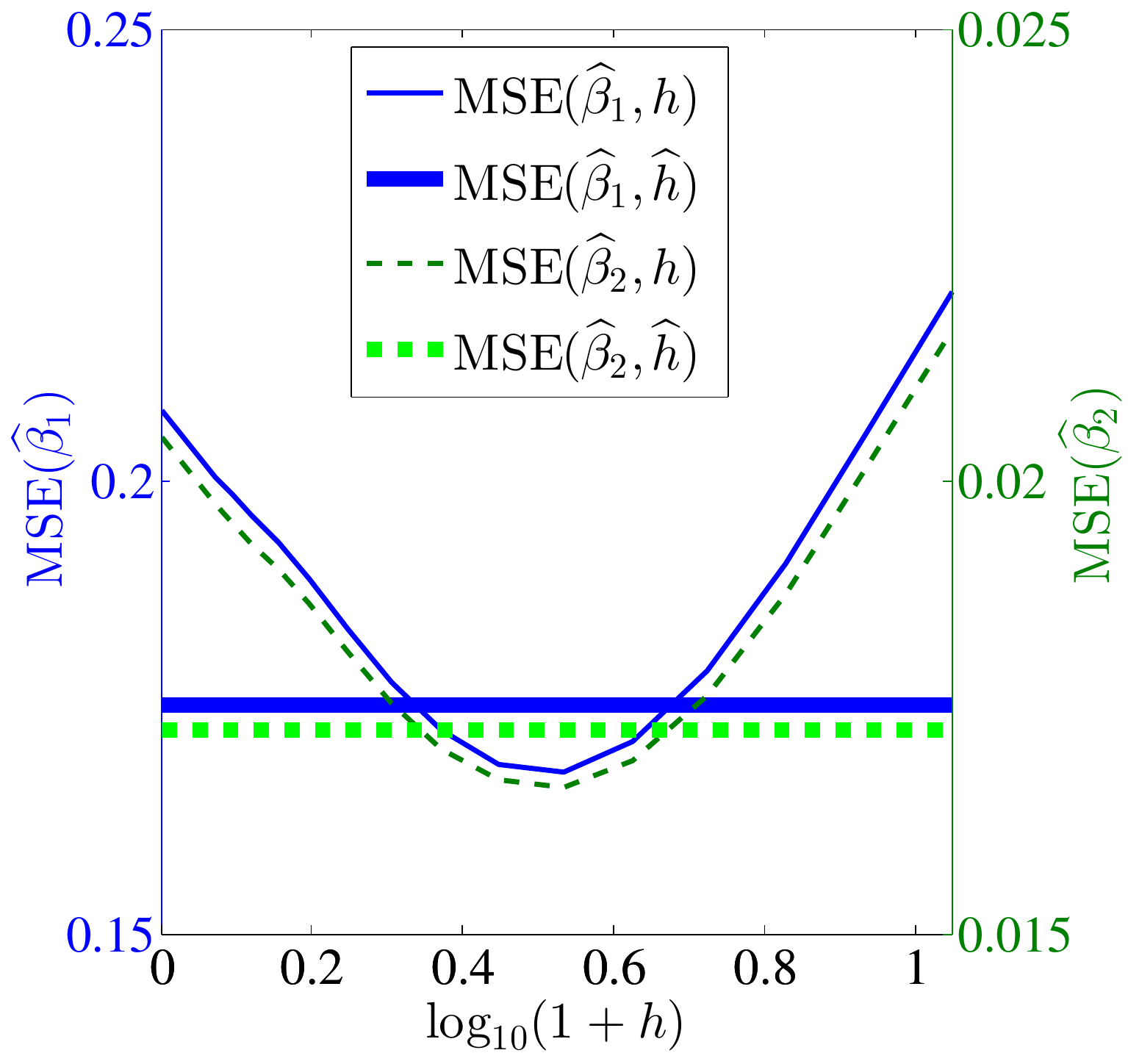}
	\hfill
\includegraphics[clip=true,trim=25 175 145 200,width=.49\textwidth]
	{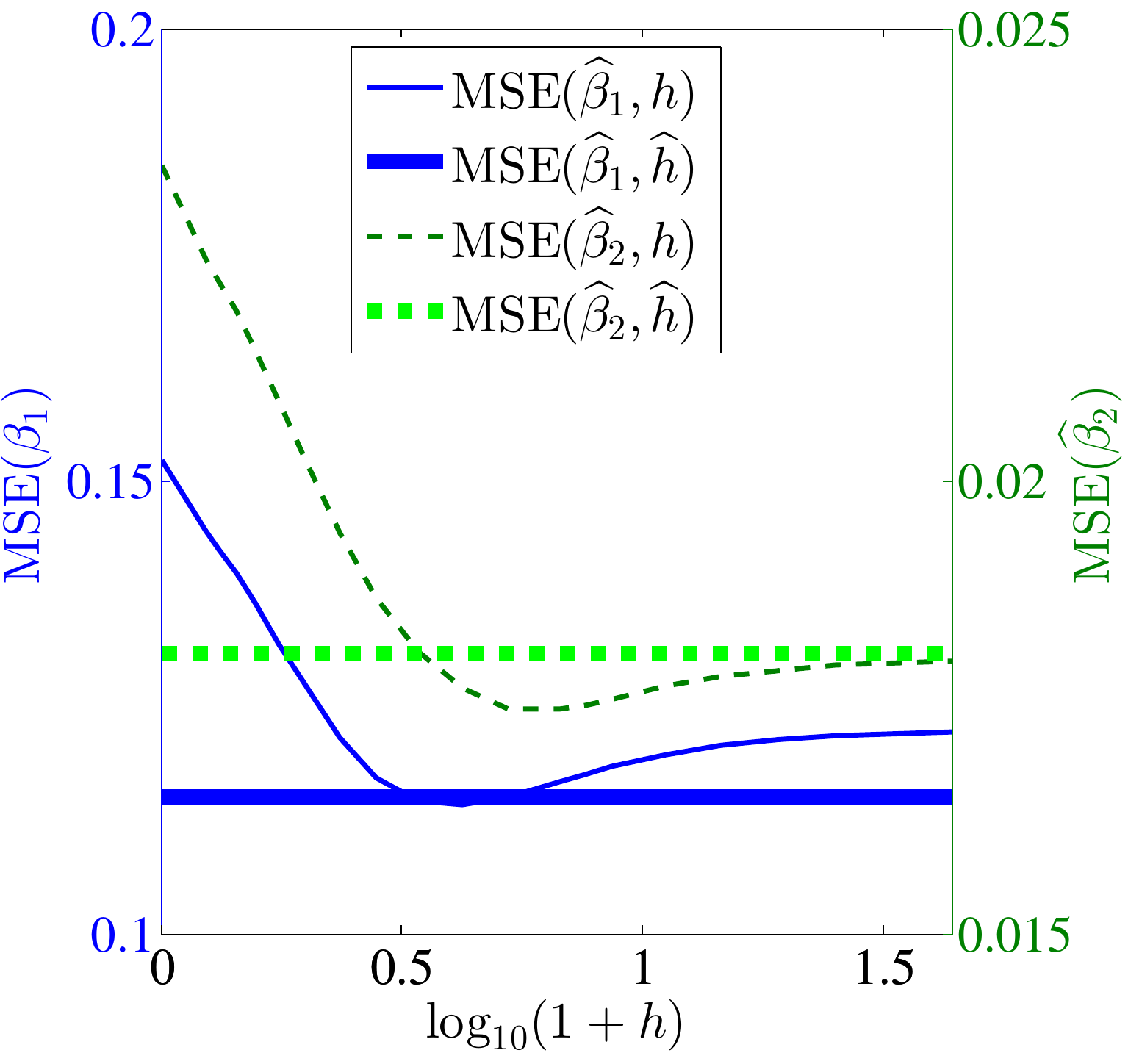}
\end{center}
\caption{MSE for DGPs 1.1 (left) and 1.3 (right).}
\label{fig:mse-1a}
\end{figure}
\begin{figure}[tbph]
\begin{center}
\includegraphics[clip=true,trim=30 175 145 200,width=.48\textwidth]
	{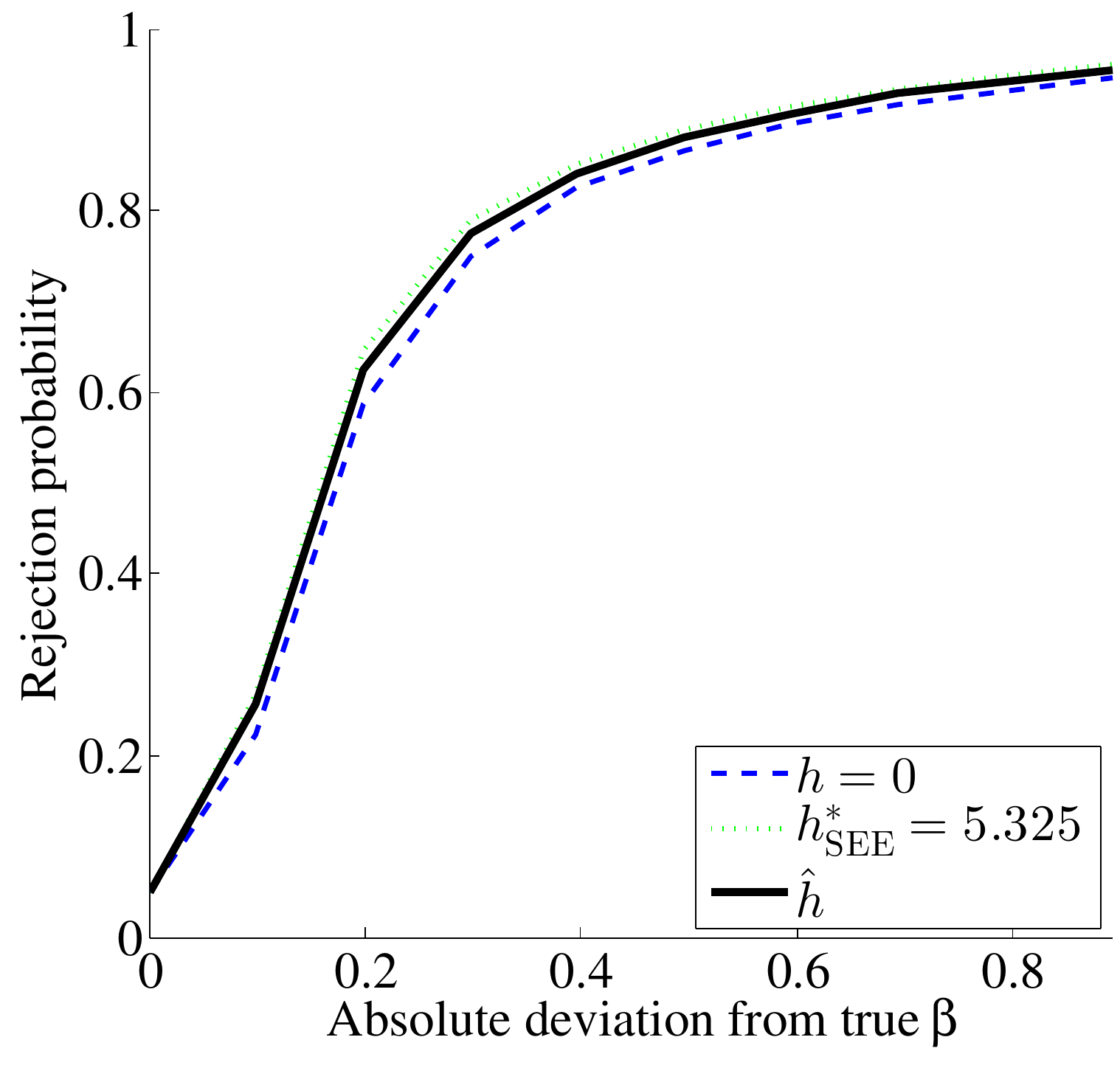}
	\hfill
\includegraphics[clip=true,trim=30 175 145 200,width=.48\textwidth]
	{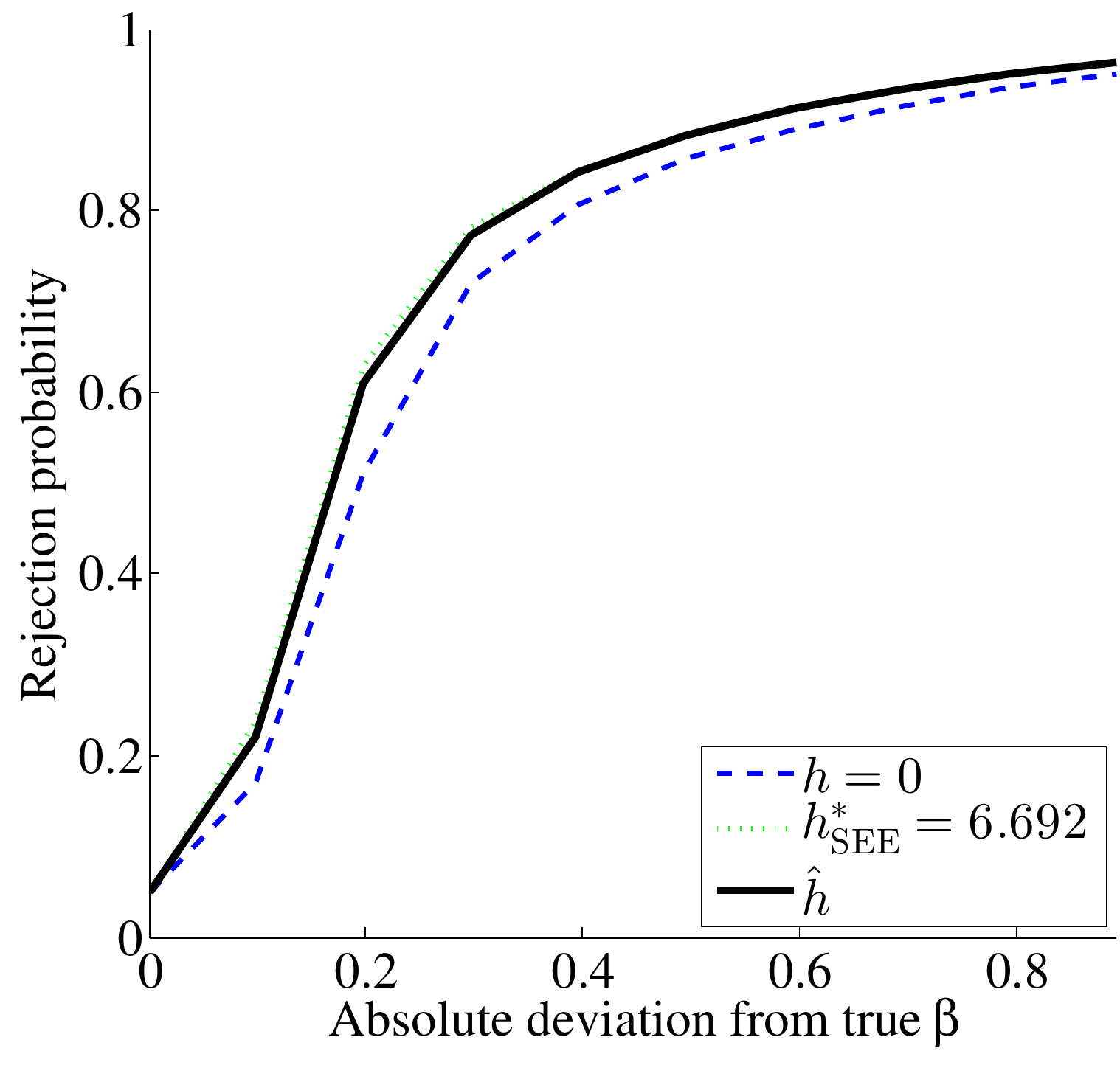}
\end{center}
\caption{Size-adjusted power for DGPs 1.1 (left) and 1.3 (right).}
\label{fig:adjpwr-1a}
\end{figure}

In DGPs 1.1--1.3, SEE($\hat h$) has smaller MSE than either the unsmoothed estimator or OLS, for both the intercept and slope coefficients.  Figure \ref{fig:mse-1a} shows MSE for DGPs 1.1 and 1.3.  It shows that the MSE of SEE($\hat h$) is very close to that of the best estimator with a fixed $h$.  In principle, a data-dependent $\hat h$ can attain MSE even lower than any fixed $h$.  SAP for SEE($\hat h$) is similar to that with $h=0$; see Figure \ref{fig:adjpwr-1a} for DGPs 1.1 and 1.3. 

\begin{figure}[tbph]
\begin{center}
\includegraphics[clip=true,trim=0 175 125 200,width=.49\textwidth]
	{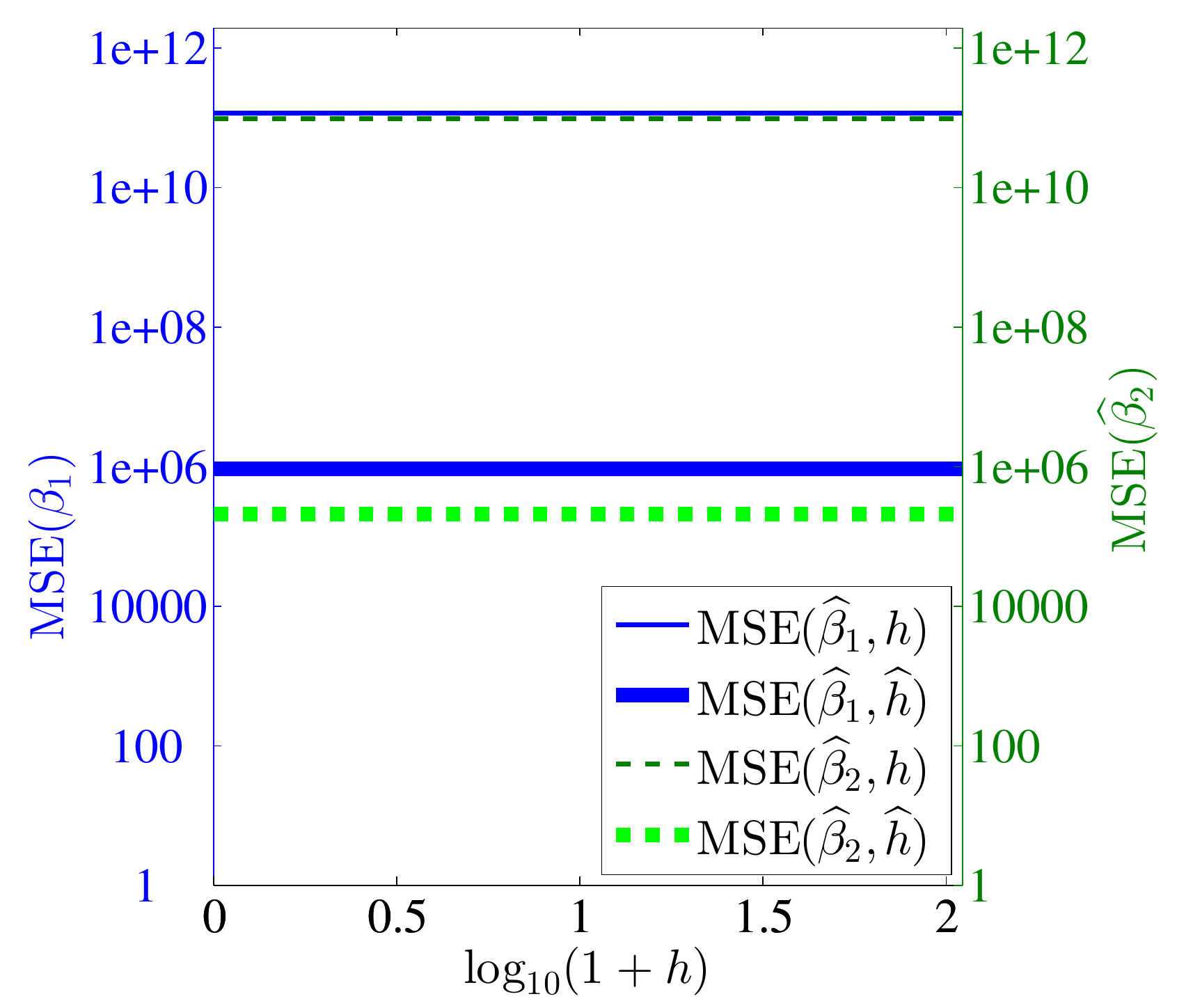}
	\hfill
\includegraphics[clip=true,trim=0 175 125 200,width=.49\textwidth]
	{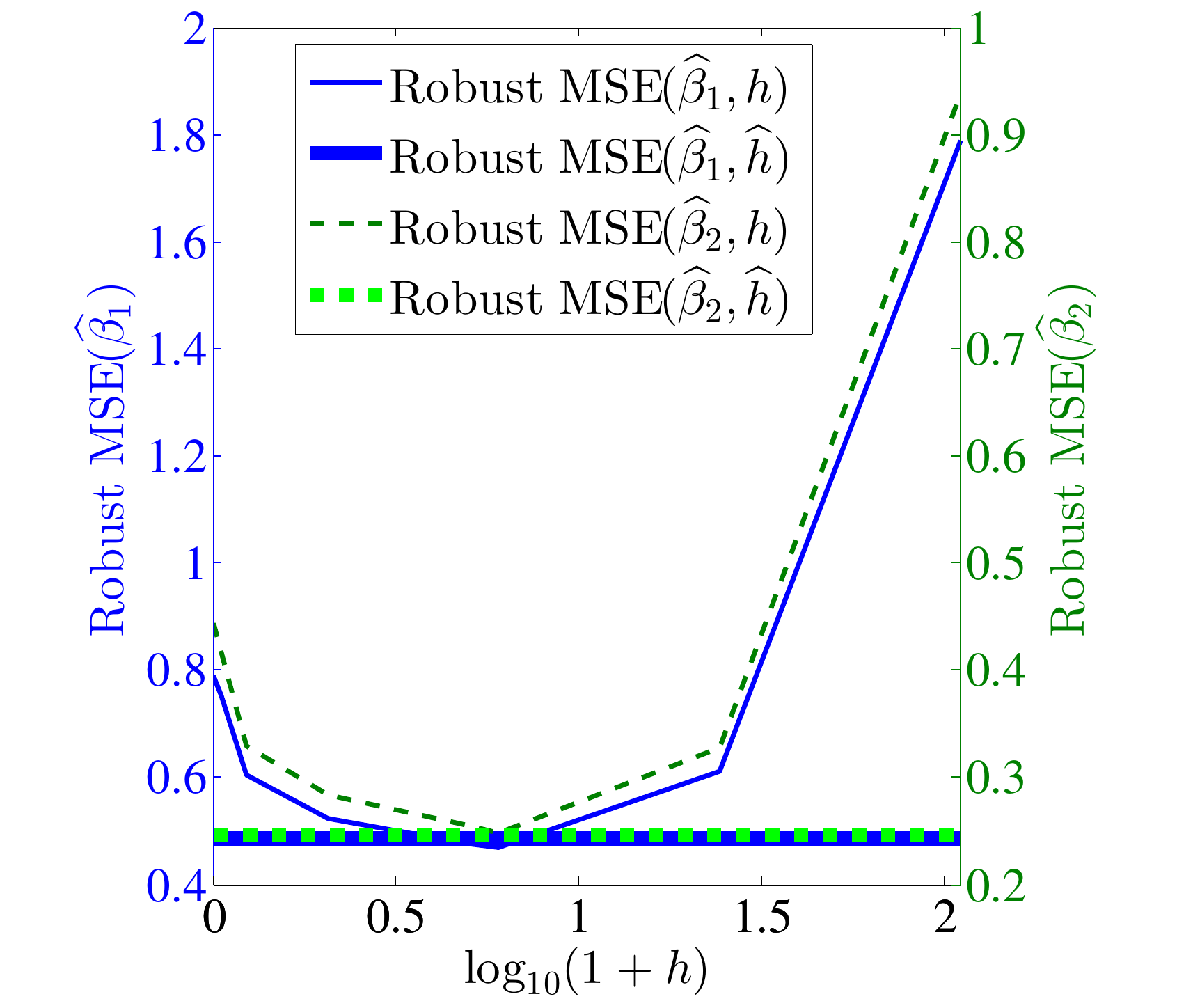}
\end{center}
\caption{For DGP 4.2, MSE (left) and ``robust
MSE'' (right): squared median-bias plus the square of the
interquartile range divided by $1.349$, $\text{Bias}_{\mathrm{median}}^{2}+(%
\text{IQR}/1.349)^{2}$.}
\label{fig:mse-4.2}
\end{figure}

\begin{figure}[tbph]
\centering
\includegraphics[clip=true,trim=25 175 125 200,width=.49\textwidth]
	{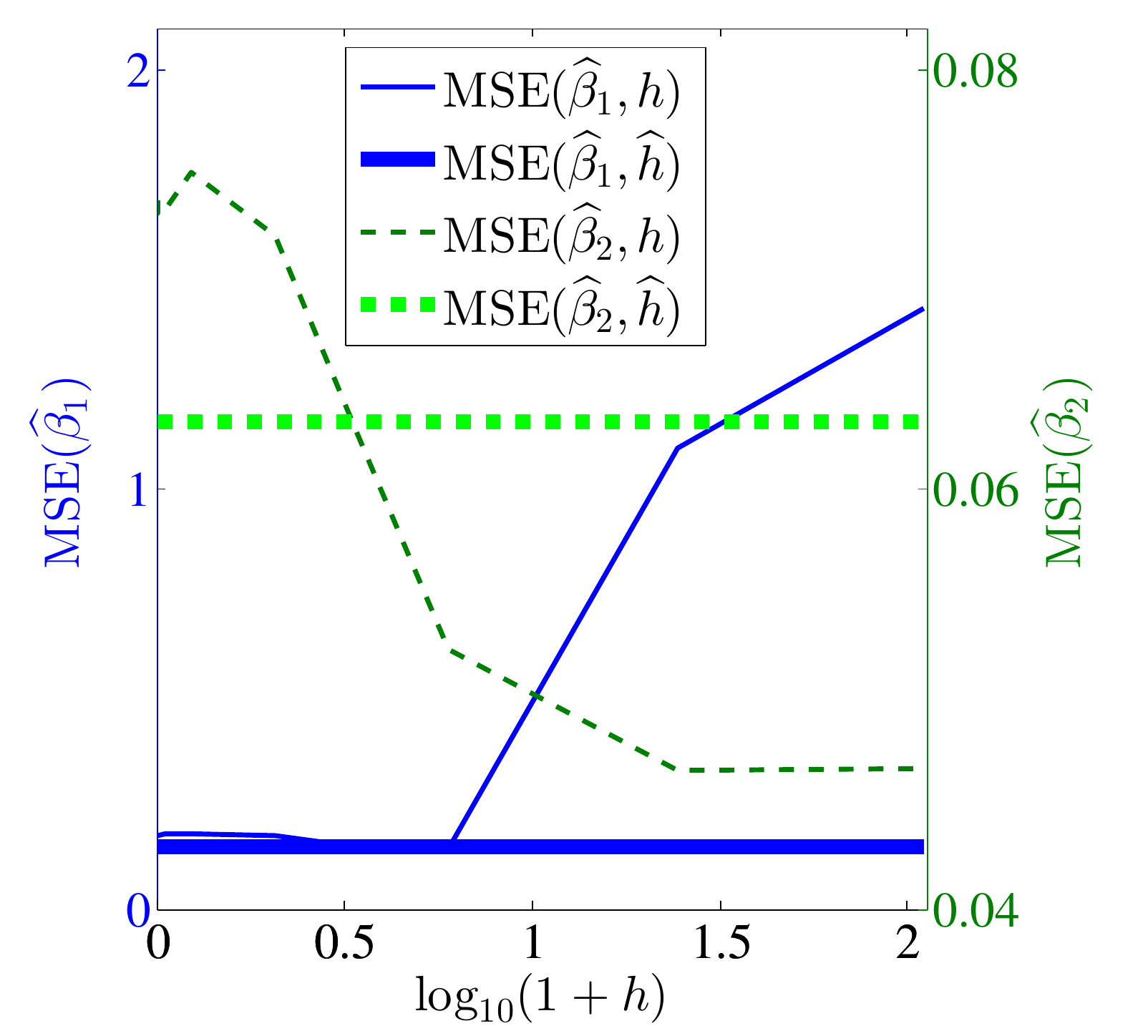}
	\hfill
\includegraphics[clip=true,trim=25 175 125 200,width=.49\textwidth]
	{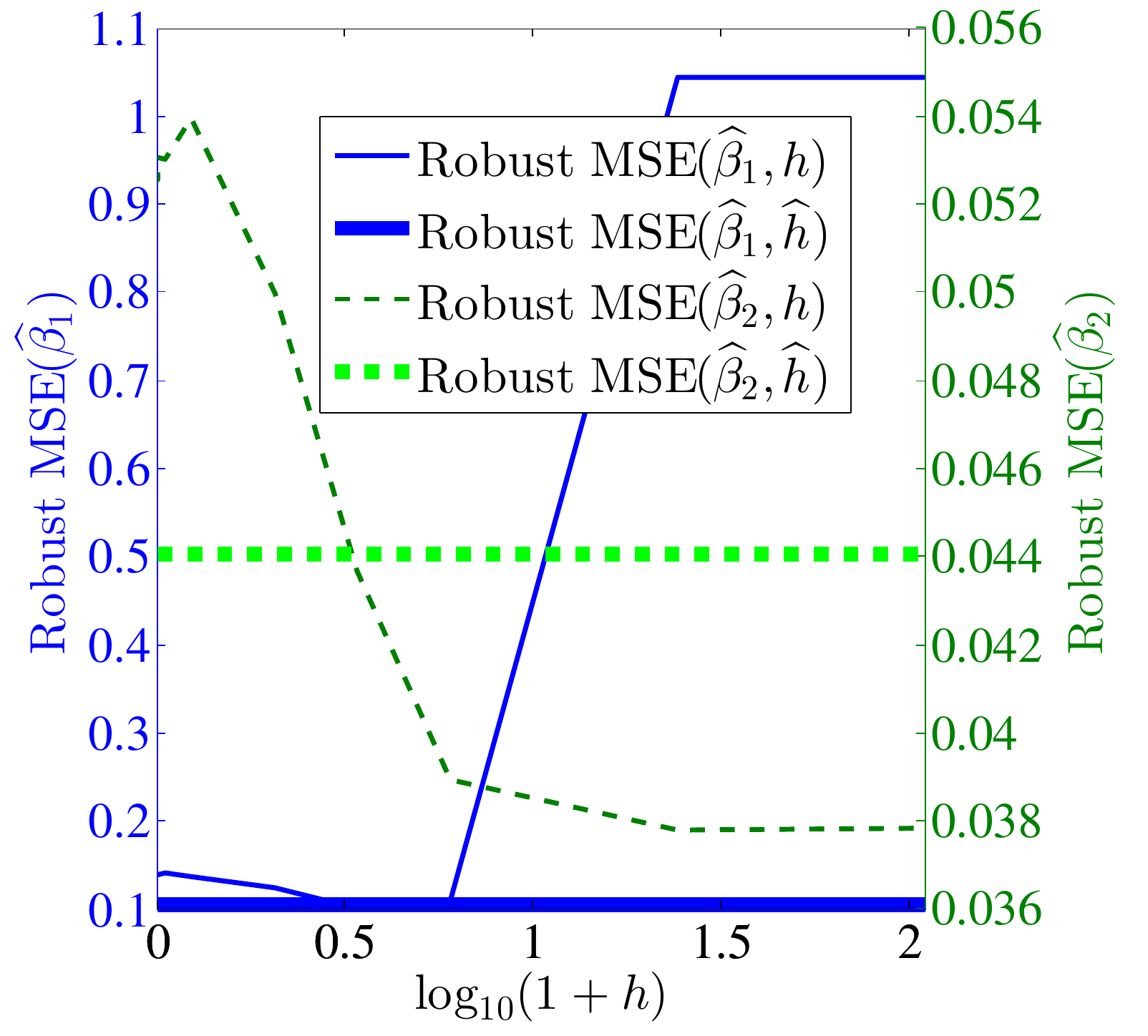}
\caption{Similar to Figure \ref{fig:mse-4.2}, MSE (left) and ``robust MSE'' (right) for DGP 4.3.}
\label{fig:mse-4.3}
\end{figure}

Figures \ref{fig:mse-4.2} and \ref{fig:mse-4.3} show MSE and ``robust MSE'' for two DGPs with endogeneity.  Graphs for the other endogenous DGP (4.1) are similar to those for the slope estimator in DGP 4.3 but with larger MSE; they may be found in the working paper.  The MSE graph for DGP 4.2 is not as informative since it is sensitive to very large outliers that occur in only a few replications.  However, as shown, the MSE for SEE($\hat h$) is still better than that for the unsmoothed IV-QR estimator, and it is nearly the same as the MSE for the mean IV estimator (not shown: $1.1\times10^6$ for $\beta_1$, $2.1\times10^5$ for $\beta_2$).  
For robust MSE, SEE($\hat h$) is again always better than the unsmoothed estimator.  For DGP 4.3 with normal errors and $q=0.35$, it is similar to the IV estimator, slightly worse for the slope coefficient and slightly better for the intercept, as expected.  Also as expected, for DGP 4.2 with Cauchy errors, SEE($\hat h$) is orders of magnitude better than the mean IV estimator. 
Overall, using $\hat{h}$ appears
to consistently reduce the MSE of all estimator components compared with $h=0$ and with IV ($h=\infty $). Almost always, the exception is cases where
MSE is monotonically decreasing with $h$ (mean regression is more
efficient), in which $\hat{h}$ is much better than $h=0$ but not quite large
enough to match $h=\infty $. 

\begin{figure}[tbph]
\begin{center}
\includegraphics[clip=true,trim=30 175 145 200,width=.48\textwidth]
	{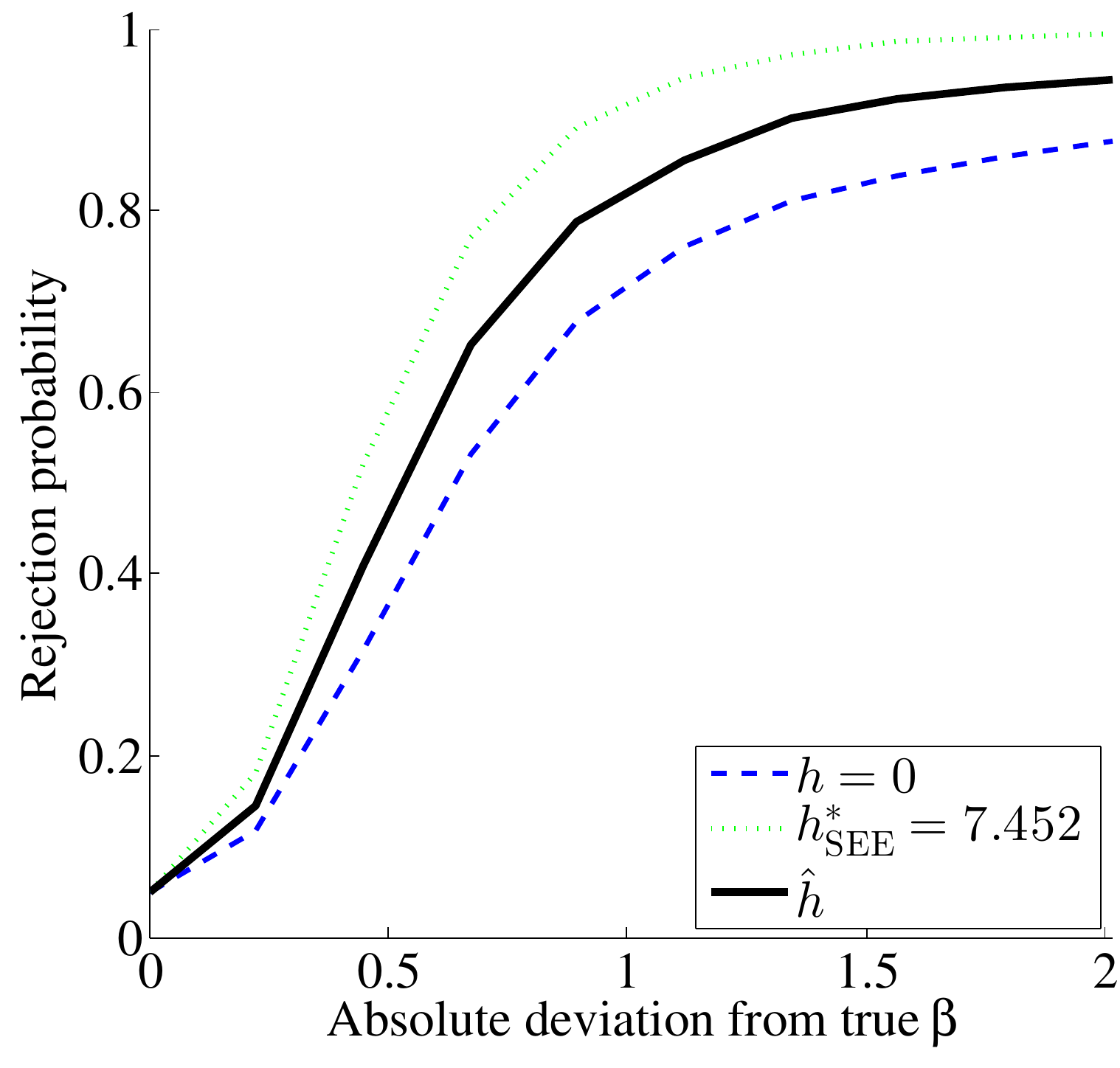} 
	\hfill
\includegraphics[clip=true,trim=30 175 145 200,width=.48\textwidth]
	{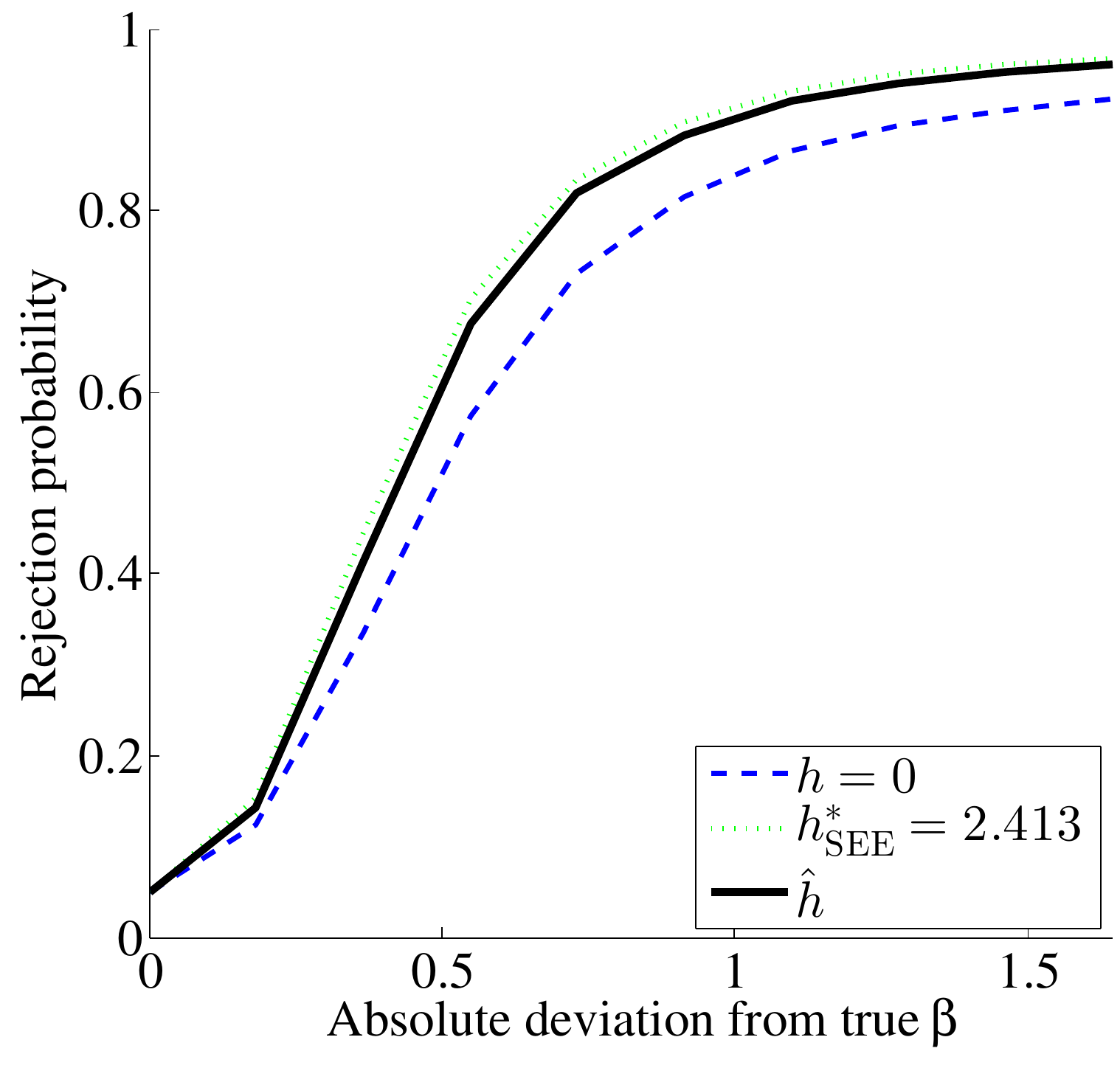}
\end{center}
\caption{Size-adjusted power for DGPs 4.1 (left) and 4.3 (right).}
\label{fig:adjpwr-4a}
\end{figure}

Figure \ref{fig:adjpwr-4a} shows SAP for DGPs 4.1 and 4.3.  The gain from smoothing is more substantial than in the exogenous DGPs, close to $10$ percentage points for a range of deviations.  Here, the randomness in $\hat h$ is not helpful.  In DGP 4.2 (not shown), the SAP for $\hat h$ is actually a few percentage points \emph{below} that for $h=0$ (which in turn is below the infeasible $h^*$), and in DGP 4.1, the SAP improvement from using the infeasible $h^*$ instead of $\hat h$ is similar in magnitude to the improvement from using $\hat h$ instead of $h=0$. Depending on one's loss function of type I and type II errors, the SEE-based test may be preferred or not.

\section{Conclusion\label{sec:conclusion}}

We have presented a new estimator for quantile regression with or without
instrumental variables. Smoothing the estimating equations (moment
conditions) has multiple advantages beyond the known advantage of allowing
higher-order expansions. It can reduce the MSE of both the estimating
equations and the parameter estimator, minimize type I error and improve
size-adjusted power of a chi-square test, and allow more reliable
computation of the instrumental variables quantile regression estimator
especially when the number of endogenous regressors is larger. We have given
the theoretical bandwidth that optimizes these properties, and simulations
show our plug-in bandwidth to reproduce all these advantages over the
unsmoothed estimator. Links to mean instrumental variables regression and
robust estimation are insightful and of practical use.

The strategy of smoothing the estimating equations can be applied to any
model with nonsmooth estimating equations; there is nothing peculiar to the
quantile regression model that we have exploited. For example, this strategy
could be applied to censored quantile regression, or to select the optimal
smoothing parameter in \citeauthor{Horowitz2002}'s (\citeyear{Horowitz2002})
smoothed maximum score estimator. The present paper has focused on
parametric and linear IV quantile regression; extensions to nonlinear IV
quantile regression and nonparametric IV quantile regression along the lines
of \citet{ChenPouzo2009,ChenPouzo2012} are currently under development. 


\theendnotes

\bibliographystyle{chicago} 

\appendix

\renewcommand{\theequation}{\Alph{section}.\arabic{equation}}
\setcounter{equation}{0}

\section{Appendix of Proofs\label{sec:app-pfs}}

\subsection*{Proof of Theorem \protect \ref{thm:Wj}}

\subsubsection*{First moment of $W_j$}

Let $\left[ U_{L}(z),U_{H}(z)\right] $ be the support of the conditional PDF
of $U$ given $Z=z$. Since $P(U_{j}<0\mid Z_{j})=q$ for almost all $Z_{j}$
and $h\rightarrow 0$, we can assume without loss of generality that $%
U_{L}(Z_{j})\leq -h$ and $U_{H}(Z_{j})\geq h$ for almost all $Z_{j}$. For
some $\tilde{h}\in \left[ 0,h\right] $, we have 
\begin{align*}
\E(W_{j})
  &= \E\left \{ Z_{j}\left[ G(-U_{j}/h)-q\right] \right \} 
   = \E\left \{
\left( \int_{U_{L}(Z_{j})}^{U_{H}(Z_{j})}\left[ G(-u/h)-q\right]
dF_{U|Z}(u\mid Z_{j})\right) Z_{j}\right \} \\
& =\E\left[ \left( \left[ G(-u/h)-q\right] \left. F_{U|Z}(u\mid Z_{j})\right
\vert _{U_{L}(Z_{j})}^{U_{H}(Z_{j})}+\frac{1}{h}%
\int_{U_{L}(Z_{j})}^{U_{H}(Z_{j})}F_{U|Z}(u\mid Z_{j})G^{\prime
}(-u/h)du\right) Z_{j}\right] \\
& =\E\left \{ \left[ -q+\int_{-1}^{1}F_{U|Z}(-hv\mid Z_{j})G^{\prime
}(v)dv\right] Z_{j}\right \} \quad \text{(since $G^{\prime }(v)=0$ for $%
v\notin \lbrack -1,1]$)} \\
& =\E\left \{ \left[ -q+F_{U|Z}(0\mid Z_{j})+\int_{-1}^{1}\left(
\sum_{k=1}^{r}f_{U|Z}^{(k-1)}(0\mid Z_{j})\frac{(-h)^{k}v^{k}}{k!}\right)
G^{\prime }(v)dv\right] Z_{j}\right \} \\
& \quad +\E\left \{ \left[ \int_{-1}^{1}f_{U|Z}^{(r)}(-\tilde{h}v\mid
Z_{j})v^{r}G^{\prime }(v)dv\right] Z_{j}\right \} \frac{(-h)^{r+1}}{\left(
r+1\right) !} \\
& =\frac{(-h)^{r}}{r!}\left[ \int_{-1}^{1}G^{\prime }(v)v^{r}dv\right] \E%
\left[ f_{U|Z}^{(r-1)}(0\mid Z_{j})Z_{j}\right] \\
& \quad +\E\left \{ \left[ \int_{-1}^{1}f_{U|Z}^{(r)}(-\tilde{h}v\mid
Z_{j})v^{r}G^{\prime }(v)dv\right] Z_{j}\right \} O\left( h^{r+1}\right) .
\end{align*}%
Under Assumption \ref{a:fUZ}, for some bounded $C(\cdot )$ we have 
\begin{align*}
& \left \Vert \E\left \{ \left[ \int_{-1}^{1}f_{U|Z}^{(r)}(-\tilde{h}v\mid
Z)v^{r}G^{\prime }(v)dv\right] Z\right \} \right \Vert \\
& \quad \leq \E\left[ \int_{-1}^{1}C(Z)\left \Vert Z\right \Vert \left
\vert v^{r}G^{\prime }(v)\right \vert dv\right] =O(1).
\end{align*}%
Hence 
\begin{equation*}
\E(W_{j})=\frac{(-h)^{r}}{r!}\left[ \int_{-1}^{1}G^{\prime }(v)v^{r}dv\right]
\E\left[ f_{U|Z}^{(r-1)}(0\mid Z_{j})Z_{j}\right] +o(h^{r}).
\end{equation*}


\subsubsection*{Second moment of $W_j$}

For the second moment, 
\begin{equation*}
\E\left(W_{j}^{\prime }W_{j}\right)
   = \E\left \{ \left[ G(-U_{j}/h)-q\right]^{2}Z_{j}^{\prime }Z_{j}\right \} 
   = \E\left \{ \left[\int_{U_{L}(Z_{j})}^{U_{H}(Z_{j})}\left[ G(-u/h)-q\right] ^{2}dF_{U|Z}(u\mid
Z_{j})\right] Z_{j}^{\prime }Z_{j}\right \} .
\end{equation*}%
Integrating by parts and using Assumption \ref{a:fUZ}(i) in the last line
yields:%
\begin{align*}
& \int_{U_{L}(Z_{j})}^{U_{H}(Z_{j})}\left[ G(-u/h)-q\right]
^{2}dF_{U|Z}(u\mid Z_{j}) \\
& =\left. \left[ G(-u/h)-q\right] ^{2}F_{U|Z}(u\mid Z_{j})\right \vert
_{U_{L}(Z_{j})}^{U_{H}(Z_{j})}+\frac{2}{h}%
\int_{U_{L}(Z_{j})}^{U_{H}(Z_{j})}F_{U|Z}(u\mid Z_{j})\left[ G(-u/h)-q\right]
G^{\prime }(-u/h)du \\
& =q^{2}+2\int_{-1}^{1}F_{U|Z}(hv\mid Z_{j})\left[ G(-v)-q\right] G^{\prime
}(-v)dv\quad \text{(since $G^{\prime }(v)=0$ for $v\notin \lbrack -1,1]$)} \\
& =q^{2}+2q\left \{ \int_{-1}^{1}\left[ G(-v)-q\right] G^{\prime
}(-v)dv\right \} +2hf_{U|Z}(0\mid Z_{j})\left \{ \int_{-1}^{1}v\left[ G(-v)-q%
\right] G^{\prime }(-v)dv\right \} \\
& \quad +\left \{ \int_{-1}^{1}v^{2}f_{U|Z}^{\prime }(\tilde{h}v\mid Z_{j})%
\left[ G(-v)-q\right] G^{\prime }(-v)dv\right \} h^{2}.
\end{align*}%
But 
\begin{align*}
& 2\int_{-1}^{1}\left[ G(-v)-q\right] G^{\prime }(-v)dv=\int_{-1}^{1}2\left[
G(u)-q\right] G^{\prime }(u)du \\
& =\left. \left[ G^{2}(u)-2qG(u)\right] \right \vert _{-1}^{1}=1-2q , \\
& 2\int_{-1}^{1}v\left[ G(-v)-q\right] G^{\prime }(-v)dv=-2\int_{-1}^{1}u%
\left[ G(u)-q\right] G^{\prime }(u)du \\
& \quad =-2\int_{-1}^{1}uG(u)G^{\prime }(u)du = -\left[ \left.
uG^2(u)\right|_{-1}^1 - \int_{-1}^1 G^2(u)du \right] \\
& \quad = -\left[1-\int_{-1}^{1}G^{2}(u)du\right]\quad \text{(by Assumption %
\ref{a:G}(ii))} , \\
\intertext{and}%
& \left \vert \int_{-1}^{1}v^{2}f_{U|Z}^{\prime }(\tilde{h}v\mid Z_{j})\left[
G(-v)-q\right] G^{\prime }(-v)dv\right \vert \leq \int_{-1}^{1}C(Z_{j})\left
\vert v^{2}G^{\prime }(v)\right \vert dv
\end{align*}%
for some function $C(\cdot )$. So 
\begin{align*}
\E(W_{j}^{\prime }W_{j}) 
  &= \E\left( \left \{ q^{2}+q\left( 1-2q\right) -hf_{U|Z}(0\mid Z_{j}) 
\left[ 1-\int_{-1}^{1}G^{2}(u)du\right] \right \} Z_{j}^{\prime
}Z_{j}\right) +O(h^{2}) \\
  &= q\left( 1-q\right) \E\left( Z_{j}^{\prime }Z_{j}\right) -h\left[
1-\int_{-1}^{1}G^{2}(u)du\right] \E\left[ f_{U|Z}(0\mid Z_{j})Z_{j}^{\prime
}Z_{j}\right] +O(h^{2}).
\end{align*}%
Similarly, we can show that 
\begin{align*}
\E(W_{j}W_{j}^{\prime }) 
  =q\left( 1-q\right) \E(Z_{j}Z_{j}^{\prime })-h\left[ 1-%
\int_{-1}^{1}G^{2}(u)du\right] \E\left[ f_{U|Z}(0\mid
Z_{j})Z_{j}Z_{j}^{\prime }\right] +O(h^{2}).
\end{align*}


\subsubsection*{First-order asymptotic distribution of $m_n$}

We can write $m_{n}$ as 
\begin{equation}
m_{n}=\frac{1}{\sqrt{n}}\sum_{j=1}^{n}W_{j}=\frac{1}{\sqrt{n}}%
\sum_{j=1}^{n}\left[ W_{j}-\E(W_j)\right] +\sqrt{n}\E(W_{j}).  \label{CLT_proof}
\end{equation}%
In view of the mean of $W_{j}$, we have $\sqrt{n}\E(W_{j})=O(h^{r}\sqrt{n}%
)=o(1) $ by Assumption \ref{a:h}. So the bias is asymptotically
(first-order) negligible. Consequently, the variance of $W_{j}$ is $%
\E(W_{j}W_{j}^{\prime })+o(1)$, so the first-order term from the second
moment calculation above can be used for the asymptotic variance.

Next, we apply the Lindeberg--Feller central limit theorem to the first term
in \eqref{CLT_proof}, which is a scaled sum of a triangular array since the
bandwidth in $W_{j}$ depends on $n$. We consider the case when $W_{j}$ is a
scalar as vector cases can be handled using the Cram\'{e}r--Wold device.
Note that 
\begin{align*}
\sigma _{W}^{2}& \equiv \mathrm{Var}\left\{ \frac{1}{\sqrt{n}}%
\sum_{j=1}^{n}\left[ W_{j}-\E(W_{j})\right] \right\} =n\frac{1}{n}\mathrm{Var}%
\left[ W_{j}-\E(W_{j})\right] \quad \text{(by iid Assumption \ref{a:sampling})%
} \\
& =\E\left(W_{j}^{2}\right) - \left[ \E(W_{j})\right]^{2}=q\left( 1-q\right) E\left( Z_{j}^{2}%
\right) \left[1+o(1)\right].
\end{align*}%
For any $\varepsilon >0$, 
\begin{align*}
& \lim_{n\rightarrow \infty }\sum_{j=1}^{n}\E\left( \frac{W_{j}-\E W_{j}}{\sqrt{%
n}\sigma _{W}}\right) ^{2}1\left \{ \frac{\left \vert W_{j}-\E W_{j}\right
\vert }{\sqrt{n}\sigma _{W}}\geq \varepsilon \right \} \\
& \quad =\lim_{n\rightarrow \infty }\frac{1}{n}\sum_{j=1}^{n}\E\frac{\left(
W_{j}-\E W_{j}\right) ^{2}}{\sigma _{W}^{2}}1\left \{ \frac{\left \vert
W_{j}-\E W_{j}\right \vert }{\sigma _{W}}\geq \sqrt{n}\varepsilon \right \} \\
& \quad =\lim_{n\rightarrow \infty }\E\frac{\left( W_{j}-EW_{j}\right) ^{2}}{%
\sigma _{W}^{2}}1\left \{ \frac{\left \vert W_{j}-EW_{j}\right \vert }{%
\sigma _{W}}\geq \sqrt{n}\varepsilon \right \} =0,
\end{align*}%
where the last equality follows from the dominated convergence theorem, as 
\begin{equation*}
\frac{\left( W_{j}-\E W_{j}\right) ^{2}}{\sigma _{W}^{2}}
1 \left \{ \frac{\left
\vert W_{j}-\E W_{j}\right \vert }{\sigma _{W}}\geq \sqrt{n}\varepsilon \right
\} \leq C\frac{Z_{j}^{2}+\E(Z_{j}^{2})}{\sigma _{W}^{2}}
\end{equation*}%
for some constant $C$ and $\E(Z_{j}^{2})<\infty $. So the Lindeberg condition
holds and $m_{n}\overset{d}{\rightarrow }N(0,V)$. \qed

\subsection*{Properties of estimating equations derived from smoothed criterion
function}

The EE derived from smoothing the criterion function in \eqref{eqn:SCF-EE} for $Z_j=X_j$ can be written
\begin{align*}
0 &= n^{-1}\sum_{j=1}^n W_j, \quad
W_j \equiv X_j\left[G(-U_j/h)-q\right] + (1/h)G'(-U_j/h)(-X_jU_j) .
\end{align*}
We calculate $\E(W_j)$ and $\E(W_jW_j')$, which can be compared to the results in Theorem \ref{thm:Wj}. 

\subsubsection*{First moment}

Using iterated expectations, 
\begin{align*}
\E & \left \{ \left[ G\left( -U_{j}/h\right) -q\right] Z_{j}\right \} -\frac{1%
}{h}\E\left[ U_{j}G^{\prime }\left( -U_{j}/h\right) Z_{j}\right] \\
  &= \frac{(-h)^{r}}{r!}\left[ \int G^{\prime }(v)v^{r}dv\right]\E\left[
f_{U|Z}^{(r-1)}(0\mid Z_{j})Z_{j}\right] +o\left( h^{r}\right) 
    -\E\left[h\int vG^{\prime }(v)f_{U|Z}\left( -hv\mid Z_{j}\right) dvZ_{j}\right] \\
  &= \frac{(-h)^{r}}{r!}\left[ \int G^{\prime }(v)v^{r}dv\right] \E\left[
f_{U|Z}^{(r-1)}(0\mid Z_{j})Z_{j}\right] +o\left( h^{r}\right) \\
  & \quad -h\frac{(-h)^{r-1}}{(r-1)!}\left( \int G^{\prime }(v)v^{r}dv\right) \E%
\left[ f_{U|Z}^{(r-1)}(0\mid Z_{j})Z_{j}\right] +o\left( h^{r}\right) \\
  &= (-h)^{r}\left( \frac{1}{r!}+\frac{1}{\left( r-1\right) !}\right) \left(
\int G^{\prime }(v)v^{r}dv\right) \E\left[ f_{U|Z}^{(r-1)}(0\mid Z_{j})Z_{j}%
\right] +o\left( h^{r}\right) .
\end{align*}

\subsubsection*{Second moment}

Here,
\begin{align*}
W_jW_j' &= \underbrace{\left[G(-U_j/h)-q\right]X_jX_j'}_\textrm{Term 1}
          +\underbrace{(2/h)\left[G(-U_j/h)-q\right]G'(-U_j/h)(-U_j)X_jX_j'}_\textrm{Term 2} \\
    &\quad      +\underbrace{h^{-2}\left[G'(-U_j/h)\right]^2U_j^2X_jX_j'}_\textrm{Term 3} .
\end{align*}
Term 1 is identical to the (only) SEE term, so its expectation is identical, too.  The dominant terms of the expectations of Terms 2 and 3 are both positive and $O(h)$.  

For Term 2, the expectation conditional on $X_j$ is
\begin{align*}
& (2/h) \int_{-h}^h [G(-u/h)-q]G'(-u/h)(-u)f_{U|X}(u\mid X_j)du \,X_jX_j' \\
&= (2/h)\int_{-1}^1[G(-v)-q]G'(-v)(-vh)f_{U|X}(hv\mid X_j)hdv \,X_jX_j' \\
&= -2\int_{-1}^1[1-G(v)-q]G'(v)(vh)f_{U|X}(hv\mid X_j)dv \, X_jX_j' \\
&= -2(1-q)\int_{-1}^1 G'(v)(vh)[f_{U|X}(0\mid X_j)+(hv)f_{U|X}'(0\mid X_j)+\cdots+\frac{(hv)^{r-1}}{(r-1)!}f_{U|X}^{(r-1)}(0\mid X_j)]dv \, X_jX_j'  \\
 &\quad +2\int_{-1}^1 G(v)G'(v)(vh)[f_{U|X}(0\mid X_j)+(hv)f_{U|X}'(\tilde vh\mid X_j)]dv \, X_jX_j' \\
&= \left[O(h^r) + 2h\underbrace{\int_{-1}^1 G(v)G'(v)vdv}_\textrm{$>0$}f_{U|X}(0\mid X_j) + O(h^2)\right] X_jX_j' .
\end{align*}
The largest is the $O(h)$ middle term, which is positive. 

For Term 3, the expectation conditional on $X_j$ is
\begin{align*}
& h^{-2} \int_{-h}^h[G'(-u/h)]^2 u^2f_{U|X}(u\mid X_j)du \, X_jX_j' \\
&= h^{-2} \int_{-1}^1[G'(-v)]^2 (hv)^2 f_{U|X}(vh\mid X_j)hdv \, X_jX_j' \\
&= \left[ h f_{U|X}(0\mid X_j)\int_{-1}^1[G'(v)v]^2dv + O(h^2)\right] X_jX_j' .
\end{align*}
This is $O(h)$ and positive since the integrand is positive.  For the $G(\cdot)$ we use in our code, for example, $\int_{-1}^1[G'(v)v]^2dv=0.061$; values are similar if $G'(\cdot)$ is a bisquare ($0.065$) or Epanechnikov ($0.086$) kernel. 

Altogether, 
\begin{align*}
\E(W_jW_j') 
  &= \underbrace{q(1-q)\E(X_jX_j') 
                 -h\left(1-\int_{-1}^1 [G(v)]^2dv\right)\E\left[f_{U|X}(0\mid X_j)X_jX_j'\right]
                 +O(h^2)}_\textrm{from Term 1} \\
  &\qquad+
     \underbrace{h \int_{-1}^1 2G(v)G'(v)vdv \, \E\left[f_{U|X}(0\mid X_j)X_jX_j'\right]}_\textrm{from Term 2} \\
  &\qquad+
     \underbrace{h \int_{-1}^1[G'(v)v]^2dv \, \E\left[f_{U|X}(0\mid X_j)X_jX_j'\right] +O(h^2)}_\textrm{from Term 3} \\
  &= q(1-q)E(X_jX_j')
     +h \int_{-1}^1[G'(v)v]^2dv \, \E\left[f_{U|X}(0\mid X_j)X_jX_j'\right]
     +O(h^2) .
\end{align*}
The cancellation between the $O(h)$ parts of Terms 1 and 2 is by integration by parts of Term 2, using $\D{}{v}G(v)^2=2G(v)G'(v)$.  The remaining $O(h)$ term is positive, whereas for SEE the $O(h)$ term is negative.

\subsubsection*{AMSE of estimator} 

Analogous to \eqref{eqn:expansion2}, we examine the estimator's properties by way of a mean value expansion, using a similar change of variables, Taylor expansion, and properties of $G'(\cdot)$, 
\begin{align*}
\E & \D{}{\beta'} n^{-1/2}m_n(\beta_0) \\
  &= \E\left[ (1/h)G'(-U_j/h) X_jX_j' 
            +\underbrace{h^{-1}G'(-U_j/h) X_jX_j' 
            -h^{-2}G''(-U_j/h)U_jX_jX_j'}_\textrm{Not in SEE; using product rule for derivative}
            \right] \\
  &= \E\left\{ 2 \int_{-1}^1 G'(-v)\left[f_{U|X}(0\mid X_j) +\cdots\right] dv \, X_jX_j'\right\} \\
  &\quad+ \E\left\{ h^{-1}\int_{-1}^1 G''(v)vh\left[f_{U|X}(0\mid X_j)+(vh)f_{U|X}'(0\mid X_j)+\cdots\right]dv \, X_jX_j' \right\} \\
  &= \E[2f_{U|X}(0\mid X_j)X_jX_j'] + O(h^r) \\
  &\quad+(-1)\E[f_{U|X}(0\mid X_j)X_jX_j'] 
    +(-1)h \E\left[f_{U|X}'(0\mid X_j)X_jX_j'\right] 
    +O(h^2) \\
  &= \E[f_{U|X}(0\mid X_j)X_jX_j'] 
    -h \E\left[f_{U|X}'(0\mid X_j)X_jX_j'\right] 
    +O(h^2) ,
\end{align*}
where $G''(-v)=-G''(v)$, and integration by parts gives 
\[ \int_{-1}^1 vdG'(v)=\left.vG'(v)\right|_{-1}^1 -\int_{-1}^1 G'(v)dv=0-1=-1. \]
The dominant term is the same as for SEE.  However, the $O(h)$ term will also affect the estimator's AMSE, through the first-order variance.  Its sign is indeterminant since it depends on the conditional PDF derivative.  

\subsection*{Proof of Proposition \protect \ref{prop:hSEE}}

The first expression comes directly from the FOC. Under the assumption $U%
\mathpalette{\protect \independenT}{\perp} Z$, we have 
\begin{equation*}
h_{\text{SEE}}^{\ast }=\left( \frac{\left( r!\right) ^{2}\left[
1-\int_{-1}^{1}G^{2}(u)du\right] f_{U}(0)}{2r\left( \int_{-1}^{1}G^{\prime
}(v)v^{r}dv\right) ^{2}\left[ f_{U}^{\left( r-1\right) }(0)\right] ^{2}}%
\frac{\E(Z^{\prime }V^{-1}Z)}{\E(Z)'V^{-1}\E(Z)}\frac{1}{n}\right) ^{%
\frac{1}{2r-1}}.
\end{equation*}%
The simplified $h_{\text{SEE}}^{\ast }$ then follows from the lemma below.

\begin{lemma}
\label{lem:h-simplify} If $Z\in \mathbb{R}^d$ is a random vector with first
element equal to one and $V\equiv \E(ZZ^{\prime})$ is nonsingular, then 
\begin{equation*}
\E(Z'V^{-1}Z) / \left[\E(Z')V^{-1}\E(Z)\right] = d.
\end{equation*}
\end{lemma}

\begin{proof}
For the numerator, rearrange using the trace: 
\begin{align*}
\E(Z^{\prime}V^{-1}Z) 
  &= \E\left[\mathrm{tr}\left(Z^{\prime}V^{-1}Z\right)%
\right] 
   = \E\left[\mathrm{tr}\left(V^{-1}ZZ'\right) \right] 
   = \mathrm{tr}\left[V^{-1}\E(ZZ')\right]
   = \mathrm{tr}\left(I_d\right) = d.
\end{align*}

For the denominator, let $\E(Z^{\prime })=(1,t^{\prime })$ for some $t\in 
\mathbb{R}^{d-1}$. Since the first element of $Z$ is one, the first row and
first column of $V$ are $\E(Z^{\prime })$ and $\E(Z)$. Writing the other $%
(d-1)\times (d-1)$ part of the matrix as $\Omega $, 
\begin{equation*}
V=\E(ZZ^{\prime })=\left( 
\begin{array}{cc}
1 & t^{\prime } \\ 
t & \Omega%
\end{array}%
\right) .
\end{equation*}%
We can read off $V^{-1}\E(Z)=(1,0,\ldots ,0)^{\prime }$ from the first column
of the identity matrix since 
\begin{equation*}
V^{-1}\left( 
\begin{array}{cc}
1 & t^{\prime } \\ 
t & \Omega%
\end{array}%
\right) =V^{-1}V=I_{d}=\left( 
\begin{array}{cccc}
1 & 0 & \cdots & 0 \\ 
0 & 1 & \cdots & 0 \\ 
\vdots & \vdots & \ddots & \vdots \\ 
0 & 0 & \cdots & 1%
\end{array}%
\right) .
\end{equation*}%
Thus, 
\begin{equation*}
\E(Z')V^{-1}\E(Z) = (1,t^{\prime })(1,0,\ldots ,0)^{\prime }=1. \qedhere
\end{equation*}
\end{proof}


\subsection*{Proof of Theorem \protect \ref{thm:inf}}

Adding to the variables already defined in the main text, let 
\begin{equation*}
Z_{j}^{\ast }\equiv \left( \E Z_{j}Z_{j}^{\prime }\right) ^{-1/2}Z_{j}\text{
and }D_{n}\equiv n^{-1}\sum_{j=1}^{n}\left( Z_{j}^{\ast }Z_{j}^{\ast \prime
}-\E Z_{j}^{\ast }Z_{j}^{\ast \prime }\right) =\frac{1}{n}%
\sum_{j=1}^{n}Z_{j}^{\ast }Z_{j}^{\ast \prime }-I_{d}.
\end{equation*}%
Then using the definition of $\Lambda _{n}$ in \eqref{Lambda_n}, we have 
\begin{align}
\Lambda _{n}^{-1}\hat{V}\left( \Lambda _{n}^{-1}\right) ^{\prime }&
=n^{-1}\sum_{j=1}^{n}\Lambda _{n}^{-1}Z_{j}\left( \Lambda
_{n}^{-1}Z_{j}\right) ^{\prime }q(1-q)  \notag \\
& =\left[ I_{d}-\E\left( AA^{\prime }\right) h+O\left(h^{2}\right)\right]^{-1/2}\left( 
\frac{1}{n}\sum_{j=1}^{n}Z_{j}^{\ast }Z_{j}^{\ast \prime }\right) 
   \left[ I_{d}-E\left( AA^{\prime }\right) h+O\left(h^{2}\right)\right]^{-1/2}  \notag \\
& =\left[ I_{d}-\E\left( AA^{\prime }\right) h+O\left(h^{2}\right)\right]^{-1/2}\left(
I_{d}+D_{n}\right) \left[ I_{d}-\E\left( AA^{\prime }\right)
h+O\left(h^{2}\right)\right]^{-1/2}  \notag \\
& =\left[ I_{d}+(1/2)\E\left( AA^{\prime }\right) h+O\left(h^{2}\right)\right] \left(
I_{d}+D_{n}\right) \left[ I_{d}+(1/2)\E\left( AA^{\prime }\right) h+O\left(h^{2}\right)%
\right] .  \notag
\end{align}%
Let $\xi _{n}=\left( I_{d}+D_{n}\right) ^{-1}-\left( I_{d}-D_{n}\right)
=\left( I_{d}+D_{n}\right) ^{-1}D_{n}^{2}$, then%
\begin{align}
&\left[ \Lambda _{n}^{-1}\hat{V}\left( \Lambda _{n}^{-1}\right) ^{\prime }%
\right] ^{-1}  \notag \\
&\quad= \left[ I_{d}-\frac{1}{2}\E\left( AA^{\prime }\right) h+O\left(h^{2}\right)\right]
\left( I_{d}-D_{n}+\xi _{n}\right) \left[ I_{d}-\frac{1}{2}\E\left(
AA^{\prime }\right) h+O\left(h^{2}\right)\right]  \notag \\
&\quad= I_{d}-\E\left( AA^{\prime }\right) h+\eta _{n},
\label{proof_thm:inf1.5}
\end{align}%
where $\eta_n=-D_{n}+D_{n}O\left( h\right) +\xi _{n}+O(h^{2})+\xi _{n}O(h) $
collects the remainder terms. To evaluate the order of $\eta _{n}$, we start
by noting that $\E\left( \left \Vert D_{n}\right \Vert ^{2}\right) =O\left(
1/n\right) $. Let $\lambda _{\min }\left( \cdot \right) $ and $\lambda
_{\max }\left( \cdot \right) $ be the smallest and largest eigenvalues of a
matrix, then for any constant $C>2\sqrt{d}>0$: 
\begin{align*}
P & \left \{ \left \Vert \left( I_{d}+D_{n}\right) ^{-1}\right \Vert \geq
C\right \} \leq P\left \{ \lambda _{\max }\left[\left( I_{d}+D_{n}\right)
^{-1}\right]\geq C/\sqrt{d}\right \} \\
&= P\left \{ \lambda _{\min }\left( I_{d}+D_{n}\right) \leq \sqrt{d}/C\right
\} = P\left \{ 1+\lambda _{\min }\left( D_{n}\right) \leq \sqrt{d}/C\right \}
\\
&= P\left\{ \lambda _{\min }\left( D_{n}\right) \leq -1/2\right\} \leq P\left\{
\lambda _{\min }^{2}\left( D_{n}\right) >1/4\right\} \\
&\leq P\left( \left \Vert D_{n}\right \Vert ^{2}>1/4\right) =O\left( 1/n\right)
\end{align*}%
by the Markov inequality. Using this probability bound and the Chernoff
bound, we have for any $\epsilon >0$, 
\begin{align*}
P & \left \{ \frac{n}{\log n}\left \Vert \xi _{n}\right \Vert >\epsilon
\right \} \overset{}{\leq }P\left \{ \frac{n}{\log n}\left \Vert \left(
I_{d}+D_{n}\right) ^{-1}\right \Vert \times \left \Vert D_{n}\right \Vert
^{2}>\epsilon \right \} \\
&= P\left \{ n\left \Vert D_{n}\right \Vert ^{2}>\frac{\epsilon \log n}{C}%
\right \} +P\left \{ \left \Vert \left( I_{d}+D_{n}\right) ^{-1}\right \Vert
>C\right \} =O\left( 1/n \right) .
\end{align*}%
It then follows that 
\begin{equation*}
P\left \{ \left \Vert \eta _{n}\right \Vert \geq C\max \left( h^{2},\sqrt{%
\frac{\log n}{n}},h\sqrt{\frac{\log n}{n}},\frac{\log n}{n},\frac{h\log n}{n}%
\right) \right \} =O\left( \frac{1}{n}+h^{2}\right) .
\end{equation*}%
Under Assumption \ref{a:h}, we can rewrite the above as 
\begin{equation}
P\left \{ \left \Vert \eta _{n}\right \Vert \geq Ch^{2}/\log n\right \}
=O\left( h^{2}\right)  \label{proof_thm:inf3}
\end{equation}%
for any large enough constant $C>0$.

Using \eqref{proof_thm:inf1.5} and defining $W_{j}^{\ast }\equiv \Lambda
_{n}^{-1}Z_{j}\left[ G(-U_{j}/h)-q\right] $, we have 
\begin{align*}
S_{n}& =\left( \Lambda _{n}^{-1}m_{n}\right) ^{\prime }\Lambda _{n}^{\prime }%
\hat{V}^{-1}\Lambda _{n}\left( \Lambda _{n}^{-1}m_{n}\right) 
 =\left( \Lambda _{n}^{-1}m_{n}\right) ^{\prime }\left[ \Lambda _{n}^{-1}%
\hat{V}\left( \Lambda _{n}^{-1}\right) ^{\prime }\right] ^{-1}\left( \Lambda
_{n}^{-1}m_{n}\right) =S_{n}^{L}+e_{n}
\end{align*}%
where 
\begin{align*}
S_{n}^{L} &= \left( \sqrt{n}\bar{W}_{n}^{\ast }\right) ^{\prime }\left( 
\sqrt{n}\bar{W}_{n}^{\ast }\right) -h\left( \sqrt{n}\bar{W}_{n}^{\ast
}\right) ^{\prime }\E\left( AA^{\prime }\right) \left( \sqrt{n}\bar{W}%
_{n}^{\ast }\right) , \\
e_{n} &= \left( \sqrt{n}\bar{W}_{n}^{\ast }\right) ^{\prime }\eta _{n}\left( 
\sqrt{n}\bar{W}_{n}^{\ast }\right) ,
\end{align*}%
and $\bar{W}_{n}^{\ast }=n^{-1}\sum_{j=1}^{n}W_{j}^{\ast }$ as defined in (%
\ref{define_W_star}). Using the Chernoff bound on $\sqrt{n}\bar{W}_{n}^{\ast
}$ and the result in \eqref{proof_thm:inf3}, we can show that $P\left(
\left
\vert e_{n}\right \vert >Ch^{2}\right) =O(h^{2})$. This ensures that
we can ignore $e_{n}$ to the order of $O(h^{2})$ in approximating the
distribution of $S_{n}$.

The characteristic function of $S_{n}^{L}$ is 
\begin{align*}
\E\left[ \exp \left(itS_{n}^{L}\right)\right] & =C_{0}(t)-hC_{1}(t)+O\left(
h^{2}\right) \text{ where} \\
C_{0}(t)& \equiv \E\left \{ \exp \left[ it\left( \sqrt{n}\bar{W}_{n}^{\ast
}\right) ^{\prime }\left( \sqrt{n}\bar{W}_{n}^{\ast }\right) \right] \right
\} , \\
C_{1}(t)& \equiv \E\left \{ it\left( \sqrt{n}\bar{W}_{n}^{\ast }\right)
^{\prime }\left( \E AA^{\prime }\right) \left( \sqrt{n}\bar{W}_{n}^{\ast
}\right) \exp \left[ it\left( \sqrt{n}\bar{W}_{n}^{\ast }\right) ^{\prime
}\left( \sqrt{n}\bar{W}_{n}^{\ast }\right) \right] \right \} .
\end{align*}

Following \citet{Phillips1982} 
and using arguments similar to those in \citet{Horowitz1998} and %
\citet{Whang2006}, we can establish an expansion of the PDF of $%
n^{-1/2}\sum_{j=1}^{n}\left( W_{j}^{\ast }-\E W_{j}^{\ast }\right) $ of the
form 
\begin{equation*}
pdf(x)=(2\pi )^{-d/2}\exp (-x^{\prime }x/2)\left[1+n^{-1/2}p(x)\right]+O(n^{-1}),
\end{equation*}%
where $p(x)$ is an odd polynomial in the elements of $x$ of degree 3. When $%
d=1$, we know from \citet[\S2.8]{Hall1992} 
that 
\begin{equation*}
p(x)=-\frac{\kappa _{3}}{6}\frac{1}{\phi (x)}\frac{d}{dx}\phi
(x)(x^{2}-1)\quad \text{for}\quad \kappa _{3}=\frac{\E\left( W_{j}^{\ast
}-\E W_{j}^{\ast }\right) ^{3}}{V_{n}^{3/2}}=O(1).
\end{equation*}%
We use this expansion to compute the dominating terms in $C_{j}(t)$ for $%
j=0,1$.

First, 
\begin{align*}
C_{0}(t)& =\E\left \{ \exp \left[ it\left( \sqrt{n}\bar{W}_{n}^{\ast }\right)
^{\prime }\left( \sqrt{n}\bar{W}_{n}^{\ast }\right) \right] \right \} \\
& =\left( 2\pi \right) ^{-d/2}\int \exp \left \{ it\left[ x+\sqrt{n}%
\E(W_{j}^{\ast })\right] ^{\prime }\left[ x+\sqrt{n}\E(W_j^*)\right]
\right \} \exp \left( -x'x/2\right) dx+O\left(n^{-1}\right) \\
& \quad +\frac{1}{\sqrt{n}}\left( 2\pi \right) ^{-d/2}\int \exp \left \{ it%
\left[ x+\sqrt{n}\E(W_j^*)\right] ^{\prime }\left[ x+\sqrt{n}%
\E(W_j^*)\right] \right \} p\left( x\right) \exp \left( -x'x/2\right) dx \\
& =\left( 1-2it\right) ^{-d/2}\exp \left( \frac{i\left \Vert \sqrt{n}%
\E(W_j^*)\right \Vert ^{2}t}{1-2it}\right) +O\left(n^{-1}\right) \\
& \quad +\frac{1}{\sqrt{n}}\left( 2\pi \right) ^{-d/2}\int p\left( x\right)
\exp \left( -x'x/2\right) \left[ 1+it2x^{\prime }\sqrt{%
n}\E(W_j^*)+O\left(n\|\E W_j^*\|^{2}\right)\right] dx \\
& =\left( 1-2it\right) ^{-d/2}\exp \left( \frac{i\left \Vert \sqrt{n}%
\E(W_j^*)\right \Vert ^{2}t}{1-2it}\right) +O\left( \left \Vert
\E(W_j^*)\right \Vert +\sqrt{n}h^{2r}+n^{-1}\right) \\
& =\left( 1-2it\right) ^{-d/2}\exp \left( \frac{i\left \Vert \sqrt{n}%
\E(W_j^*)\right \Vert ^{2}t}{1-2it}\right) +O\left( h^{r}\right) ,
\end{align*}%
where the third equality follows from the characteristic function of a
noncentral chi-square distribution.

Second, for $C_{1}(t)$ we can put any $o(1)$ term into the remainder since $%
hC_{1}(t)$ will then have remainder $o(h)$. Noting that $x$ is an odd
function (of $x$) and so integrates to zero against any symmetric PDF, 
\begin{align*}
C_{1}\left( t\right) & = \E\left \{ it\left( \sqrt{n}\bar{W}_{n}^{\ast
}\right) ^{\prime }\E\left( AA^{\prime }\right) \left( \sqrt{n}\bar{W}%
_{n}^{\ast }\right) \exp \left[ it\left( \sqrt{n}\bar{W}_{n}^{\ast }\right)
^{\prime }\left( \sqrt{n}\bar{W}_{n}^{\ast }\right) \right] \right \} \\
& =\left( 2\pi \right) ^{-d/2}\int it\left( x+\sqrt{n}\E W_{j}^{\ast }\right)
^{\prime }\E\left( AA^{\prime }\right) \left( x+\sqrt{n}\E W_{j}^{\ast }\right)
\\
& \quad \times \exp\left\{ it\left[ x+\sqrt{n}\E(W_j^*)\right] ^{\prime }\left[
x+\sqrt{n}\E(W_j^*)\right]\right\} \exp \left( -x'x/2 \right)
dx \\
& \quad \times \left[ 1+O\left( n^{-1/2}\right) \right] \\
& =\left( 2\pi \right) ^{-d/2}\int itx' \E\left( AA^{\prime }\right)
x\exp \left[ -x'x\left( 1-2it\right)/2 \right] dx \\
&\quad+O\left(
\left \Vert \sqrt{n}\E(W_j^*)\right \Vert ^{2}\right) +O\left( \left
\Vert \E(W_j^*)\right \Vert \right) \\
& =\left( 1-2it\right) ^{-d/2}it\left\{ \mathrm{tr}\left[\E\left( AA^{\prime
}\right) \E\left(\mathbb{XX}^{\prime }\right)\right]\right\} +O\left( \left \Vert \sqrt{n}%
\E(W_j^*)\right \Vert ^{2}\right) +O\left( \left \Vert \E(W_j^*)\right \Vert \right) \\
& =\left( 1-2it\right) ^{-d/2-1}it\left\{ \mathrm{tr}\left[\E\left( AA^{\prime
}\right)\right] \right\} +O\left( \left \Vert \sqrt{n}\E(W_j^*)\right \Vert
^{2}\right) +O\left( \left \Vert \E(W_j^*)\right \Vert \right) ,
\end{align*}%
where $\mathbb{X}\sim N\left(0,diag\left( 1-2it\right) ^{-1}\right)$.

Combining the above steps, we have, for $r\geq 2$, 
\begin{align}
\E\left[ \exp\left(itS_{n}^{L}\right)\right] & =\overbrace{\left( 1-2it\right)
^{-d/2}\exp \left( \frac{i\left \Vert \sqrt{n}\E W_{j}^{\ast }\right \Vert ^{2}t%
}{1-2it}\right) }^{\text{$C_{0}(t)$}}-h\overbrace{(1-2it)^{-d/2-1}it\mathrm{%
tr}\left[ \E\left( AA^{\prime }\right) \right] }^{\text{$O(1)$ term in $%
C_{1}(t)$}}  \notag \\
& \quad +\overbrace{O\left( nh^{2r+1}\right) +O\left( h^{r+1}\right) }^{%
\text{remainder from $hC_{1}(t)$}}+O(h^{2})  \notag \\
& =\left( 1-2it\right) ^{-d/2}+(1-2it)^{-d/2-1}it\left \{ \left \Vert \sqrt{n}%
\E W_{j}^{\ast }\right \Vert ^{2}-h\mathrm{tr}\left[\E\left( AA^{\prime }\right)
\right] \right\}  \notag \\*
\label{eqn:SnL-char-fn}
& \quad +O\left(h^{2}+nh^{2r+1}\right).
\end{align}%
The $\chi _{d}^{2}$ characteristic function is $(1-2it)^{-d/2}$, and
integrating by parts yields the Fourier--Stieltjes transform of the $\chi
_{d}^{2}$ PDF: 
\begin{align*}
\int_{0}^{\infty }\exp (itx)d\mathcal{G}_{d}^{\prime }(x)& =\int_{0}^{\infty
}\exp (itx)\mathcal{G}_{d}^{\prime \prime }(x)dx=\left. \exp (itx)\mathcal{G}%
_{d}^{\prime }(x)\right \vert _{0}^{\infty }-\int_{0}^{\infty }(it)\exp (itx)%
\mathcal{G}_{d}^{\prime }(x)dx \\
& =(-it)(1-2it)^{-d/2}.
\end{align*}%
Taking a Fourier--Stieltjes inversion of \eqref{eqn:SnL-char-fn} thus yields 
\begin{align*}
P\left( S_{n}^{L}<x\right) & =\mathcal{G}_{d}\left( x\right) -\mathcal{G}%
_{d+2}^{\prime }(x)\left \{ \left \Vert \sqrt{n}\E W_{j}^{\ast }\right \Vert
^{2}-h\mathrm{tr}\left[\E\left( AA^{\prime }\right) \right]\right \} 
    +O\left(h^{2}+nh^{2r+1}\right)
\\
& =\mathcal{G}_{d}\left( x\right) -\mathcal{G}_{d+2}^{\prime }(x)\left \{
nh^{2r}\E(B)'\E(B)-h\mathrm{tr}\left[\E\left( AA^{\prime
}\right) \right]\right \} +O\left(h^{2}+nh^{2r+1}\right).
\end{align*}

A direct implication is that type I error is 
\begin{equation*}
P\left( m_{n}^{\prime }\hat{V}^{-1}m_{n}>c_{\alpha }\right) 
  = \alpha +%
\mathcal{G}_{d+2}^{\prime }(c_{\alpha }) \left \{
nh^{2r}\E(B)'\E(B)-h\mathrm{tr}\left[\E\left( AA^{\prime
}\right)\right] \right \} +O\left(h^{2}+nh^{2r+1}\right). \qed
\end{equation*}


\subsection*{Proof of Theorem \protect \ref{thm:power}}

Define 
\begin{equation*}
W_{j}\equiv W_{j}\left( \delta \right) \equiv Z_{j}\left[ G\left( \frac{%
X_{j}^{\prime }\delta }{\sqrt{n}h}-\frac{U_{j}}{h}\right) -q\right] ,
\end{equation*}%
then 
\begin{equation*}
m_{n}\left( \beta _{0}\right) =\frac{1}{\sqrt{n}}\sum_{j=1}^{n}W_{j}=\frac{1%
}{\sqrt{n}}\sum_{j=1}^{n}W_{j}\left( \delta \right) .
\end{equation*}%
We first compute the mean of $m_{n}\left( \beta _{0}\right)$. Let $%
[U_{L}(Z_{j},X_{j}),U_{H}\left( Z_{j},X_{j}\right) ]$ be the support of $%
U_{j}$ conditional on $Z_{j}$ and $X_{j}$. Using the same argument as in the
proof of Theorem \ref{thm:Wj}, 
\begin{align*}
\E\left[m_{n}\left( \beta _{0}\right)\right] 
  &= \sqrt{n}\E(W_j)
   = \sqrt{n}\E\left\{ Z_{j}%
\int_{U_{L}(Z_{j},X_{j})}^{U_{H}\left( Z_{j},X_{j}\right) }\left[ G\left( 
\frac{X_{j}^{\prime }\delta }{\sqrt{n}h}-\frac{u}{h}\right) -q\right]
dF_{U|Z,X}(u\mid Z_{j},X_{j}) \right\} \\
  &= \sqrt{n}\E\left\{ Z_{j}\left. \left[ G\left( \frac{X_{j}^{\prime }\delta }{\sqrt{n}%
h}-\frac{u}{h}\right) -q\right] F_{U|Z,X}(u\mid Z_{j},X_{j})\right \vert
_{U_{L}(Z_{j},X_{j})}^{U_{H}(Z_{j},X_{j})} \right\} \\
&\quad +\frac{\sqrt{n}}{h}\E\left\{ Z_{j}\int_{U_{L}(Z_{j},X_{j})}^{U_{H}\left(
Z_{j},X_{j}\right) }F_{U|Z,X}(u\mid Z_{j},X_{j})G^{\prime }\left( \frac{%
X_{j}^{\prime }\delta }{\sqrt{n}h}-\frac{u}{h}\right) du \right\} \\
  &= -\sqrt{n}\E\left(Z_j q\right)
     +\sqrt{n}\E\left[Z_{j}\int_{-1}^{1}F_{U|Z,X}\left( \frac{%
X_{j}^{\prime }\delta }{\sqrt{n}}-hv\mid Z_{j},X_{j}\right) G^{\prime
}\left( v\right) dv \right] \\
  &= \sqrt{n}\E\left\{Z_{j}\left[ F_{U|Z,X}\left( \frac{X_{j}^{\prime }\delta }{\sqrt{n%
}}\mid Z_{j},X_{j}\right) -q\right] \right\} \\
  &\quad +\sqrt{n}\E\left\{Z_{j}\int_{-1}^{1}\left[ f_{U|Z,X}^{(r-1)}\left( \frac{%
X_{j}^{\prime }\delta }{\sqrt{n}}\mid Z_{j},X_{j}\right) \frac{(-h)^{r}v^{r}%
}{r!}\right] G^{\prime }(v)dv \right\} +O\left(\sqrt{n}h^{r+1}\right).
\end{align*}%
Expanding $F_{U|Z,X}\left( \frac{X_{j}^{\prime }\delta }{\sqrt{n}}\mid
Z_{j},X_{j}\right) $ and $f_{U|Z,X}^{(r-1)}\left( \frac{X_{j}^{\prime
}\delta }{\sqrt{n}}\mid Z_{j},X_{j}\right) $ at zero, and since $r$ is even, 
\begin{align*}
\E\left[m_{n}\left( \beta _{0}\right)\right] 
  &= \sqrt{n}\E\left\{Z_{j}\left[ F_{U|Z,X}\left( 0\mid
Z_{j},X_{j}\right) -q\right]\right\}
    +\E\left]Z_{j}X_{j}^{\prime }\delta f_{U|Z,X}\left(
0\mid Z_{j},X_{j}\right)\right]
    +O\left( n^{-1/2}\right) \\
  &\quad +\frac{h^r}{r!}\sqrt{n}\E\left[ Z_{j}f_{U|Z,X}^{(r-1)}\left( 0\mid
Z_{j},X_{j}\right) \right] \left( \int_{-1}^{1}v^{r}G^{\prime }(v)dv\right)
+O\left(\sqrt{n}h^{r+1}+h^r\right) \\
  &= \E\left[ f_{U|Z,X}\left( 0\mid Z_{j},X_{j}\right) Z_{j}X_{j}^{\prime
}\delta \right] +\sqrt{n}h^rV^{1/2}\E(B)+O\left( n^{-1/2}+\sqrt{n}%
h^{r+1} +h^r\right) .
\end{align*}%
Here we have used the following extensions of the law of iterated
expectation: 
\begin{align*}
\E &\left \{Z_j\left[ F_{U|Z,X}\left( 0\mid Z_j,X_j\right) -q\right]\right \}
= \E\left \{Z_j\E\left[ \E\left(1\{U_j<0\} \mid Z_j,X_j\right) - q \mid Z_j%
\right]\right \} \\
&= \E\left \{Z_j \left[F_{U|Z}\left( 0\mid Z_j\right) -q\right]\right \} = 0 ,
\\
\E & \left[f_{U|Z,X}(u\mid Z,X)\mid Z=z\right] 
    = \int_\mathcal{X} f_{U|Z,X}(u\mid
z,x)f_{X|Z}(x\mid z)dx \\
  &= \int_\mathcal{X} \frac{f_{U,Z,X}(u,z,x)}{f_{Z,X}(z,x)} \frac{f_{Z,X}(z,x)%
}{f_Z(z)}dx = [f_Z(z)]^{-1} \int_\mathcal{X} f_{U,Z,X}(u,z,x) dx \\
  &= f_{U,Z}(u,z)/f_Z(z) = f_{U|Z}(u\mid z) , \\
\E & \left[f_{U|Z,X}(0\mid Z,X)g(Z)\right] 
    = \E\left \{ \E\left[ f_{U|Z,X}(0\mid Z,X) g(Z)
\mid Z\right] \right \} \\
  &= \E\left \{ \E\left[ f_{U|Z,X}(0\mid Z,X) \mid Z\right] g(Z) \right \} 
   = \E\left \{ f_{U|Z}(0\mid Z) g(Z) \right \} ,
\end{align*}
and similarly for derivatives of the PDF by exchanging the order of
differentiation and integration. 

Next, we compute the variance $V_{n}$ of $m_{n}\left( \beta _{0}\right)$: 
\begin{align*}
V_{n} &= \mathrm{Var}\left[ m_{n}\left( \beta _{0}\right) \right] = \mathrm{%
Var}\left[W_{j}\left( \delta \right) \right] \\
  &= \E\left\{ \left[ G\left( \frac{X_{j}^{\prime }\delta }{\sqrt{n}h}-\frac{U_{j}}{h}%
\right) -q\right] ^{2}Z_{j}Z_{j}^{\prime } \right\}
    -\left[ \E W_{j}\left( \delta
\right) \right] \left[ \E W_{j}\left( \delta \right) \right] ^{\prime } \\
  &= \E\left\{ \left[ G\left( \frac{X_{j}^{\prime }\delta }{\sqrt{n}h}-\frac{U_{j}}{h}%
\right) -q\right] ^{2}Z_{j}Z_{j}^{\prime }\right\} 
    +O\left(n^{-1}+h^{2r}\right) .
\end{align*}%
Now 
\begin{align*}
\E & \left \{ \left[ G\left( \frac{X_{j}^{\prime }\delta }{\sqrt{n}h}-\frac{%
U_{j}}{h}\right) -q\right] ^{2} \mid Z_{j},X_{j}\right \} \\
&= \int_{U_{L}(Z_{j},X_{j})}^{U_{H}\left( Z_{j},X_{j}\right) }\left[ G\left( 
\frac{X_{j}^{\prime }\delta }{\sqrt{n}h}-\frac{u}{h}\right) -q\right]
^{2}dF_{U|Z,X}\left( u\mid Z_{j},X_{j}\right) \\
&= \left. \left[ G\left( \frac{X_{j}^{\prime }\delta }{\sqrt{n}h}-\frac{u}{h}%
\right) -q\right] ^{2}F_{U|Z,X}\left( u\mid Z_{j},X_{j}\right) \right \vert
_{U_{L}(Z_{j},X_{j})}^{U_{H}\left( Z_{j},X_{j}\right) } \\
&\quad +\frac{2}{h}\int_{U_{L}(Z_{j},X_{j})}^{U_{H}\left( Z_{j},X_{j}\right)
}F_{U|Z,X}\left( u\mid Z_{j},X_{j}\right) \left[ G\left( \frac{X_{j}^{\prime
}\delta }{\sqrt{n}h}-\frac{u}{h}\right) -q\right] G^{\prime }\left( \frac{%
X_{j}^{\prime }\delta }{\sqrt{n}h}-\frac{u}{h}\right) du \\
&= q^{2}+2\int_{-1}^{1}F_{U|Z,X}\left( hv+\frac{X_{j}^{\prime }\delta }{%
\sqrt{n}}\mid Z_{j},X_{j}\right) \left[ G\left( -v\right) -q\right]
G^{\prime }\left( -v\right) dv \\
&= q^{2}+2F_{U|Z,X}\left( 0\mid Z_{j},X_{j}\right) \int_{-1}^{1}\left[
G\left( -v\right) -q\right] G^{\prime }\left( -v\right) dv \\
&\quad +2hf_{U|Z,X}\left( 0\mid Z_{j},X_{j}\right) \left[ \int_{-1}^{1}v%
\left[ G\left( -v\right) -q\right] G^{\prime }\left( -v\right) dv\right] \\
&\quad +\frac{2}{\sqrt{n}}\left[ f_{U|Z,X}\left( 0\mid Z_{j},X_{j}\right)
X_{j}^{\prime }\delta \int_{-1}^{1}\left[ G\left( -v\right) -q\right]
G^{\prime }\left( -v\right) dv\right] +O\left( h^{2}+n^{-1}\right) \\
&= q^{2}+F_{U|Z,X}\left( 0\mid Z_{j},X_{j}\right) \left( 1-2q\right)
-hf_{U|Z,X}\left( 0\mid Z_{j},X_{j}\right) \left(
1-\int_{-1}^{1}G^{2}(u)du\right) \\
&\quad +\frac{\left( 1-2q\right) }{\sqrt{n}}\left[ f_{U|Z,X}\left( 0\mid
Z_{j},X_{j}\right) X_{j}^{\prime }\delta \right] +O\left(
h^{2}+n^{-1}\right) ,
\end{align*}%
and so 
\begin{align*}
V_{n} 
  &= q^{2}\E\left(Z_{j}Z_{j}^{\prime }\right)
    +\left( 1-2q\right) \E\left[
F_{U|Z,X}\left( 0\mid Z_{j},X_{j}\right) Z_{j}Z_{j}^{\prime }\right] \\
&\quad +h\left( 1-\int_{-1}^{1}G^{2}(u)du\right) \E\left[f_{U|Z,X}\left(
0\mid Z_{j},X_{j}\right) Z_{j}Z_{j}^{\prime }\right] \\
&\quad +\frac{\left( 1-2q\right) }{\sqrt{n}}\E\left \{ \left[ f_{U|Z,X}\left(
0\mid Z_{j},X_{j}\right) X_{j}^{\prime }\delta \right] Z_{j}Z_{j}^{\prime
}\right \} +O\left(n^{-1}+h^{2}\right) \\
&= V-hV^{1/2} \E\left(AA^{\prime }\right) \left(V^{1/2}\right)^{\prime
}+O\left( n^{-1/2}+h^{2}\right) ,
\end{align*}%
where the last line holds because of the above law of iterated expectation
extension and 
\begin{align*}
q^{2} & \E\left(Z_{j}Z_{j}^{\prime }\right)+\left( 1-2q\right) \E\left[ F_{U|Z,X}\left(
0\mid Z_{j},X_{j}\right) Z_{j}Z_{j}^{\prime }\right] \\
&= q^{2}\E\left(Z_j Z_j'\right)+\left( 1-2q\right) \E\left \{ \E\left[ 1\left \{
U<0\right \} \mid Z_{j},X_{j}\right] Z_{j}Z_{j}^{\prime }\right \} \\
&= q^{2}\E\left(Z_j Z_j'\right)+\left( 1-2q\right) \E\left \{ \E\left[ 1\left \{
U<0\right \} Z_{j}Z_{j}^{\prime }\mid Z_{j},X_{j}\right] \right \} \\
&= q^{2}\E\left(Z_j Z_j'\right)+\left( 1-2q\right) \E\left( 1\left \{
U<0\right \} Z_{j}Z_{j}^{\prime }\right) =q(1-q)\E\left(Z_j Z_j'\right).
\end{align*}

Let $\Lambda _{n}=V_{n}^{1/2}$, then 
\begin{equation*}
\Lambda _{n}=V^{1/2}\left[ I_{d}-h\E\left( AA^{\prime }\right) +O\left(
n^{-1/2}+h^{2}\right) \right] ^{1/2}.
\end{equation*}%
Define $W_{j}^{\ast }\equiv W_{j}^{\ast }\left( \delta \right) =\Lambda
_{n}^{-1}W_{j}\left( \delta \right) $ and 
\begin{equation*}
\bar{W}_{n}^{\ast }\equiv \bar{W}_{n}^{\ast }\left( \delta \right)
=n^{-1}\sum_{j=1}^{n}W_{j}^{\ast }\left( \delta \right) .
\end{equation*}%
Then $\Delta =\sqrt{n}\E\left(W_j^*\right)$ and 
\begin{align*}
\left \Vert \Delta \right \Vert ^{2} &= \left \Vert V_{n}^{-1/2}\Sigma
_{ZX}\delta +V_{n}^{-1/2}\sqrt{n}\left( -h\right) ^{r}V^{1/2}\E(B)\right
\Vert ^{2} \\
&= \left \Vert V_{n}^{-1/2}V^{1/2}\tilde{\delta}+\sqrt{n}\left( -h\right)
^{r}V_{n}^{-1/2}V^{1/2}\E(B)\right \Vert ^{2} \\
&= \left \Vert \left[ I_{d}-h\E\left(AA'\right)\right]^{-1/2}\tilde{\delta}+%
\sqrt{n}\left( -h\right) ^{r}\E(B)\right \Vert ^{2}\left[ 1+o(1)\right] \\
&= \left \Vert \left[ I_{d}+\frac{1}{2}h\E\left(AA'\right)\right] \tilde{\delta}%
+\sqrt{n}\left( -h\right) ^{r}\E(B)\right \Vert ^{2}\left[ 1+o(1)\right] \\
&= \left\{ \left \Vert \tilde{\delta}\right \Vert ^{2}+h\tilde{\delta}%
^{\prime }\left[ \E\left(AA'\right)\right] \tilde{\delta}+nh^{2r}\E(B)'\E(B) +2\tilde{\delta}^{\prime }\sqrt{n}%
\left( -h\right) ^{r}\E(B)\right\} \left[ 1+o(1)\right] .
\end{align*}

We can now write 
\begin{equation*}
S_{n}=m_{n}\left( \beta _{0}\right) ^{\prime }\hat{V}^{-1}m_{n}\left( \beta
_{0}\right) =S_{n}^{L}+e_{n}
\end{equation*}%
where 
\begin{align*}
S_{n}^{L} &= \left( \sqrt{n}\bar{W}_{n}^{\ast }\right) ^{\prime }\left( 
\sqrt{n}\bar{W}_{n}^{\ast }\right) -h\left( \sqrt{n}\bar{W}_{n}^{\ast
}\right) ^{\prime }\E\left( AA^{\prime }\right) \left( \sqrt{n}\bar{W}%
_{n}^{\ast }\right) , \\
e_{n} &= \left( \sqrt{n}\bar{W}_{n}^{\ast }\right) ^{\prime }\eta _{n}\left( 
\sqrt{n}\bar{W}_{n}^{\ast }\right) .
\end{align*}%
By the same argument as in the proof of Theorem \ref{thm:inf}, we can show
that the presence of $e_{n}$ generates an approximation error that is not
larger than that given in Theorem \ref{thm:power}.

The characteristic function of $S_{n}^{L}$ is 
\begin{align*}
\E\left[ \exp\left(itS_{n}^{L}\right)\right] & =C_{0}(t)-hC_{1}(t)+O\left( h^{2}+%
n^{-1/2}\right) \text{ where} \\
C_{0}(t)& \equiv \E\left \{ \exp \left[ it\left( \sqrt{n}\bar{W}_{n}^{\ast
}\right) ^{\prime }\left( \sqrt{n}\bar{W}_{n}^{\ast }\right) \right]
\right \} , \\
C_{1}(t)& \equiv \E\left \{ it\left( \sqrt{n}\bar{W}_{n}^{\ast }\right)
^{\prime }\E\left( AA^{\prime }\right) \left( \sqrt{n}\bar{W}_{n}^{\ast
}\right) \exp \left[ it\left( \sqrt{n}\bar{W}_{n}^{\ast }\right) ^{\prime
}\left( \sqrt{n}\bar{W}_{n}^{\ast }\right) \right] \right \} .
\end{align*}%
Using the expansion of the PDF of $n^{-1/2}\sum_{j=1}^{n}\left(W_{j}^{\ast
}-\E W_{j}^{\ast }\right)$: 
\begin{equation*}
pdf(x)=(2\pi )^{-d/2}\exp\left(-x^{\prime }x/2\right)\left[1+n^{-1/2}p(x)\right]+O\left(n^{-1}\right),
\end{equation*}%
where $p(x)$ is an odd polynomial in the elements of $x$ of degree 3, we
obtain 
\begin{align*}
C_{0}(t)& =\E\left \{ \exp \left[ it\left( \sqrt{n}\bar{W}_{n}^{\ast }\right)
^{\prime }\left( \sqrt{n}\bar{W}_{n}^{\ast }\right) \right] \right \}  \\
& =\left( 2\pi \right) ^{-d/2}\int \exp \left \{ it\left[ x+\sqrt{n}%
\E\left(W_j^*\right)\right] ^{\prime }\left[ x+\sqrt{n}\E\left(W_j^*\right)\right]
\right \} \exp \left( -x'x/2\right) dx+O\left( n^{-1/2}\right)  \\
& =\left( 1-2it\right) ^{-d/2}\exp \left( \frac{it\left \Vert \sqrt{n}%
\E\left(W_j^*\right)\right \Vert ^{2}}{1-2it}\right) +O\left( n^{-1/2}\right) .
\end{align*}%
Similarly, 
\begin{align*}
C_{1}\left( t\right) & =\left( 2\pi \right) ^{-d/2}\int it\left( x+\sqrt{n}%
\E W_{j}^{\ast }\right) ^{\prime }\E\left( AA^{\prime }\right) \left( x+\sqrt{n}%
\E W_{j}^{\ast }\right)  \\
&\quad \times \exp \left \{ it\left[ x+\sqrt{n}\E\left(W_j^*\right)\right] ^{\prime }%
\left[ x+\sqrt{n}\E\left(W_j^*\right)\right] -x'x/2\right \}
dx+O\left( n^{-1/2}\right) .
\end{align*}%
Since 
\begin{align*}
it& \left[ x+\sqrt{n}\E\left(W_j^*\right)\right] ^{\prime }\left[ x+\sqrt{n}%
\E\left(W_j^*\right)\right] -x'x/2 \\
& =-\frac{1}{2}\left( 1-2it\right) \left[ x-\frac{2it}{1-2it}\sqrt{n}%
\E\left(W_j^*\right) \right] ^{\prime }\left[ x-\frac{2it}{1-2it} 
\sqrt{n}\E\left(W_j^*\right) \right]  \\
& \quad +\frac{it}{1-2it}\left( \sqrt{n}\E W_{j}^{\ast }\right) ^{\prime
}\left( \sqrt{n}\E W_{j}^{\ast }\right) ,
\end{align*}%
we have 
\begin{align*}
C_{1}\left( t\right) & =\left( 1-2it\right) ^{-d/2}\exp \left[ \frac{it}{%
1-2it}\left( \sqrt{n}\E W_{j}^{\ast }\right) ^{\prime }\left( \sqrt{n}%
\E W_{j}^{\ast }\right) \right]  \\
&\quad \times \E\left[it\left( \mathbb{X}+\sqrt{n}\E W_{j}^{\ast }\right) ^{\prime
}\E\left( AA^{\prime }\right) \left( \mathbb{X}+\sqrt{n}\E W_{j}^{\ast }\right)\right]
+O\left( n^{-1/2}\right)  \\
& = \left( 1-2it\right) ^{-d/2}\exp \left[ 
\frac{it}{1-2it}\left( \sqrt{n}\E W_{j}^{\ast }\right) ^{\prime }\left( \sqrt{n%
}\E W_{j}^{\ast }\right) \right]  \\
&\quad \times it\mathrm{tr}\left\{\E\left( AA^{\prime }\right) \left[ \frac{1}{1-2it}%
I_{d}+\left( \frac{2it}{1-2it}+1\right) ^{2}\left( \sqrt{n}\E W_{j}^{\ast
}\right) \left( \sqrt{n}\E W_{j}^{\ast }\right) ^{\prime }\right] \right\}
+O\left( n^{-1/2}\right)  \\
& = \left( 1-2it\right) ^{-d/2}\exp \left[ 
\frac{it}{1-2it}\left( \left \Vert \sqrt{n}\E W_{j}^{\ast }\right \Vert
^{2}\right) \right]  \\
&\quad \times \frac{it}{1-2it}\mathrm{tr}\left\{\E\left( AA^{\prime }\right) \left[ I_{d}+%
\frac{1}{1-2it}\left( \sqrt{n}\E W_{j}^{\ast }\right) \left( \sqrt{n}%
\E W_{j}^{\ast }\right) ^{\prime }\right]\right\} +O\left( n^{-1/2}\right) ,
\end{align*}%
where $\mathbb{X}\sim N\left[ \frac{2it}{1-2it}\left( \sqrt{n}\E W_{j}^{\ast
}\right) ,diag\left( 1-2it\right) ^{-1}\right] $. 

Combining the above steps, we have, for $r\geq 2$, 
\begin{align*}
& \E\left[ \exp\left(itS_{n}^{L}\right)\right] \\
& =\left( 1-2it\right) ^{-d/2}\exp \left( \frac{it\left \Vert \sqrt{n}%
\E W_{j}^{\ast }\right \Vert ^{2}}{1-2it}\right) \\
&\quad +(1-2it)^{-d/2-1}\exp \left( \frac{it\left \Vert \sqrt{n}\E W_{j}^{\ast
}\right \Vert ^{2}}{1-2it}\right) h\left( -it\right) \mathrm{tr}\left[\E\left(
AA^{\prime }\right)\right] +O\left(h^{2}+n^{-1/2}\right) \\
&\quad +(1-2it)^{-d/2-2}\exp \left( \frac{it\left \Vert \sqrt{n}\E W_{j}^{\ast
}\right \Vert ^{2}}{1-2it}\right) h(-it)\left( \sqrt{n}\E W_{j}^{\ast }\right)
^{\prime }\E\left( AA^{\prime }\right) \left( \sqrt{n}\E W_{j}^{\ast }\right) .
\end{align*}

Let $\mathcal{G}_{d}^{\prime }(x;\lambda )$ be the PDF of the noncentral
chi-square distribution with noncentrality parameter $\lambda $, so 
\begin{align}
\mathcal{G}_{d}^{\prime }(x;\lambda )& =\frac{1}{2\pi }\int_{\mathbb{R}%
}\left( 1-2it\right) ^{-d/2}\exp \left( \frac{it\lambda }{1-2it}\right) \exp
\left( -itx\right) dt,  \label{eqn:ncPDF} \\
\mathcal{G}_{d}^{\prime \prime }(x;\lambda )& =\frac{1}{2\pi }\int_{\mathbb{R%
}}\left( -it\right) \left( 1-2it\right) ^{-d/2}\exp \left( \frac{it\lambda }{%
1-2it}\right) \exp \left( -itx\right) dt.  \notag
\end{align}%
Using the above results and taking a Fourier--Stieltjes inversion, 
\begin{align*}
P_{\beta _{n}}\left( S_{n}<x\right) 
  &= \mathcal{G}_{d}\left( x;\left \Vert\Delta \right \Vert ^{2}\right) 
    +\mathcal{G}_{d+2}^{\prime }\left(x;\left \Vert\Delta \right \Vert ^{2}\right)
        h \mathrm{tr}\left[\E\left( AA^{\prime }\right) %
\right] \\
& \quad +\mathcal{G}_{d+4}^{\prime }\left(x;\left \Vert \Delta \right \Vert ^{2}\right)h%
\left[ \Delta ^{\prime }\E\left( AA^{\prime }\right) \Delta \right] +O\left(
h^{2}+n^{-1/2}\right) .
\end{align*}

Expanding $\mathcal{G}_{d}\left( x;\left \Vert \Delta \right \Vert ^{2}\right) 
$ around $\mathcal{G}_{d}\left( x;\Vert \tilde{\delta}\Vert ^{2}\right) $
yields 
\begin{align*}
\mathcal{G}_{d}\left( x;\left \Vert \Delta \right \Vert ^{2}\right) & =%
\mathcal{G}_{d}\left( x;\left \Vert \tilde{\delta}\right \Vert ^{2}\right)
+\left. \frac{\partial \mathcal{G}_{d}\left( x,\lambda \right) }{\partial
\lambda }\right \vert _{\lambda =\left \Vert \tilde{\delta}\right \Vert ^{2}} \\
& \quad \times \left[ h\tilde{\delta}^{\prime }\E\left( AA^{\prime }\right) 
\tilde{\delta}+nh^{2r}\E(B)'\E(B) +2\tilde{%
\delta}^{\prime }\sqrt{n}\left( -h\right) ^{r}\E(B)\right] \left[ 1+o(1)\right] 
\\
& =\mathcal{G}_{d}\left( x;\left \Vert \tilde{\delta}\right \Vert ^{2}\right) -%
\mathcal{G}_{d+2}^{\prime }\left(x;\left \Vert \tilde{\delta}\right \Vert ^{2}\right) \\
& \quad \times \left[ h\tilde{\delta}^{\prime } \E\left(AA^{\prime }\right) 
\tilde{\delta}+nh^{2r}\E(B)'\E(B) +2\tilde{%
\delta}^{\prime }\sqrt{n}\left( -h\right) ^{r}\E(B)\right] \left[ 1+o(1)\right] 
\end{align*}%
using the result that $\frac{\partial }{\partial \lambda }\mathcal{G}%
_{d}(x;\lambda )=-\mathcal{G}_{d+2}^{\prime }(x;\lambda )$, which can be
derived from \eqref{eqn:ncPDF}. Hence 
\begin{align*}
P_{\beta _{n}} \left( S_{n}<x\right)  
  &= \mathcal{G}_{d}\left( x;\left \Vert \tilde{\delta}\right \Vert ^{2}\right) -%
\mathcal{G}_{d+2}^{\prime }\left(x;\left \Vert \tilde{\delta}\right \Vert ^{2}\right)%
\left\{ nh^{2r}\E(B)'\E(B) -h\mathrm{tr}\left[%
\E\left( AA^{\prime }\right) \right]\right\}  \\
& \quad +\left[ \mathcal{G}_{d+4}^{\prime }\left(x;\left \Vert \tilde{\delta}%
\right \Vert ^{2}\right)-\mathcal{G}_{d+2}^{\prime }\left(x;\left \Vert \tilde{\delta}%
\right \Vert ^{2}\right)\right] h\left[ \tilde{\delta}^{\prime }\E\left( AA^{\prime
}\right) \tilde{\delta}\right]  \\
& \quad -\mathcal{G}_{d+2}^{\prime }\left(x;\left \Vert \tilde{\delta}\right \Vert
^{2}\right)2\tilde{\delta}^{\prime }\sqrt{n}\left( -h\right) ^{r}\E(B)+O\left(
h^{2}+n^{-1/2}\right) .
\end{align*}

Under the assumption that $\tilde{\delta}$ is uniform on the sphere $%
\mathcal{S}_{d}(\tau )$, we can write $\tilde{\delta}=\tau \xi /\left \Vert
\xi \right \Vert $ where $\xi \thicksim N(0,I_{d})$. Then%
\begin{align*}
\E_{\tilde{\delta}}& \left[P_{\beta _{n}}\left( S_{n}<x\right)\right] \\
& =\mathcal{G}_{d}\left( x;\tau ^{2}\right) -\mathcal{G}_{d+2}^{\prime
}\left(x;\tau ^{2}\right)\left\{ nh^{2r}\E(B)'\E(B) -h%
\mathrm{tr}\left[\E\left( AA^{\prime }\right) \right] \right\} \\
& \quad +\left[ \mathcal{G}_{d+4}^{\prime }\left(x;\tau ^{2}\right)-\mathcal{G}%
_{d+2}^{\prime }\left(x;\tau ^{2}\right)\right] \tau ^{2}h\mathrm{tr}\left[ \E\left(
AA^{\prime }\right) \E_{\xi }\left(\xi \xi ^{\prime }/\left \Vert \xi \right \Vert
^2\right)\right] +O\left( h^{2}+n^{-1/2}\right)
\end{align*}%
where $\E_{\xi }$ is the expectation with respect to $\xi $. As a
consequence, 
\begin{align*}
\E_{\tilde{\delta}} \left[P_{\beta _{n}}\left( S_{n}>x\right)\right] 
& =1-\mathcal{G}_{d}\left( x;\tau ^{2}\right) +\mathcal{G}_{d+2}^{\prime
}\left(x;\tau ^{2}\right)\left\{ nh^{2r}\E(B)'\E(B) -h%
\mathrm{tr}\left[\E\left( AA^{\prime }\right)\right] \right\} \\
& \quad -\left[ \mathcal{G}_{d+4}^{\prime }\left(x;\tau ^{2}\right)-\mathcal{G}%
_{d+2}^{\prime }\left(x;\tau ^{2}\right)\right] \frac{\tau ^{2}}{d}h\mathrm{tr}\left[
\E\left( AA^{\prime }\right) \right] +O\left( h^{2}+n^{-1/2}\right) .
\end{align*}%
Letting $x=c_{\alpha }$ yields the desired result. \qed

\subsection*{Proof of Corollary \protect \ref{cor:power}}

By direct calculations, 
\begin{align}
\E_{\tilde{\delta}}& \left[ P_{\beta _{n}}\left( S_{n}>c_{\alpha }^{\ast }\right) \right] \notag \\
& =1-\mathcal{G}_{d}\left( c_{\alpha }^{\ast };\tau ^{2}\right) +\mathcal{G}%
_{d+2}^{\prime }\left(c_{\alpha }^{\ast };\tau ^{2}\right)\left\{ nh^{2r}\E(B)'\E(B) -h\mathrm{tr}\left[\E\left( AA^{\prime
}\right) \right] \right\} \notag \\
& \quad -\left[ \mathcal{G}_{d+4}^{\prime }\left(c_{\alpha }^{\ast };\tau ^{2}\right)-%
\mathcal{G}_{d+2}^{\prime }\left(c_{\alpha }^{\ast };\tau ^{2}\right)\right] \frac{\tau
^{2}}{d}h\mathrm{tr}\left[ \E\left( AA^{\prime }\right) \right] +O\left(
h^{2}+n^{-1/2}\right)  \notag \\
& =1-\mathcal{G}_{d}\left( c_{\alpha };\tau ^{2}\right) +\mathcal{G}%
_{d}^{\prime }\left( c_{\alpha };\tau ^{2}\right) \frac{\mathcal{G}%
_{d+2}^{\prime }(c_{\alpha })}{\mathcal{G}_{d}^{\prime }\left( c_{\alpha
}\right) }\left( 1-\frac{1}{2r}\right) \mathrm{tr}\left[ \E\left(
AA^{\prime }\right) \right] h_{\text{SEE}}^{\ast }  \notag \\
& \quad +\mathcal{G}_{d+2}^{\prime }\left(c_{\alpha };\tau ^{2}\right)\left( \frac{1}{2r%
}-1\right) \mathrm{tr}\left[ \E\left( AA^{\prime }\right) \right] h_{%
\text{SEE}}^{\ast }  \notag \\
& \quad -\left[ \mathcal{G}_{d+4}^{\prime }\left(c_{\alpha };\tau ^{2}\right)-\mathcal{G%
}_{d+2}^{\prime }\left(c_{\alpha };\tau ^{2}\right)\right] \frac{\tau ^{2}}{d}h\mathrm{%
tr}\left[ \E\left( AA^{\prime }\right) \right] +O\left( h^{2}+n^{-1/2}\right)
\notag \\
& =1-\mathcal{G}_{d}\left( c_{\alpha };\tau ^{2}\right) +Q_{d}\left(
c_{\alpha },\tau ^{2},r\right) \mathrm{tr}\left[\E\left( AA^{\prime }\right)\right] h_{%
\text{SEE}}^{\ast }+O\left( h_{\text{SEE}}^{\ast 2}+n^{-1/2}\right) ,
\end{align}%
where%
\begin{align*}
Q_{d}\left( c_{\alpha },\tau ^{2},r\right) & =\left[ \mathcal{G}_{d}^{\prime
}\left( c_{\alpha };\tau ^{2}\right) \frac{\mathcal{G}_{d+2}^{\prime
}(c_{\alpha })}{\mathcal{G}_{d}^{\prime }\left( c_{\alpha }\right) }-%
\mathcal{G}_{d+2}^{\prime }\left(c_{\alpha };\tau ^{2}\right)\right] \left( 1-\frac{1}{%
2r}\right) \\
& \quad -\frac{1}{d}\left[ \mathcal{G}_{d+4}^{\prime }\left(c_{\alpha };\tau
^{2}\right)-\mathcal{G}_{d+2}^{\prime }\left(c_{\alpha };\tau ^{2}\right)\right] \tau ^{2}
\end{align*}%
as desired. 
\qed

\begin{lemma}
\label{lem:stochastic_expansion_beta} Let the assumptions in Theorem \ref%
{thm:est-MSE} hold. Then 
\begin{equation*}
\sqrt{n}(\hat{\beta}-\beta _{0})
  = -\left \{ \E\left[\frac{\partial }{\partial \beta
^{\prime }}\frac{1}{\sqrt{n}}m_{n}\left( \beta _{0}\right)\right] \right \}
^{-1}m_{n}+O_{p}\left( \frac{1}{\sqrt{nh}}\right) +O_{p}\left( \frac{1}{%
\sqrt{n}}\right) ,
\end{equation*}%
and 
\begin{equation*}
\E\left[\frac{\partial }{\partial \beta ^{\prime }}\frac{1}{\sqrt{n}}m_{n}\left(
\beta _{0}\right)\right] =\Sigma _{ZX}+O(h^{r}).
\end{equation*}
\end{lemma}

\begin{proof}
We first prove that $\hat{\beta}$ is consistent. Using the Markov
inequality, we can show that when $\E\left(\left \Vert Z_{j}\right \Vert
^{2}\right)<\infty $, 
\begin{equation*}
\frac{1}{\sqrt{n}}m_{n}\left( \beta \right) =\frac{1}{\sqrt{n}}\E\left[m_{n}\left(
\beta \right)\right] +o_{p}\left( 1\right)
\end{equation*}%
for each $\beta \in \mathcal{B}$. It is easy to show that the above $%
o_{p}\left( 1\right) $ term also holds uniformly over $\beta \in \mathcal{B}$%
. But 
\begin{align*}
& \lim_{h\rightarrow 0}\sup_{\beta \in \mathcal{B}}\left \Vert \frac{1}{%
\sqrt{n}}\E\left[m_{n}\left( \beta \right)\right] -\E\left[ Z\left( 1\{Y<X^{\prime }\beta
\}-q\right) \right] \right \Vert \\
&\quad= \lim_{h\rightarrow 0}\max_{\beta \in \mathcal{B}}\left \Vert \E\left\{Z\left[
G\left( \frac{X^{\prime }\beta -Y}{h}\right) -1\{Y<X^{\prime }\beta \} %
\right] \right\} \right \Vert \\
&\quad= \lim_{h\rightarrow 0}\left \Vert \E\left\{Z\left[ G\left( \frac{X^{\prime
}\beta ^{\ast }-Y}{h}\right) -1\{Y<X^{\prime }\beta ^{\ast }\} \right]
\right\} \right \Vert =0
\end{align*}%
by the dominated convergence theorem, where $\beta ^{\ast }$ is the value of 
$\beta $ that achieves the maximum. Hence 
\begin{equation*}
\frac{1}{\sqrt{n}}m_{n}\left( \beta \right) 
  = \E\left[ Z\left( 1\{Y<X^{\prime
}\beta \}-q\right) \right] +o_{p}\left( 1\right)
\end{equation*}%
uniformly over $\beta \in \mathcal{B}$. Given the uniform convergence and
the identification condition in Assumption \ref{a:beta}, we can invoke
Theorem 5.9 of \citet{vanderVaart1998} to obtain that $\hat{\beta}\stackrel{p}{\to}
\beta _{0}$.

Next we prove the first result of the lemma. Under Assumption \ref{a:G}%
(i--ii), we can use the elementwise mean value theorem to obtain%
\begin{equation*}
\sqrt{n}(\hat{\beta}-\beta _{0})=-\left[ \frac{\partial }{\partial \beta
^{\prime }}\frac{1}{\sqrt{n}}m_{n}\left( \tilde{\beta}\right) \right]
^{-1}m_{n}
\end{equation*}%
where 
\begin{equation*}
\frac{\partial }{\partial \beta ^{\prime }}m_{n}(\tilde{\beta})=\left[ \frac{%
\partial }{\partial \beta }m_{n,1}(\tilde{\beta}_{1}),\ldots ,\frac{\partial 
}{\partial \beta }m_{n,d}(\tilde{\beta}_{d})\right] ^{\prime }
\end{equation*}%
and each $\tilde{\beta}_{i}$ is a point between $\hat{\beta}$ and $\beta
_{0} $. Under Assumptions \ref{a:sampling} and \ref{a:G}(i--ii) and that $\E\left[%
\frac{\partial }{\partial \beta ^{\prime }}\frac{1}{\sqrt{n}}m_{n}\left(
\beta \right)\right]$ is continuous at $\beta =\beta _{0}$, we have, using
standard textbook arguments, that $\frac{\partial }{\partial \beta ^{\prime }%
}\frac{1}{\sqrt{n}}m_{n}\left( \tilde{\beta}\right) =\frac{\partial }{%
\partial \beta ^{\prime }}\frac{1}{\sqrt{n}}m_{n}\left( \beta _{0}\right)
+o_{p}\left( 1\right) $. But%
\begin{equation*}
\frac{\partial }{\partial \beta ^{\prime }}\frac{1}{\sqrt{n}}m_{n}\left(
\beta _{0}\right) =\frac{1}{nh}\sum_{j=1}^{n}Z_{j}X_{j}^{\prime }G^{\prime
}\left( -U_j/h\right) \overset{p}{\to} \Sigma _{ZX} .
\end{equation*}%
Hence, under the additional Assumption \ref{a:h} and nonsingularity of $%
\Sigma _{ZX}$, we have $\sqrt{n}(\hat{\beta}-\beta _{0})=O_{p}\left(
1\right) $. With this rate of convergence, we can focus on a $\sqrt{n}$
neighborhood $\mathcal{N}_{0}$ of $\beta _{0}$. We write 
\begin{equation*}
\sqrt{n}\left( \hat{\beta}-\beta _{0}\right) 
   = -\left( \frac{\partial }{%
\partial \beta ^{\prime }}\frac{1}{\sqrt{n}}m_{n}\left( \beta _{0}\right) +%
\left\{ \frac{\partial }{\partial \beta ^{\prime }}\frac{1}{\sqrt{n}}\left[
m_{n}\left( \tilde{\beta}\right) -m_{n}\left( \beta _{0}\right) \right] %
\right\} \right) ^{-1}m_{n}.
\end{equation*}%
Using standard arguments again, we can obtain the following stochastic
equicontinuity result: 
\begin{equation*}
\sup_{\beta \in \mathcal{N}_{0}}\left \Vert \left[ \frac{\partial }{\partial
\beta ^{\prime }}m_{n}\left( \beta \right) -\E\frac{\partial }{\partial \beta
^{\prime }}m_{n}\left( \beta \right) \right] -\left[ \frac{\partial }{%
\partial \beta ^{\prime }}m_{n}\left( \beta _{0}\right) -\E\frac{\partial }{%
\partial \beta ^{\prime }}m_{n}\left( \beta _{0}\right) \right] \right \Vert
=o_{p}\left( 1\right) ,
\end{equation*}%
which, combined with the continuity of $\E\frac{\partial }{\partial \beta
^{\prime }}m_{n}\left( \beta \right) $, implies that%
\begin{equation*}
\left\{ \frac{\partial }{\partial \beta ^{\prime }}\frac{1}{\sqrt{n}}\left[
m_{n}\left( \tilde{\beta}\right) -m_{n}\left( \beta _{0}\right) \right] %
\right\} =O_{p}\left( n^{-1/2}\right) .
\end{equation*}%
Therefore 
\begin{align*}
\sqrt{n}\left( \hat{\beta}-\beta _{0}\right) 
  &= -\left[ \frac{\partial }{%
\partial \beta ^{\prime }}\frac{1}{\sqrt{n}}m_{n}\left( \beta _{0}\right)
+O_{p}\left( n^{-1/2}\right) \right]^{-1}m_{n} \\
&= -\left( \frac{\partial }{\partial \beta ^{\prime }}\frac{1}{\sqrt{n}}%
m_{n}\right)^{-1}m_{n}+O_{p}\left( n^{-1/2}\right) .
\end{align*}

Now 
\begin{align*}
& \mathrm{Var}\left( \text{vec}\left[ \frac{\partial }{\partial \beta
^{\prime }}m_{n}/\sqrt{n}\right] \right)  \\
& =n^{-1}\mathrm{Var}\left[ \text{vec}\left( Z_{j}X_{j}^{\prime }\right)
h^{-1}G^{\prime }(-U_{j}/h)\right]  \\
& \leq n^{-1}\E\left\{ \text{vec}\left( Z_{j}X_{j}^{\prime }\right) \left[ 
\text{vec}\left( Z_{j}X_{j}^{\prime }\right) \right] ^{\prime }h^{-2}\left[
G^{\prime }(-U_{j}/h)\right] ^{2}\right\}  \\
& =n^{-1}\E\left \{ \text{vec}\left( Z_{j}X_{j}^{\prime }\right) \left[ \text{%
vec}\left( Z_{j}X_{j}^{\prime }\right) \right] ^{\prime }\int h^{-2}\left[
G^{\prime }(-u/h)\right] ^{2}f_{U|Z,X}(u\mid Z_{j},X_{j})du\right \}  \\
& =(nh)^{-1}\E\left \{ \text{vec}\left( Z_{j}X_{j}^{\prime }\right) \left[ 
\text{vec}\left( Z_{j}X_{j}^{\prime }\right) \right] ^{\prime }\int \left[
G^{\prime }(v)\right] ^{2}f_{U|Z,X}(-hv\mid Z_{j},X_{j})dv\right \}  \\
& =O\left( \frac{1}{nh}\right) ,
\end{align*}%
so 
\begin{equation*}
\frac{\partial }{\partial \beta ^{\prime }}\frac{1}{\sqrt{n}}m_{n}
  = \E\left(\frac{%
\partial }{\partial \beta ^{\prime }}\frac{1}{\sqrt{n}}m_{n} \right)
+O_{p}\left( \frac{1}{\sqrt{nh}}\right) .
\end{equation*}%
As a result,%
\begin{align*}
\sqrt{n}\left( \hat{\beta}-\beta _{0}\right) 
  &= -\left[ \E\left(\frac{\partial }{%
\partial \beta ^{\prime }}\frac{1}{\sqrt{n}}m_{n}\right)+O_{p}\left( \frac{1}{\sqrt{%
nh}}\right) \right]^{-1}m_{n}+O_{p}\left( n^{-1/2}\right)  \\
& =-\left[ \E\left(\frac{\partial }{\partial \beta ^{\prime }}\frac{1}{\sqrt{n}}%
m_{n}\right)\right]^{-1}m_{n}+O_{p}\left( \frac{1}{\sqrt{nh}}\right) +O_{p}\left( 
n^{-1/2}\right) .
\end{align*}

For the second result of the lemma, we use the same technique as in the
proof of Theorem \ref{thm:Wj}. We have 
\begin{align*}
\E\left( \frac{\partial }{\partial \beta ^{\prime }}m_{n}/\sqrt{n}\right) 
  &= \E\left[ \frac{1}{nh}\sum_{j=1}^{n}Z_{j}X_{j}^{\prime }G^{\prime }(-U_{j}/h)%
\right] 
   = \E\left\{ \E\left[ Z_{j}X_{j}^{\prime }h^{-1}G^{\prime
}(-U_{j}/h)\mid Z_{j},X_{j}\right] \right\} \\
  &= \E\left[ Z_{j}X_{j}^{\prime }\int G^{\prime }(-u/h)f_{U|Z,X}(u\mid
Z_{j},X_{j})d(u/h)\right] \\
& =\E\left[ Z_{j}X_{j}^{\prime }\int G^{\prime }(v)f_{U|Z,X}(-hv\mid
Z_{j},X_{j})dv\right] \\
& =\E\left[ Z_{j}X_{j}^{\prime }f_{U|Z,X}(0\mid Z_{j},X_{j})\right] +O(h^{r}),
\end{align*}%
as desired. 
\end{proof}

\end{document}